\documentclass[12pt, fleqn, final]{article}
\usepackage[]{structure}

\usepackage{tikz}
\usetikzlibrary{arrows}
\usepackage{tikz-cd}
\tikzset{/tikz/commutative diagrams/background color=bg}
\usepackage{frcursive}

\newcommand{\nord}[1]{\mathbf{:} #1 \mathbf{:}}
 
\DeclareFontEncoding{LS1}{}{}
\DeclareFontSubstitution{LS1}{stix}{m}{n}
\DeclareSymbolFont{symbols2}{LS1}{stixfrak}{m}{n}
\DeclareMathSymbol{\typecolon}{\mathbin}{symbols2}{"25}

\usepackage{relsize}

\title{Chirality in 2d pAQFT}
\author{
    \small Sam Crawford \\
    \small\textit{University of York,}  \\
	\small\textit{Department of Mathematics,}\\
    \small\href{mailto:sam.crawford@york.ac.uk}{\texttt{sam.crawford@york.ac.uk}}
        \and
    \small Kasia Rejzner  \\
    \small\textit{University of York,}  \\
	\small\textit{Department of Mathematics,}\\
    \small\href{mailto:kasia.rejzner@york.ac.uk}{\texttt{kasia.rejzner@york.ac.uk}}
        \and
    \small Beno\^{\i}t Vicedo \\
    \small \textit{University of York,}  \\
	\small \textit{Department of Mathematics,}\\
    \small \href{mailto:benoit.vicedo@york.ac.uk}{\texttt{benoit.vicedo@york.ac.uk}}
}

\date{\today}

\begin{document}

\maketitle


\definecolor{todoKasia}{RGB}{207, 58, 174}
\definecolor{todoBenoit}{RGB}{87, 159, 222}
\newcommand{\commentKasia}[1]{\todo[color=todoKasia]{#1}}
\newcommand{\commentBenoit}[1]{\todo[color=todoBenoit]{#1}}



\ifdraft{
\noindent
{\Large\bfseries Key:}
\begin{itemize}
    \item[] \hspace{-2.7em} \tikz[baseline=0.5ex]{\path[draw=black,fill=todogreen ] (0, 0) rectangle (0.4cm, 0.4cm);} :
                Boring things for Sam
    \item[] \hspace{-2.7em} \tikz[baseline=0.5ex]{\path[draw=black,fill=todoKasia ] (0, 0) rectangle (0.4cm, 0.4cm);} :
                Comments by Kasia
    \item[] \hspace{-2.7em} \tikz[baseline=0.5ex]{\path[draw=black,fill=todoBenoit] (0, 0) rectangle (0.4cm, 0.4cm);} :
                Comments by Benoit
    \item[] \hspace{-2.7em} \tikz[baseline=0.5ex]{\path[draw=black,fill=todoorange] (0, 0) rectangle (0.4cm, 0.4cm);} :
                Format/layout queries
    \item[] \hspace{-2.7em} \tikz[baseline=0.5ex]{\path[draw=black,fill=todored   ] (0, 0) rectangle (0.4cm, 0.4cm);} :
                Mathematical queries/concerns
\end{itemize}
Comments can be made using the macro
\verb| \commentKasia{Your comment here} | or \\
\verb| \commentBenoit{Your comment here}|
}{}


\listoftodos

\tableofcontents

\begin{abstract}
In this article, which builds upon the work done in \cite{crawfordLorentzian2dCFT2021},
the chiral aspects of 2\textsc{dcft} on globally hyperbolic Lorentzian manifolds
are developed and explored within the perturbative algebraic quantum field theory ({\smaller{p}}\textsc{aqft}) framework.
In the example of the massless scalar field on globally hyperbolic 2-dimensional spacetimes, we identify the subalgebras of a given theory comprising only chiral (or anti-chiral) observables. These subalgebras are constructed explicitly, with the help of structures naturally associated to a Cauchy surface, for both the classical and the quantised theory,
and it is shown that they then embed naturally into the algebra of the full theory. 
Finally, it is demonstrated that the construction of these subalgebras is independent of the choice
of Cauchy surface and that they unambiguously define a covariant theory on the spaces of null-geodesics.
\end{abstract}

\section{Introduction}

An almost universal assumption in the study of conformal field theory
in 2 dimensions is that chiral fields `live' on Riemann surfaces.
This is a powerful assumption, as it enables one to use the tools
of complex analysis to tackle certain problems.
For instance, the tricky functional analysis of distributions
is replaced by either complex analysis or the algebraic manipulation of Laurent distributions.

The use of Riemann surfaces in chiral field theory
is typically justified in one of two ways:
One can perform a Wick rotation from
2\textsc{d} Minkowski space to Euclidean space.
By introducing the complex variable $z = x + it$,
this space is then identified with the complex plane.
Conformal symmetry of the original theory
suggests that only the \textit{conformal class} of the Riemannian metric
$dz d\overline{z}$ matters.
Thus, a more general theory might be defined on an arbitrary Riemann surface.
The fact that this Wick rotated theory still describes what we began with
is then a consequence of the Osterwalder-Schrader theorem \cite{osterwalderAxiomsEuclideanGreen1973},
which gives sufficient conditions for the $n$-point correlator functions of
the Wick rotated theory to be analytically continued back to the original spacetime.

An alternative justification,
found, for example in~\cite[Chapter~1]{kacVertexAlgebrasBeginners1998},
begins by describing a spacetime event using the null coordinates
$(u, v) = (t - x, t + x)$.
One may complexify each of these coordinates independently.
This is possible in a conformally symmetric Wightman \textsc{qft} because
the null momentum operators $P_u, P_v$ have a jointly positive spectrum,
allowing the translation operator
$e^{i(q P_u + r P_v)}$ to be defined for
values of $q$ and $r$ lying in the upper complex half-plane
$\mathbb{H} \subset \mathbb{C}$.
This leads to an embedding of $\mathbb{M}^2$ as the boundary of
$\mathbb{H}^2 \subset \mathbb{C}^2$.
We then identify the chiral/anti-chiral fields as those which depend on
only one copy of $\mathbb{H}$.
By restricting attention to one or the other,
the result is a conformally symmetric theory defined over $\mathbb{H}$.

However, neither of the approaches described above can be applied easily to Lorentzian manifolds more general than Minkowski spacetime.
On the other hand, methods for constructing $*$-algebras of quantum observables on a large class of spacetimes
(namely the \textit{globally hyperbolic} ones%
\footnote{
  Even this requirement may be relaxed in certain circumstances,
  e.g.\ if the spacetime is \emph{locally} globally hyperbolic in a particular sense
  \cite{kayPrincipleLocalityQuantum1992}.}%
)
exist and such theories are now very-well studied.
In particular, there is the locally-covariant approach of  \textsc{aqft},
as proposed in \cite{brunettiGenerallyCovariantLocality2003}.
A crucial feature of \textsc{aqft} is that the algebra of observables
is constructed \textit{before} a Hilbert space of states.
This is preferable because,
distinguished 'vacuum-like' states on generic globally hyperbolic spacetimes
cannot be obtained in a systematic manner
\cite[Theorem 5.4]{fewsterLocallyCovariantQuantum2015}.

This is the approach that we follow in this paper,
while we also try to stay close to the more traditional methods of studying 2\textsc{dcft}.
We give a characterisation of chirality
on any two-dimensional globally hyperbolic Lorentzian manifold $\mathcal{M}$.
To begin, we consider why the two 
ways of introducing Riemann surfaces into the study of  2\textsc{dcft} described 
above are ill-suited
for this generalisation.

The idea of Wick rotating the time coordinate seems appealing,
however it is well known that it does not generalize to curved spacetimes (see e.g. \cite[Chapter 5]{brunettiQuantumFieldTheory2009} for review)
An attempt to analytically continue the translation operators also suffers from
the lack of a preferred state on curved spacetimes.
In fact, by employing the joint spectrum of momentum operators,
this approach assumes not only the existence of a preferred state,
but in particular the existence of a \textit{highly symmetric} preferred state.

However, there is one aspect of this second approach that
\textit{can} be generalised to an arbitrary spacetime:
as we shall see,
it is possible to identify a chiral/anti-chiral sector
of a conformally covariant \textsc{qft} which is defined on a one-dimensional manifold,
without the use of a preferred system of coordinate.

\section{Chiral Configuration Spaces}

The primary example of a chiral field arises from the theory of
a massless scalar field on 2\textsc{d} Minkowski spacetime.
Recall that, in null coordinates,
the equation of motion of this theory is simply
\begin{equation}
  \label{eq:eom-null}
    \partial_u \partial_v \phi =: P \phi = 0.
\end{equation}
The solution of this equation is famously the sum of two terms:
one which depends only on the $u$ coordinate and one solely dependent on $v$.
To understand the geometry of the spaces the decoupled coordinates describe,
consider the Einstein cylinder $\mathscr{E}$,
which we recall is defined as the quotient of $\mathbb{M}^2$
under the equivalence relation
\begin{equation}
    \label{eq:cylinder-def}
    (t, x) \sim (t, x + 2\pi).
\end{equation}
In null coordinates, all functions on the cylinder satisfy the
periodicity condition $\phi(u - 2\pi, v + 2\pi) = \phi(u, v)$.
We can write a general solution of \eqref{eq:eom-null} as
\begin{align}
  \label{eq:cylinder-gen-soln}
  \phi(u, v) = \phi_\ell(u) + \phi_r(v) + \frac{p}{2\pi}(u + v),
\end{align}
where $p \in \mathbb{R}$ is a constant and each term $\phi_{\ell/r}$ must be $2\pi$ periodic.
In other words, a general solution to the wave equation on the cylinder
is determined by a pair of functions
$\phi_{\ell / r} \in \mathfrak{E}(S^1)$ and the constant $p$.

We can describe this phenomenon without coordinates in the following manner:
consider the map $\pi_\ell$ which sends an event $(u, v) \in \mathbb{M}^2$
to the line $\left\{ (u, v') \in \mathbb{M}^2 \right\}_{v' \in \mathbb{R}}$.
Each such line is determined uniquely by the choice of $u$,
and the space $\mathbb{M}_\ell$ of all such lines
is then diffeomorphic to $\mathbb{R}$.
Defining the analogous map on the cylinder,
we again have that $\pi_\ell [u, v] = \pi_\ell [u, v']$ for all $v, v'$,
but we have a further identification that
$\pi_\ell [u, v] = \pi_\ell [u + 2 \pi, v]$,
which tells us that the space $\mathscr{E}_\ell$ of all such lines
is actually one-to-one with $\mathbb{R} / 2 \pi \mathbb{Z} \simeq S^1$.

Thus we can describe solutions to the wave equation on
either Minkowski space or the cylinder in the same language:
for $\mathcal{M} \in \left\{ \mathbb{M}^2, \mathscr{E} \right\}$,
let $\mathcal{M}_{\ell / r}$ denote respectively the spaces of
right / left moving null rays, for instance
\begin{align}
  \mathbb{M}^2_{\ell} \ni \gamma = \left\{ (u_0, v) \in \mathbb{M}^2 \,|\, v \in \mathbb{R} \right\}
\end{align}
for some fixed $u \in \mathbb{R}$.
We also denote by $\pi_{\ell / r} : \mathcal{M} \to \mathcal{M}_{\ell / r}$
the obvious surjections onto these spaces.
We can then write a class of solutions to the wave equation on $\mathcal{M}$ as
$\phi = \pi_\ell^* \phi_\ell + \pi_r^* \phi_r$.
\footnote{
  Note that for $\mathcal{M} = \mathscr{E}$,
  this method only generates the solutions for $p = 0$.
  As we are only hoping to find subalgebras of observables,
  it is not necessary to generate \emph{all} the solutions to the equation of motion.
  However, the missing solutions would have to be found if one wished to
  reconstruct the full spacetimes' observables from the chiral subalgebras.
}

We can define these spaces also on arbitrary spacetimes.
We follow \cite{fewsterLocallyCovariantQuantum2015},
and define a \emph{spacetime} as a tuple $\mathcal{M} = (M, g, \mathfrak{o}, \mathfrak{t})$
comprising, in order,
a smooth manifold $M$ of a fixed dimension (in our case $2$),
a Lorentzian metric $g$ on $M$,
an orientation $\mathfrak{o}$
(i.e. an equivalence class of nowhere-vanishing $2$-forms)
and a time-orientation $\mathfrak{t}$
(which may be defined as
an equivalence class of time-like vector fields which are all future-pointing).

\begin{definition}
  Consider an inextensible null geodesic $\gamma: \mathbb{R} \to \mathcal{M}$
  which is everywhere future-directed according to $\mathfrak{t}$.
  We can call $\gamma$ \emph{left-moving} if,
  for $o \in \mathfrak{o}$ and $t \in \mathfrak{t}$,
  $o(t \otimes \dot{\gamma}) < 0$,
  otherwise we call $\gamma$ \emph{right-moving}.
\end{definition}

Note, we never have $o(t \otimes \dot{\gamma}) = 0$,
because $\gamma$ being null implies $\dot{\gamma}$ is never colinear with $t$,
and $o$ is nowhere-vanishing, and hence non-degenerate.

If we then identify reparametrisations
(i.e. we consider only the image of a geodesic),
we can define a set $\mathcal{M}_\ell$ of \textit{right}-moving,
inextensible null geodesics and similarly $\mathcal{M}_r$,
the space of \textit{left}-moving geodesics.
Note the apparent mismatch is so that elements of $\mathcal{M}_\ell$
are identified by `left-moving' coordinates on $\mathcal{M}$ and
\textit{vice-versa}.
For $\mathcal{M} \in \{\mathbb{M}^2, \mathscr{E}\}$,
these spaces are clearly analogous to the spaces of null rays introduced above.
By a slight abuse of the terminology, we shall refer to elements of
$\mathcal{M}_{\ell / r}$ as `geodesics'.
The maps $\pi_{\ell/r}$ from before generalise to this setting by defining
$\pi_{\ell/r}(x) \in \mathcal{M}_{\ell/r}$ as the equivalence class of the
inextensible right/left moving null geodesic $\gamma$ such that $\gamma(0) = x$.

As we are studying conformal field theories,
it is necessary to know how these spaces behave under
the appropriate spacetime morphisms.
As in \cite{crawfordLorentzian2dCFT2021},
we shall be working in the category $\mathsf{CLoc}$,
whose objects are 2\textsc{d}, globally hyperbolic spacetimes
and whose morphisms are \emph{conformally admissible embeddings},
i.e. smooth embeddings of manifolds preserving orientation, time-orientations,
and the conformal class of the metric.
A useful characterisation of conformal embeddings is that they preserve null geodesics.
The admissibility (i.e. orientation-preserving) property further indicates that
left-moving geodesics are mapped to left-moving geodesics and right to right.
More precisely, we have the following proposition:

\begin{proposition}
  \label{prop:null-projection-functorial}
    Let $\chi: \mathcal{M} \to \widetilde{\mathcal{M}}$
    be a conformally admissible embedding,
    then there exist natural maps
    $
        \chi_{\ell / r}:
        \mathcal{M}_{\ell / r} \to
        \widetilde{\mathcal{M}}_{\ell / r}
    $
    such that
    \begin{align}
      \label{eq:null-geodesic-naturality}
        \widetilde{\pi}_{\ell / r} \circ \chi
            =
        \chi_{\ell / r} \circ \pi_{\ell / r}.
    \end{align}
    Hence the maps $\pi_{\ell / r}$ define a functor $\mathsf{CLoc} \rightarrow \mathsf{Man}_+^1$,
    where $\mathsf{Man_+^1}$ is the category of smooth, oriented $1$-manifolds
    with smooth, oriented embeddings as morphisms.
\end{proposition}

\begin{proof}
  From the definition of $\pi_{\ell / r}$, we have that $\pi_{\ell / r}(x) = \pi_{\ell / r}(x')$
  if and only if there is a right/left-moving null geodesic $\gamma$ connecting $x$ with $x'$.
  As $\chi$ is conformally admissible, $\chi \circ \gamma$ will also be a right/left-moving null geodesic
  connecting $\chi(x)$ with $\chi(x')$ hence $\widetilde{\pi}_{\ell / r} \circ \chi (x) = \widetilde{\pi}_{\ell / r} \circ \chi (x')$.
  This means we can define $\chi_{\ell / r}(\gamma)$ to be
  $\pi_{\ell / r} \circ \chi (x)$ for any $x$ such that $\pi_{\ell / r}(x) = \gamma$,
  as this definition does not depend on our choice of $x$.
  Clearly this definition makes $\chi_{\ell / r}$ injective,
  and satisfies \eqref{eq:null-geodesic-naturality}.

  Statements about smoothness and orientability can be verified once we have defined
  the smooth structure and orientation on $\mathcal{M}_{\ell / r}$ in the following discussion.
  For now, we shall simply comment that, given a sequence
  of conformally admissible embeddings
  $\mathcal{M}_1 \overset{\chi_1}{\rightarrow} \mathcal{M}_2 \overset{\chi_2}{\rightarrow} \mathcal{M}_{3}$,
  if we denote the corresponding maps to null geodesics by
  $\pi_{i, \ell/r}: \mathcal{M}_i \rightarrow (\mathcal{M}_i)_{\ell/r}$,
  and the maps we have just defined by
  $\chi_{i, \ell/r}: (\mathcal{M}_i)_{\ell / r} \rightarrow (\mathcal{M}_{i+1})_{\ell / r}$,
  then 
  the following diagram commutes.

  \begin{equation*}
    \begin{tikzcd}
      \mathcal{M}_1
        \ar[r, "\chi_1"] \ar[d, two heads, "\pi_{1, \ell / r}"]
        &
      \mathcal{M}_2
        \ar[r, "\chi_2"] \ar[d, two heads, "\pi_{2, \ell / r}"]
        &
      \mathcal{M}_3
        \ar[d, two heads, "\pi_{3, \ell / r}"]
        \\
      (\mathcal{M}_1)_{\ell/r}
        \ar[r, "\chi_{1, \ell/r}"]
        &
      (\mathcal{M}_2)_{\ell/r}
        \ar[r, "\chi_{2, \ell/r}"]
        &
      (\mathcal{M}_3)_{\ell/r}
    \end{tikzcd}
  \end{equation*}
From which we see that
  \begin{align*}
    \pi_{3, \ell/r} \circ (\chi_2 \circ \chi_1)
    =:
    (\chi_2 \circ \chi_1)_{\ell/r} \circ \pi_{1, \ell/r}.
    =
    (\chi_{2, \ell/r} \circ \chi_{1, \ell/r}) \circ \pi_{1, \ell/r}
  \end{align*}
  Hence $(\chi_2 \circ \chi_1)_{\ell/r} = \chi_{2, \ell/r} \circ \chi_{1, \ell/r}$.
\end{proof}

The spaces $\mathcal{M}_{\ell / r}$ are
somewhat awkward to work with directly.
Instead it is easier to note that,
given a Cauchy surface $\Sigma \subset \mathcal{M}$,
the restriction $\pi_{\ell / r}|_\Sigma$ becomes a bijection.
Because the elements of $\mathcal{M}_{\ell / r}$ are in particular inextensible causal curves
this is true even on an arbitrary globally hyperbolic spacetime%
\footnote{
  We assume throughout that Cauchy surfaces are smooth and spacelike,
  hence the intersection of any causal curve with any Cauchy surface is necessarily transversal.
  The existence of such surfaces on globally hyperbolic spacetimes is
  \cite[Theorem 1]{bernalSmoothCauchyHypersurfaces2003}.
}.

These maps endow $\mathcal{M}_{\ell / r}$ with differentiable structures
as well as orientations, independent of the choice of $\Sigma$.
This is well-defined, as any pair of Cauchy surfaces $\Sigma, \Sigma' \subset \mathcal{M}$
of the same spacetime are diffeomorphic to one another,
as a consequence of \cite[Theorem 1.1]{bernalSmoothCauchyHypersurfaces2003}.
In particular, we shall actually define the orientation of $\mathcal{M}_{\ell}$
such that each diffeomorphism $\pi_{\ell}|_{\Sigma}: \Sigma \overset{\simeq}{\rightarrow} \mathcal{M}_{\ell}$
\textit{reverses} orientation, this is motivated by the example of the
$t = 0$ Cauchy surface $\Sigma_0 \subset \mathbb{M}^2$, where the spatial coordinate $x \in \Sigma_0$
corresponds to the $u = -x$ null ray in $\mathbb{M}^2_{\ell}$.
Given $\chi: \mathcal{M} \rightarrow \widetilde{\mathcal{M}}$ from before we have that,
for any Cauchy surface $\Sigma \subset \mathcal{M}$
\begin{align}
  \chi_{\ell / r} = \widetilde{\pi}_{\ell / r} \circ \chi \circ \pi_{\ell / r}|_{\Sigma}^{-1}.
\end{align}
From which we can deduce that each map is smooth and oriented.

However, despite each $\Sigma$ inheriting a Riemannian metric from $\mathcal{M}$,
it is clear that different Cauchy surfaces would yield
different metrics on $\mathcal{M}_{\ell / r}$,
hence we must be able to show that any algebra constructed on a Cauchy surface $\Sigma$
is in some sense independent of this metric.

In the sections that follow,
we shall build a configuration space,
and hence both classical and quantum algebras of observables,
using arbitrarily selected Cauchy surfaces as the underlying space.
Clearly, it will be important to establish that our constructions
do not depend in any significant way upon this choice.

Not only should we expect
diffeomorphic Cauchy surfaces to yield isomorphic algebras,
but we should also expect a `reparametrisation invariance',
where any algebra constructed over a surface $\Sigma$
should carry an action of $\mathrm{Diff}_+(\Sigma)$ by automorphisms.

It is well known that every 2\textsc{d} Lorentzian manifold is conformally flat,
i.e.\ they are locally conformally isometric to $\mathbb{M}^2$.
However, this does not mean that every spacetime is `the same as' Minkowski from the perspective of a \textsc{cft}.
A recent work by Benini, Giorgetti and Schenkel~\cite{beniniSkeletalModel2d2021}
explains in detail the manner in which the category $\mathsf{CLoc}$
can be replaced by a \emph{skeletal} category, the objects of which are just $\mathbb{M}^2$ and $\mathscr{E}$.

Central to this discussion is the extension of the conformal flatness result,
which shows that all \emph{globally hyperbolic} 2\textsc{d} Lorentzian manifolds can be embedded into one of these two spacetimes in a particular way.

\begin{theorem}
  \label{thm:cauchy-extension-skeletal}
  Let $\Sigma_0 \subset \mathcal{M}_0$ denote the $t = 0$ Cauchy surface of either 2\textsc{d} Minkowski space or the Einstein cylinder
  (i.e.\ $\mathcal{M}_0 \in \{\mathbb{M}^2, \mathscr{E}\}$).
  Then, for any orientation-preserving diffeomorphism $\Sigma \overset{\sim}{\rightarrow} \Sigma_0$ where
  $\Sigma \subset \mathcal{M}$ is a Cauchy surface of a 2\textsc{d} globally hyperbolic spacetime,
  there exists a $\mathsf{CLoc}$ morphism $\mathcal{M} \rightarrow \mathcal{M}_0$ such that the following diagram commutes.
  \begin{equation}
    \begin{tikzcd}
      \Sigma \ar[r, "\sim"] \ar[d, hookrightarrow] &
      \Sigma_0 \ar[d, hookrightarrow] \\
      \mathcal{M} \ar[r] &
      \mathcal{M}_0
    \end{tikzcd}
  \end{equation}
\end{theorem}

\begin{proof}
  In~\cite[Theorem 3.3]{beniniSkeletalModel2d2021}, it was shown,
  using~\cite[Proposition 4.2]{finsterLorentzianSpectralGeometry2016}
  and~\cite[Theorem 2.2]{monclairIsometriesLorentzSurfaces2014} for planar and cylindrical spacetimes respectively,
  that there exist \textsf{CLoc} morphisms $\chi_0: \mathcal{M} \rightarrow \mathcal{M}_0$.
  In particular, in~\cite{finsterLorentzianSpectralGeometry2016} an arbitrary Cauchy surface $\Sigma \subset \mathcal{M}$
  is selected such that
  $\Sigma \rightarrow \Sigma_{(0, 1)} := \{(-x, x) \in \mathbb{M}^2 \,|\, x \in (0, 1) \}$, expressed in null coordinates.
  As the image of $\mathcal{M}$ must be causally convex,
  it must be contained within the diamond $U = (-1, 0) \times (0, 1)$.
  Given any choice of oriented diffeomorphism $\rho_1: (0, 1) \overset{\sim}{\rightarrow} \mathbb{R}$,
  one can construct a map $\chi_1: U \rightarrow \mathbb{M}^2$ by sending $(u, v) \mapsto (-\rho_1(-u), \rho_1(v))$.
  This map is clearly conformally admissible,
  as both $\rho_1(v)$ and $-\rho_1(-u)$ are orientation preserving diffeomorphisms.
  A point of the form $(-x, x)$ is mapped to $(-\rho_1(x), \rho_1(x))$, hence $\Sigma_{(0, 1)} \overset{\sim}{\rightarrow} \Sigma_0$.
  Going back to our original Cauchy surface $\Sigma$,
  if we are given an arbitrary diffeomorphism $\rho: \Sigma \overset{\sim}{\rightarrow} \Sigma_0$,
  then we can construct $\chi_1$ from $\rho_1 = \rho \circ \chi_0|_{\Sigma_{(0, 1)}}^{-1}$.
  The desired embedding $\mathcal{M} \rightarrow \mathbb{M}^2$ is then simply $\chi_1 \circ \chi_0$.

  An important property of this embedding is that it can be expressed entirely in terms of $\rho$.
  We can define maps $\pi_{\ell/r}^{\Sigma}: \mathcal{M} \rightarrow \Sigma$ such that $\Sigma \cap \pi_{\ell/r}(x) = \{\pi_{\ell/r}(x)\}$.
  One then has that
  \begin{align}
    \label{eq:cauchy-extension}
    \chi(x) = (-\rho \circ \pi_{\ell}^{\Sigma} (x), \rho \circ \pi_r^{\Sigma} (x)),
  \end{align}
  using null coordinates on $\mathbb{M}^2$.
  This is because $\chi(\Sigma \cap \pi_{\ell/r}(x)) = \rho(\Sigma) \cap \tilde{\pi}_{\ell/r}(\chi(x))$,
  where $\tilde{\pi}_{\ell}(u, v) = \{(u, v') \in \mathbb{M}^2\}_{v' \in \mathbb{R}}$ \emph{etc}.
  This correspondence is represented visually in~\cref{fig:extension}.

  We now consider the case where $\rho: \Sigma \overset{\sim}{\rightarrow} \Sigma_0 \subset \mathscr{E}$ is a compact Cauchy surface of a spacetime $\mathcal{M}$.
  Let $p_{\Sigma}: \overline{\Sigma} \rightarrow \Sigma$ be the universal cover of $\Sigma$.
  Because $\mathcal{M} \simeq \Sigma \times \mathbb{R}$,
  this also defines a universal cover $p: \overline{\mathcal{M}} \rightarrow \mathcal{M}$ which,
  given the pullback metric along $p$, is also globally hyperbolic, with $\overline{\Sigma}$ as a Cauchy surface.
  If we now denote by $\overline{\Sigma}_0$ the $t = 0$ Cauchy surface of Minkowski space,
  which is also the universal cover of $\Sigma_0 \subset \mathscr{E}$,
  and by $p_0: \mathbb{M}^2 \rightarrow \mathscr{E}$ and $p_{\Sigma_0}: \overline{\Sigma}_0 \rightarrow \Sigma_0$
  the canonical projections, we can construct the commutative diagram
  \begin{equation}
    \begin{tikzcd}
      \overline{\Sigma} \ar[r, two heads, swap, "p_{\Sigma}"] \ar[d, hook] \ar[rrr, bend left=20, dashed, "\bar{\rho}"]&
      \Sigma \ar[r, "\rho"] \ar[d, hook] &
      \Sigma_0 \ar[d, hook] &
      \overline{\Sigma}_0 \ar[l, two heads, "p_{\Sigma_0}"] \ar[d, hook]\\
      \overline{\mathcal{M}} \ar[r, two heads, "p"] \ar[rrr, bend right=20, swap, dashed, "\overline{\chi}"]&
      \mathcal{M} &
      \mathscr{E} &
      \mathbb{M}^2 \ar[l, two heads, swap, "p_0"]
    \end{tikzcd}
  \end{equation}
  where $\bar{\rho}$ is the lift of $\rho$,
  and $\overline{\chi}$ is the \textsf{CLoc} morphism corresponding to $\bar{\rho}$ by the previous argument.
  We would like to define $\chi: \mathcal{M}$ $\rightarrow \mathscr{E}$ by $\chi(x) = p_0 \circ \overline{\chi}(\bar{x})$ for some $p(\bar{x}) = x$.
  Clearly this map is well defined if an only if $\overline{\chi}$
  is equivariant with respect to the automorphisms of the covering maps,
  i.e.\ for every $d: \overline{\mathcal{M}} \rightarrow \mathcal{M}$ such that $p \circ d = p$,
  there must be $d_0: \mathbb{M}^2 \rightarrow \mathbb{M}^2$ such that $p_0 \circ d_0 = p_0$,
  and $\overline{\chi} \circ d = d_0 \circ \overline{\chi}$.

  By definition $\bar{\rho}$ is equivariant with respect to the automorphisms of $p_{\Sigma}$ and $p_{\Sigma_0}$,
  where for both covering maps, such automorphisms are restrictions of the aforementioned $d$ and $d_0$.
  Spelling it out, this means that, for every $d_\Sigma: \overline{\Sigma} \rightarrow \overline{\Sigma}$ such that $p_{\Sigma} d_{\Sigma} = p_{\Sigma}$
  there exists $d_0: \overline{\Sigma}_0 \rightarrow \overline{\Sigma}_0$ such that $p_{\Sigma_0} d_{\Sigma_0} = p_{\Sigma_0}$ and
  \begin{align*}
    \bar{\rho} \circ d_{\Sigma} = d_{\Sigma_0} \circ \bar{\rho},
  \end{align*}
  moreover, $d_{\Sigma}$ extends to a map $d: \overline{\mathcal{M}} \rightarrow \overline{\mathcal{M}}$
  and likewise for $d_{\Sigma_0}$.
  (In particular $d_0(u, v) = (u - 2 \pi n, v + 2 \pi n)$,
  hence $d_{\Sigma_0}(x) = x + 2 \pi n$ for some $n \in \mathbb{Z}$.)
  Moreover, the manner in which $\overline{\mathcal{M}}$ is constructed means that $d$ is always an isometry preserving $\Sigma$,
  hence $\pi_{\ell/r}^{\overline{\Sigma}} \circ d = d \circ \pi_{\ell/r}^{\overline{\Sigma}}$.
  We can then use~\eqref{eq:cauchy-extension} to show that
  \begin{align}
    \chi \circ d(x) &= ( -\bar{\rho} \circ \pi_{\ell}^{\overline{\Sigma}} \circ d (x), \bar{\rho} \circ \pi_r^{\overline{\Sigma}} \circ d (x)) \nonumber \\
             &= (- d_{\Sigma_0} \circ \bar{\rho} \circ \pi_{\ell}^{\overline{\Sigma}} (x), d_{\Sigma_0} \circ \bar{\rho} \circ \pi_{r}^{\overline{\Sigma}} (x)) \nonumber \\
             &= d_0 \circ \chi (x),
  \end{align}
  as required.
\end{proof}

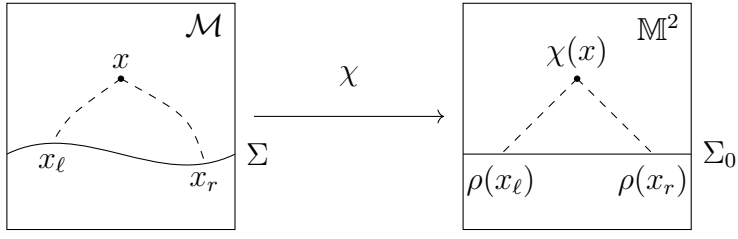
\begin{figure}[ht]
  \centering
  \begin{tikzpicture}
    \draw (-3, -1) rectangle (0, 2) node[anchor=north east] {$\mathcal{M}$};
    \draw (-3, 0) .. controls (-2, 0.5) and (-1, -0.5) .. (0, 0) node[anchor=west] {$\Sigma$};
    \filldraw (-1.5, 1) circle (1pt) node[anchor=south] (x) {$x$};
    \draw[dashed] (-1.5, 1) .. controls (-2.2, 0.5) .. (-2.4, .15) node[anchor=north] {$x_{\ell}$};
    \draw[dashed] (-1.5, 1) .. controls (-0.6, 0.5) .. (-0.4, -0.1) node[anchor=north] {$x_{r}$};

    \node[anchor=south] (map) at (1.5, 0.75) {$\chi$};
    \draw[->] (0.25, 0.5) -- (2.75, 0.5);

    \draw (3, -1) rectangle (6, 2) node[anchor=north east] {$\mathbb{M}^2$};
    \draw (3, 0) -- (6, 0) node[anchor=west] {$\Sigma_0$};
    \filldraw (4.5, 1) circle (1pt) node[anchor=south] (chi) {$\chi(x)$};
    \draw[dashed] (4.5, 1) -- (3.5, 0) node[anchor=north] {$\rho(x_{\ell})$};
    \draw[dashed] (4.5, 1) -- (5.5, 0) node[anchor=north] {$\rho(x_{r})$};
  \end{tikzpicture}
  \caption{\label{fig:extension} A diagrammatic representation of~\eqref{eq:cauchy-extension},
  where we have used the shorthand $x_{\ell/r} = \pi_{\ell/r}^{\Sigma}(x)$.
  The dashed lines represent null geodesics, which are necessarily preserved by the conformally admissible embedding $\chi$.}
\end{figure}

One way of phrasing this result is that,
any diffeomorphism $\Sigma \overset{\sim}{\rightarrow} \Sigma_0$, where $\Sigma_0$ is a Cauchy surface of $\mathbb{M}^2$ of $\mathscr{E}$,
can be extended such that its domain is the entirety of $\mathcal{M}$.
If we instead have a diffeomorphism
$\Sigma \overset{\sim}{\rightarrow} \widetilde{\Sigma}$ to a Cauchy surface of some other spacetime $\widetilde{\mathcal{M}}$,
then we can use~\cref{thm:cauchy-extension-skeletal} to prove a weaker, but more general extension as follows.

\begin{corollary}
  \label{cor:embedding-extension}
  Let $\mathcal{M}, \widetilde{\mathcal{M}}$ be
  a pair of 2\textsc{d} globally hyperbolic spacetimes with Cauchy surfaces $\Sigma, \widetilde{\Sigma}$ respectively.
  For any orientation-preserving embedding $\rho: \Sigma \hookrightarrow \widetilde{\Sigma}$,
  there exists an open, causally convex subset $\mathcal{N} \subseteq \mathcal{M}$ such that $\Sigma$ is also a Cauchy surface of $\mathcal{N}$
  and $\rho$ extends to a \textsf{CLoc} morphism
  $\chi: \mathcal{N} \rightarrow \widetilde{\mathcal{M}}$ such that the following diagram commutes.
  \begin{equation}
    \begin{tikzcd}
      \label{eq:ccauchy-diagram}
      \Sigma \ar[r, "\rho", "\sim"'] \ar[d, hook] &
      \widetilde{\Sigma} \ar[d, hook] \\
      \mathcal{N} \ar[r, "\chi"] &
      \widetilde{\mathcal{M}}
    \end{tikzcd}
  \end{equation}
  Moreover, if $\rho$ is a diffeomorphism, then $\chi$ is Cauchy.
\end{corollary}

\begin{proof}
  Suppose first that $\rho$ is invertible.
  Choose a diffeomorphism $\rho_0: \widetilde{\Sigma} \overset{\sim}{\rightarrow} \Sigma_0$,
  where $\Sigma_0$ is a Cauchy surface for the appropriate choice of $\mathcal{M}_0 \in \{\mathbb{M}^2, \mathscr{E}\}$
  (it is implicit that this and all following diffeomorphisms are orientation-preserving).
  This also provides a diffeomorphism $\rho_1 = \rho_0 \circ \rho: \Sigma \overset{\sim}{\rightarrow} \Sigma_0$.
  Applying~\cref{thm:cauchy-extension-skeletal} to both, we obtain $\chi_1, \chi_0$ in the following diagram.
  \begin{equation}
    \begin{tikzcd}
      \Sigma \ar[r, "\rho", "\sim"'] \ar[dd, hook] &
      \widetilde{\Sigma} \ar[r, "\rho_0"] \ar[d, hook] &
      \Sigma_0 \ar[d, hook] \\
      & \widetilde{\mathcal{M}} \ar[r, "\chi_0"] &
      \mathcal{M}_0 \\
      \mathcal{M} \ar[rru, bend right=20, swap, "\chi_1"] &&&
    \end{tikzcd}
  \end{equation}

  We then consider the space
  $\mathcal{N} = \chi_1^{-1}(\chi_0(\widetilde{\mathcal{M}}) \cap \chi(\mathcal{M}))$.
  This space is open and causally convex as both $\chi_1(\mathcal{M})$ and $\chi_0(\widetilde{\mathcal{M}})$ are,
  and each property is preserved by intersection.
  We can also see that $\mathcal{N}$ contains $\Sigma$,
  as $\chi_1(\Sigma) = \rho_1(\Sigma) = \rho_0(\widetilde{\Sigma}) \subset \chi_0(\widetilde{\mathcal{M}})$.
  Given that $\chi_1(\mathcal{N}) \subseteq \chi_0(\widetilde{\mathcal{M}})$,
  we can also define the map $\chi := \chi_0^{-1} \circ \chi_1: \mathcal{N} \rightarrow \widetilde{\mathcal{M}}$,
  because all \textsf{CLoc} morphisms are diffeomorphisms onto their images.
  Adding this into the above diagram, we obtain
    \begin{equation}
        \begin{tikzcd}
        \Sigma \ar[r, "\rho"] \ar[d, hook] &
        \widetilde{\Sigma} \ar[r, "\rho_0"] \ar[d, hook] &
        \Sigma_0 \ar[d, hook] \\
        \mathcal{N} \ar[r, "\chi"] \ar[d, hook] &
        \widetilde{\mathcal{M}} \ar[r, "\chi_0"] &
        \mathbb{M}^2 \\
        \mathcal{M} \ar[rru, bend right=20, swap, "\chi_1"] &&&
        \end{tikzcd}
    \end{equation}
  the commutativity of which demonstrates that $\chi$ is indeed a Cauchy morphism
  $\mathcal{N} \rightarrow \widetilde{\mathcal{M}}$.

  In the case where $\rho$ is only an embedding, take $\widetilde{\mathcal{N}} \subseteq \widetilde{\mathcal{M}}$
  to be the Cauchy development of $\rho(\Sigma)$,
  which is the set of points $\widetilde{x} \in \widetilde{\mathcal{M}}$ such that
  every inextensible causal curve through $\widetilde{x}$ intersects $\rho(\Sigma)$.
  Clearly $\widetilde{\mathcal{N}}$ is open and causally convex, and hence is a sub-spacetime, with $\rho(\Sigma)$ as a Cauchy surface.
  Then the preceding argument applies by replacing $\widetilde{\mathcal{M}}$ with $\widetilde{\mathcal{N}}$.
\end{proof}

The reason that diffeomorphisms $\mathcal{M} \supset \Sigma \overset{\sim}{\rightarrow} \Sigma_0 \subset \mathbb{M}^2$ can be extended to
the entirety of $\mathcal{M}$ is that,
for any pair of points $x, y \in \Sigma_0$,
one can produce a null geodesic from each such that
they intersect precisely once in $\mathbb{M}^2$.
For example, if we express $\Sigma_0$ in null coordinates as the set $\Sigma_0 = \{(-x, x) \in \mathbb{M}^2\}_{x \in \mathbb{R}}$,
then a right-moving null geodesic from $x$ will intersect a left-moving geodesic from $y$ at the point $(-x, y) \in \mathbb{M}^2$.
If we truncate $\mathbb{M}^2$ to events in the past of some Cauchy surface,
say $t = T$ for some $T > 0$,
then this is no longer the case.
If we consider $\Sigma_0$ to also be a Cauchy surface of the truncated Minkowski spacetime $\mathbb{M}^2_{t<T}$,
the above theorem can extend the identity map in one direction,
resulting in the inclusion $\mathbb{M}^2_{t < T} \to \mathbb{M}^2$,
but there exists \textit{no} conformally admissible embedding
$\mathbb{M}^2 \to \mathbb{M}^2_{t < T}$
which acts as identity on the $t = 0$ Cauchy surface,
as demonstrated by \cref{fig:impossible_extension}.

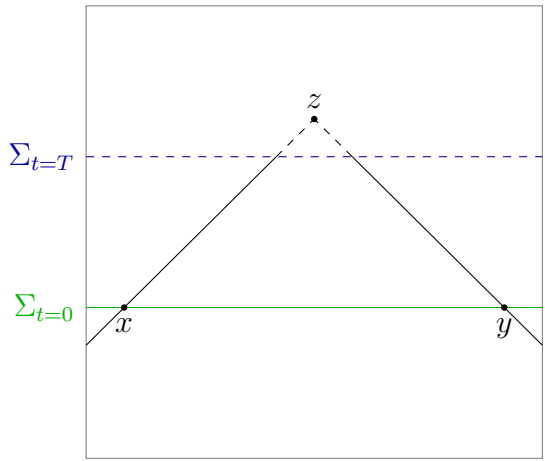
\begin{figure}[htbp]
    \centering
        \begin{tikzpicture}
            \draw[gray] (0, 0) rectangle (6, 6);
            \draw[green!70!black]
                (0, 2) node[anchor=east] {$\Sigma_{t = 0}$} -- (6, 2);
            \draw[blue!60!black, dashed]
                (0, 4) node[anchor=east] {$\Sigma_{t = T}$} -- (6, 4);
            \filldraw[black] (0.5, 2) circle (1pt)
                node[anchor=north] {$x$};
            \filldraw[black] (5.5, 2) circle (1pt)
                node[anchor=north] {$y$};
            \filldraw[black] (3, 4.5) circle (1pt)
                node[anchor=south] {$z$};
            \draw (0, 1.5) -- (2.5, 4);
            \draw (6, 1.5) -- (3.5, 4);
            \draw[dashed] (2.5, 4) -- (3, 4.5);
            \draw[dashed] (3.5, 4) -- (3, 4.5);
        \end{tikzpicture}
    \caption{
        \label{fig:impossible_extension}
        The null geodesics originating from $x$ and $y$
        intersect in $\mathbb{M}^2$, but not in $\mathbb{M}^2_{t < T}$.
        Thus any map $\mathbb{M}^2 \to \mathbb{M}^2_{t < T}$
        which restricts to the identity on $\Sigma_{t = 0}$ cannot preserve
        null geodesics, and hence cannot be conformal.
    }
\end{figure}

We have thus established that oriented diffeomorphisms $\rho: \Sigma \overset{\sim}{\rightarrow} \widetilde{\Sigma}$
extend at least partially to \textsf{CLoc} morphisms $\mathcal{N} \rightarrow \mathcal{M}$.
Notably, these morphisms are always \emph{Cauchy}, as the image of $\mathcal{N}$ always contains $\widetilde{\Sigma}$.
This is significant, because it means that for theories satisfying the time-slice axiom,
diffeomorphisms of Cauchy surfaces yield \textit{isomorphisms} of the corresponding algebras of observables.
In the special cases of $\mathbb{M}^2$ and $\mathscr{E}$,
the time-slice axiom is not even necessary,
as diffeomorphisms of Cauchy surfaces extend fully to \textsf{CLoc} isomorphisms of each spacetime.
Finally, by considering the group $\mathrm{Diff}_+(\Sigma)$,
we shall later see that, by associating a chiral subalgebra to a particular Cauchy surface $\Sigma$ of a given spacetime,
invariance of the algebra under \emph{reparametrisations} of $\Sigma$ comes as
a natural consequence of conformal covariance of the full spacetime algebra.

In the following, we 
construct algebras on Cauchy surfaces which 
capture the chiral (or anti-chiral) sector of the full algebra in question.
In order to show that our constructions are `natural',
we will 
show that for every diagram of the form~\eqref{eq:ccauchy-diagram},
there is a corresponding diagram of algebras.
We lay the groundwork for a precise formulation of this principle with the following definition.

\begin{definition}
  \label{def:ccauchy}
  The category $\mathsf{CCauchy}$ has as
  \begin{itemize}
    \item \emph{Objects:} Pairs $(\Sigma, \mathcal{M})$ such that $\Sigma$ is a Cauchy surface of $\mathcal{M} \in \mathsf{CLoc}$.
    \item \emph{Morphisms:} Pairs $(\rho, \chi)$ such that $\rho: \Sigma \rightarrow \widetilde{\Sigma}$ is a smooth oriented embedding,
          $\chi \in \mathsf{CLoc}(\mathcal{M}, \widetilde{\mathcal{M}})$ such that~\eqref{eq:ccauchy-diagram} commutes.
  \end{itemize}
\end{definition}

\begin{remark}
  There are two natural functors out of $\mathsf{CCauchy}$.
  Clearly we have $\Pi_2: \mathsf{CCauchy} \longrightarrow \mathsf{CLoc}$,
  which sends $(\Sigma, \mathcal{M})$ to $\mathcal{M}$ and $(\rho, \chi)$ to $\chi$.
  The target category for $\Pi_1$, which forgets about $\mathcal{M}$ and $\chi$,
  is $\mathsf{CRie}_1$, the category of Riemannian 1-manifolds with
  smooth oriented embeddings (which are necessarily conformal) as morphisms.
  Note that \emph{both} functors are surjective.
  This is obvious for $\Pi_2$, as every globally hyperbolic spacetime possesses a Cauchy surface.
  Given a Riemannian 1-manifold $\Sigma$, we can also easily construct a globally hyperbolic spacetime
  $\mathcal{M}_{\Sigma} := \Sigma \times \mathbb{R}$ with orientations defined in the obvious way and
  a metric $ds^2 = dt^2 - dx^2$, where $dt$ is the coordinate one-form on $\mathbb{R}$,
  and $dx$ is the metric volume form of $\Sigma$.
\end{remark}

\subsection{Covariant Parametrisation of Chiral Solutions}

In this section we construct a model for the
configuration space of chiral fields.
Similarly to the full spacetime, we refer to elements of this space as \emph{chiral field configurations}.
The term `chiral field' will instead be used for
objects 
analogous to
\emph{locally covariant fields}
which we discuss in \cref{sec:chiral-primary-fields}.

Given a field configuration $\phi \in \mathfrak{E}(\mathbb{M}^2)$
satisfying the equation of motion $\partial_u\partial_v \phi = 0$,
if we take its derivative with respect to $u$,
we obtain a function that is independent of $v$ and \emph{vice-versa}.
This allows us to separate the left-moving term of d'Alembert's solution from the right-moving term. Morover, building our configuration space from derivatives of scalar field on $\mathbb{M}^2$
allows us to avoid the well-known problems
which arise when trying to find a vacuum state for the massless scalar field
(see however \cite{bahnsLocalNetsNeumann2017} for a different way to circumvent these problems).

We would like to formulate statements such as
`$\partial_u \phi$ depends only on $u$'
without explicit reference to our choice of coordinates.
The first issue is that the operator $\partial_u$ depends on a choice of frame.
In~\cite{crawfordLorentzian2dCFT2021} we included such a choice as part of our
background spacetime data.
Here we take a different approach, which does not require the use of frames and also guaranties that the resulting function depends only on one variable.

The Lorentzian metric on a spacetime $\mathcal{M}$ allows for a natural decomposition of the cotangent bundle
$T^* \mathcal{M} = T^*_\ell \mathcal{M} \oplus T^*_r \mathcal{M}$,
where, in null coordinates $u, v$,
the fibres of $T^*_\ell \mathcal{M}$ and $T^*_r \mathcal{M}$
are spanned by $du$ and $dv$ respectively.
Let $\Pi_{\ell / r}: T^*\mathcal{M} \to T^*_{\ell / r} \mathcal{M}$
be the projections onto each subbundle,
the operation $\Pi_\ell d$ then sends $\phi \mapsto \partial_u \phi du$.

We now have a $1$-form on $\mathcal{M}$,
but we would like a function on $\mathcal{M}_\ell$.
Recalling our discussion in the previous section,
we may use a Cauchy surface $\Sigma \subset \mathcal{M}$ as a proxy for
$\mathcal{M}_\ell$.
Restricting $\Pi_\ell d \phi$ to $\Sigma$,
we may then map this to a smooth function by using the Hodge star
$*_\Sigma$ associated to the Riemannian metric on $\Sigma$.
Finally, noting that $\Sigma$ has the opposite orientation to $\mathcal{M}_{\ell}$,
we multiply the resulting function by a factor of $-1$ to account for this
(see the example below).

Thus, altogether we have a map
\begin{equation}
  \label{eq:del-sigma-derivative}
    (-1) \cdot *_\Sigma i_\Sigma^* \Pi_\ell d,
        =:
    \partial_\Sigma: \mathfrak{E}(\mathcal{M}) \longrightarrow \mathfrak{E}(\Sigma)
\end{equation}
where $i_\Sigma: \Sigma \hookrightarrow \mathcal{M}$ is the inclusion map.
We shall henceforth refer to $\partial_{\Sigma}$ as
the \emph{chiral derivative} corresponding to the Cauchy surface $\Sigma \subset \mathcal{M}$.
Similarly, we may also define the \emph{anti-chiral derivative}
$\bar{\partial}_{\Sigma} = *_{\Sigma} i^{*}_{\Sigma} \Pi_{r} d$,
though we will rarely use this, as most statements concerning $\partial_{\Sigma}$ are readily generalised.

As an example,
consider the Cauchy surface in Minkowski space expressed in null coordinates as
$\Sigma = \left\{ (-s, \gamma(s)) \right\}_{s \in \mathbb{R}}$ for some
$\gamma \in \mathrm{Diff}_+(\mathbb{R})$.
As we noted above, given an arbitrary configuration
$\phi \in \mathfrak{E}(\mathbb{M}^2)$,
we have that $\Pi_\ell d \phi = (\partial_u \phi) d u$.
After a quick computation one can verify that,
using the parametrisation of $\Sigma$ given above,
we have that $i_\Sigma^* du = - ds$,
and the induced volume form on $\Sigma$ may be expressed as
$\mathrm{d}V_\Sigma = \sqrt{\gamma'(s)} ds$.
Thus, altogether we have
\begin{equation}
    \label{eq:del-sigma-minkowski}
    (-1) \cdot *_\Sigma i^*_\Sigma \Pi_\ell d \phi(s)
        =
    (\partial_\Sigma \phi)(s)
        =
    \frac{1}{\sqrt{\gamma'(s)}} (\partial_u \phi)(-s, \gamma(s)).
\end{equation}

From this one can also show that for $\Sigma \subset \mathbb{M}^2$,
$\partial_{\Sigma}$ is surjective.
In fact, this is true for arbitrary $\mathcal{M}$. Before we prove that, we first address the question of
how these maps interact with the morphisms from \cref{def:ccauchy}.

In~\cite[\S 4.1]{crawfordLorentzian2dCFT2021},
we introduced \emph{weighted pullbacks} to describe how classical and quantum fields transform under conformal isometries.
Even though we are now considering one-dimensional Riemannian manifolds,
the definition \cite[definition 4.2]{crawfordLorentzian2dCFT2021} of a weighted pullback carries over unchanged.
Given a $\mathsf{CCauchy}$ morphism $(\rho, \chi)$ from $(\Sigma, \mathcal{M})$ to $(\widetilde{\Sigma}, \widetilde{\mathcal{M}})$
such that $\chi^{*} \widetilde{g} = \Omega^2 g$,
the restriction $\rho$ has conformal factor $\left( \Omega|_\Sigma \right)^2$, in the sense that
$\rho^* \widetilde{g}|_{\widetilde{\Sigma}} = \left( \Omega|_\Sigma \right)^2 g|_\Sigma$.

\begin{definition}
  For $\rho \in \mathsf{CRie}_1(\Sigma, \widetilde{\Sigma})$ such that $\rho^{*} g_{\widetilde{\Sigma}} = \omega^2 g_{\Sigma}$,
  we define the \emph{weighted pullback} $\rho^{*}_{(\mu)}: \mathfrak{E}(\widetilde{\Sigma}) \longrightarrow \mathfrak{E}(\Sigma)$ by
  \begin{align}
    \rho^{*}_{(\mu)} \psi := \omega^{\mu} \rho^{*}\psi
  \end{align}
\end{definition}

It turns out that, for $\mu = 1$,
these are precisely the maps required to preserve the images of the
$\partial_\Sigma$ operators, as demonstrated by the following proposition.

\begin{proposition}
  \label{prop:config-covariance}
  Define the contravariant functor $\mathfrak{E}_{(1)}: \mathsf{CRie}_1 \rightarrow \mathsf{Vec}$ such that
  $\mathfrak{E}_{(1)}(\Sigma) = \mathfrak{E}(\widetilde{\Sigma})$,
  and $\mathfrak{E}_{(1)} \rho = \rho^{*}_{(1)}: \mathfrak{E}(\widetilde{\Sigma}) \rightarrow \mathfrak{E}(\Sigma)$.
  Then the morphisms $\partial_{\Sigma}: \mathfrak{E}(\mathcal{M}) \rightarrow \mathfrak{E}(\Sigma)$
  constitute a natural transformation $\partial: \mathfrak{E} \circ \Pi_2 \Rightarrow \mathfrak{E}_{(1)} \circ \Pi_1$,
  i.e.\ for every $\mathsf{CCauchy}$ morphism $(\rho, \chi): (\Sigma, \mathcal{M}) \rightarrow (\widetilde{\Sigma}, \widetilde{\mathcal{M}})$
  the following diagram commutes.
    \begin{equation}
        \begin{tikzcd}
            \mathfrak{E}(\Sigma)
                &
            \mathfrak{E}({\widetilde{\Sigma}})
                \ar[l, "\rho^*_{(1)}"]
                \\
            \mathfrak{E}(\mathcal{M})
                \ar[u, "\partial_\Sigma"]
                &
            \mathfrak{E}(\widetilde{\mathcal{M}})
                \ar[u, "\partial_{\widetilde{\Sigma}}"]
                \ar[l, "\chi^*"]
        \end{tikzcd}
    \end{equation}
\end{proposition}

\begin{proof}
  The exterior derivative $d$ commutes with the pullback along any smooth map,
  and, by the definition of a $\mathsf{CCauchy}$ morphism,
  $i_{\widetilde{\Sigma}}^{*} \chi^{*} = \rho^{*} i_{\Sigma}^{*}$.
  Similarly, $\Pi_{\ell}$ commutes with all conformally admissible embeddings,
  leaving only the Hodge dual to check.
  A standard result of Riemannian geometry states that,
  if $\rho: X \rightarrow Y$ is a conformal embedding of Riemannian $n$-manifolds with
  $\rho^{*}g_Y = \Omega^2 g_X$,
  then the Hodge operator on $p$-forms behaves as
  $*_X \circ \rho^{*} = \Omega^{2p - n} \rho^{*} \circ *_{Y}$.
  For $p = n = 1$ we then get the necessary factor to make the diagram commute.
\end{proof}

\begin{remark}
  As a simple example, one can consider the case where $\widetilde{\mathcal{M}} = \Lambda \mathcal{M}$,
  i.e.\ the underlying manifold is held fixed and
  the metric is scaled by some constant $\Lambda^2 \in \mathbb{R}_{>0}$.
  We may then take $\chi$ to be the identity map of the underlying manifold,
  whereupon the map $\mathfrak{E}(\widetilde{\Sigma}) \rightarrow \mathfrak{E}(\Sigma)$
  in the above proposition becomes $\psi \mapsto \Lambda \psi$.
  As such, the physical interpretation of the above proposition is that
  the chiral boson $\partial\phi$ has a \emph{scaling dimension} of $1$.
\end{remark}

Finally, so that we may be sure that the chiral configuration space is not a proper subspace of $\mathfrak{E}(\Sigma)$,
we have the following result.

\begin{proposition}
  For every $(\Sigma, \mathcal{M}) \in \mathsf{CCauchy}$, $\partial_{\Sigma}$ is surjective.
  Moreover, for every $\psi \in \mathfrak{E}(\Sigma)$, there is a \emph{solution}
  $\phi \in \mathrm{Ker} P_{\mathcal{M}} \subset \mathfrak{E}(\mathcal{M})$ such that
  $\partial_{\Sigma} \phi = \psi$.
\end{proposition}

\begin{proof}
  For $\overline{\Sigma}_0 \subset \mathbb{M}^2$, we can write the solution of $\partial_{\overline{\Sigma}_0} \phi = \psi$ explicitly as
  \begin{align}
    \label{eq:minkowski-chiral-derivative-surjective}
    \phi(u, v) = \int_0^u \psi(-u') \,\mathrm{d}u'.
  \end{align}
  As $\phi$ only depends on $u$, this is clearly a solution to the equations of motion.
  By using $\mathbb{M}^2$ as the universal covering space of $\mathscr{E}$,
  we get the corresponding result for the Cauchy surface $\Sigma_0 \subset \mathscr{E}$,
  however, we must add an extra step.
  If $\psi$ in \eqref{eq:minkowski-chiral-derivative-surjective} is $2\pi$-periodic
  (i.e.\ it corresponds to a function in $\mathfrak{E}(\Sigma_0)$)
  then $\phi(u - 2\pi, v + 2\pi) = \phi(u, v) + \int_0^{2\pi}\psi(x) dx$.
  In other words, $\phi$ only defines a function on $\mathscr{E}$ if $\psi$ is exact.

  To solve this, we choose for the solution of $\partial_{\Sigma_0} \phi = \psi$ on the cylinder
  \begin{align}
    \phi(u, v) = \int_0^u \psi_0(-u') \,\mathrm{d}u' + \frac{1}{2\pi} \left( \int_0^{2\pi} \psi(x) \,\mathrm{d}x \right)(u + v),
  \end{align}
  where $\psi_0(x) = \psi(x) - \tfrac{1}{2\pi}\int_0^{2 \pi} \psi(x') \,\mathrm{d}x'$ is the `exact part' of $\psi$.
  This is then clearly in the form \eqref{eq:cylinder-gen-soln} of a general solution to the wave equation on a cylinder.

  For an arbitrary element $(\Sigma, \mathcal{M}) \in \mathsf{CCauchy}$,
  we take a diffeomorphism $\rho: \Sigma \overset{\sim}{\rightarrow} \Sigma_0$,
  where as before $\Sigma_0$ is the $t=0$ Cauchy surface of $\mathcal{M}_0 \in \{\mathbb{M}^2, \mathscr{E}\}$ as appropriate.
  We can then solve $\partial_{\Sigma} \phi = \psi$ for any $\psi \in \mathfrak{E}(\Sigma)$ using
  the corresponding embedding $\chi: \mathcal{M} \rightarrow \mathcal{M}_0$.
  In particular, suppose that $\partial_{\Sigma_0} \phi_0 = (\rho^{-1})^{*}_{(1)} \psi$
  is one of the solutions constructed above, then
  \begin{align}
    \partial_{\Sigma} \chi^{*} \phi_0 = \rho^{*}_{(1)} \partial_{\Sigma_0} \phi_0 = \psi.
  \end{align}
  Moreover, as $\chi^{*}$ maps $\mathrm{Ker}\, P_{\mathcal{M}_0} \rightarrow \mathrm{Ker}\, P_{\mathcal{M}}$
  (c.f. \cite[proposition 4.4]{crawfordLorentzian2dCFT2021}),
  $\phi$ is also a solution to the equations of motion as desired.
\end{proof}

\section{Classical Observables}

Now that we have identified our configuration space,
and how it transforms under appropriate morphisms,
we can begin to discuss the algebras of observables,
and from there the dynamics, of the massless scalar field.

\subsection{Classical Chiral Algebra}

There are multiple approaches to defining the chiral algebra of observables.
We begin by looking at a space common to all definitions: the regular, linear observables.

We denote by $\{\Psi_{\Sigma}(f) \,|\, f \in \mathfrak{D}(\mathcal{M})\}$
the family of linear observables
\begin{equation}
  \label{eq:chiral-boson}
    \mathfrak{E}(\Sigma) \to \mathbb{R}, \quad
    \psi \mapsto \int_{\Sigma} f \psi \mathrm{d}V_{\Sigma}.
\end{equation}
Naturally, we can also push these forward to maps
$\partial_{\Sigma}^{*}\Psi_{\Sigma}(f): \mathfrak{E}(\mathcal{M}) \to \mathbb{R}$,
by $\phi \mapsto \Psi_{\Sigma}(f)[\partial_{\Sigma}\phi]$.
It is easy to see that $\partial_{\Sigma}^{*}\Psi_{\Sigma}(f)$ is both linear and continuous on
$\mathfrak{E}(\mathcal{M})$, i.e.\ it is a distribution.
However, even though $\Psi_{\Sigma}(f)$ is regular,
$\partial_{\Sigma}^{*}\Psi_{\Sigma}(f)$ fails to be even microcausal:
because its support lies entirely within $\Sigma$,
its wavefront set lies normal to $\Sigma$,
which in particular means that it contains timelike covectors.
Thur arguments which showed the microcausal functionals to form a closed
Poisson algebra for the full spacetime will not be directly applicable here.


However, for these observables, a direct computation of the Peierls bracket yields
\begin{equation}
    \left\{ \partial_\Sigma^*\Psi(f), \partial_\Sigma^*\Psi(g) \right\}
        =
    \left\langle
        (\partial_\Sigma \otimes \partial_\Sigma)E,
        f \otimes g
    \right\rangle_{\Sigma^2}.
\end{equation}
In other words, the commutator function for chiral observables is simply
$(\partial_\Sigma \otimes \partial_\Sigma)E$,
where we must verify that the pullback in the definition \eqref{eq:del-sigma-derivative} of $\partial_\Sigma$
is well defined on $E$.
This is the case because $\mathrm{WF}(E)$ has a vanishing intersection
with the conormal bundle of $\Sigma^2 \subset \mathcal{M}^2$.
We shall henceforth denote this distribution by the shorthand $E_{\Sigma} \in \mathfrak{D}'(\Sigma^2)$.
Note that this would not be possible if we considered pullbacks of $E^{R/A}$ independently.

On Minkowski space,
the integral kernel of the Pauli-Jordan function is expressed
in null coordinates as
\begin{equation}
  \label{eq:pauli-jordan-minkowski}
    E(u, v, u', v')
        =
    -\frac{1}{4} \left(
        \mathrm{sgn}(u - u') + \mathrm{sgn}(v - v')
    \right).
\end{equation}
Taking $\Sigma_0$ to be the $t = 0$ Cauchy surface,
we use \eqref{eq:del-sigma-minkowski}
to compute the chiral commutator function as
\begin{equation}
  \label{eq:chiral-commutator-minkowski}
    E_{\Sigma_0}(s, s')
        =
    \frac{1}{2} \delta'(s - s').
\end{equation}
This agrees with~\cite[(3.10)]{crawfordLorentzian2dCFT2021}
though now the reduced number of coordinates is expressed in a more geometric manner by
taking the pullback of $(\partial_u \otimes \partial_u)E$ along an embedding
of the $1$-manifold $\Sigma$ into the $2$-manifold $\mathbb{M}^2$.
By the method of images \cite{crawfordLorentzian2dCFT2021}, we can deduce that
$E_{\Sigma_0}$ for $\Sigma_0 \subset \mathscr{E}$ must be of the same form.

It will be useful to note that
We then transfer the result from these two explicit examples to arbitrary spacetimes.
Let $(\Sigma, \mathcal{M}) \in \mathsf{CCauchy}$, and choose a diffeomorphism
$\rho: \Sigma \overset{\sim}{\rightarrow} \Sigma_0 \subset \mathcal{M}_0 \in \{\mathbb{M}^2, \mathscr{E}\}$.
From \cref{prop:config-covariance}, we have that
$\rho^{*}_{(1)} \partial_{\Sigma_0} = \partial_{\Sigma} \chi^*$.
Given the conformal covariance of the causal propagator (\cite[Proposition 4.3]{crawfordLorentzian2dCFT2021}),
this implies
\begin{align}
  \label{eq:chiral-commutator-function}
  (\rho^{*}_{(1)} \otimes \rho^{*}_{(1)}) E_{\Sigma_0}
  &= (\rho^{*}_{(1)} \otimes \rho^{*}_{(1)}) (\partial_{\Sigma_0} \otimes \partial_{\Sigma_0}) E_{\mathcal{M}_0} \nonumber \\
  &= (\partial_{\Sigma} \otimes \partial_{\Sigma}) (\chi^{*} \otimes \chi^{*}) E_{\mathcal{M}_0} \nonumber \\
  &=: E_{\Sigma}.
\end{align}
This tells us that $E_{\Sigma}$ is a well-defined distribution in $\mathfrak{D}'(\Sigma^2)$ as desired,
and that $\mathrm{WF}(E_{\Sigma}) = (\rho^{\otimes 2})^{*} \mathrm{WF}(E_{\Sigma_0})$.

Now that we have a bi-distribution on $\Sigma$,
we can define a binary operation, for $F, G \in \mathfrak{F}_{\mathrm{reg}}(\Sigma)$,
and $\psi \in \mathfrak{E}(\Sigma)$ by
\begin{equation}
\label{eq:chiral-bracket-def}
    \left\{ F, G \right\}_{\ell}^{\Sigma}[\psi]
        =
    \left\langle E_{\Sigma}, F^{(1)}[\psi] \otimes G^{(1)}[\psi] \right\rangle.
\end{equation}
Similarly to the construction of $\mathfrak{P}(\mathcal{M})$,
we may now ask if there exists a space of functionals which is closed under this operation.

\begin{proposition}
\label{prop:chiral-poisson}
    Let $\mathfrak{F}_c(\Sigma)$ be the space comprising functionals $F: \mathfrak{E}(\Sigma) \rightarrow \mathbb{R}$
    such that
    \begin{enumerate}
        \item $F$ is Bastiani smooth with respect to the Fréchet topology on $\mathfrak{E}(\Sigma)$
        \item $\mathrm{WF}(F^{(n)}[\psi]) \cap \left( \Xi_{+}^{n} \cup \Xi_{-}^{n} \right) = \emptyset$ where
        \begin{align*}
            \Xi_\pm^{n}
                =
          \left\{
            (s_1, \ldots, s_n; \xi_1, \ldots, \xi_n) \in T^{*} \Sigma^n
                \,|\,
            \pm \xi_i \ge 0, 0 \leq i \leq 1
          \right\},
        \end{align*}
        and the sign of a covector is defined with respect to an arbitrary oriented coordinate on $\Sigma$.
    \end{enumerate}
then $\left\{ \cdot, \cdot \right\}_{\ell}^{\Sigma}$
is a Poisson bracket on $\mathfrak{F}_{c}(\Sigma)$.
We denote the resulting Poisson algebra $\mathfrak{P}_{\ell}(\Sigma, \mathcal{M})$
\end{proposition}

\begin{proof}
  Note that $\Xi^1_+ \cup \Xi^1_- = \dot{T}^{*} \Sigma$,
  hence $F^{(1)}[\psi]$ is a regular distribution $\forall F \in \mathfrak{F}_c(\Sigma), \psi \in \mathfrak{E}(\Sigma)$.
  This means that, for $F, G \in \mathfrak{F}_c(\Sigma)$,
  the bracket \eqref{eq:chiral-bracket-def} is well-defined.
  As such, it remains to show that $\mathfrak{F}_c(\Sigma)$ is closed under these operations,
  and that $\left\{ \cdot, \cdot \right\}_{\ell}^{\Sigma}$ has all the properties of a Poisson bracket.

  The fact that $\left\{ \cdot, \cdot \right\}_{\ell}^{\Sigma}$ is a skew-symmetric bilinear form
  follows immediately from the definition,
  as does the fact that it is a derivation in each of its arguments.
  Rather than directly proving that the Jacobi identity is satisfied,
  we shall later prove that there is an injective homomorphism
  $\mathfrak{P}_{\ell}(\Sigma, \mathbb{M}^2) \rightarrow \mathfrak{P}(\mathcal{M})$,
  thus the Jacobi identity on $\mathfrak{P}_{\ell}(\Sigma, \mathbb{M}^2)$ follows from
  the same identity on $\mathfrak{P}(\mathcal{M})$.

  Thus all that remains is to show $\mathfrak{F}_c(\Sigma)$ is closed under
  $\left\{ \cdot, \cdot \right\}_{\ell}^{\Sigma}$.
  The argument here proceeds along the same lines as
  the closure proof in~\cite[Appendix B]{crawfordLorentzian2dCFT2021},
  but we shall outline the steps explicitly here.

  It is sufficient to show that,
  $\forall F, G \in \mathfrak{F}_c(\Sigma), \psi \in \mathfrak{E}(\Sigma)$, and $k, m \in \mathbb{N}$,
  \begin{equation}
    \label{eq:classical-testimate}
    \mathrm{WF}\left( \left\langle E_{\Sigma}, F^{(k + 1)}[\psi] \otimes G^{(m + 1)}[\psi] \right\rangle \right)
    \cap
    (\Xi^{(k + m)}_+ \cup \Xi^{(k + m)}_-) = \emptyset,
  \end{equation}
  where $\left\langle \cdot, \cdot \right\rangle$ pairs the first variable of $E_{\Sigma}$ with the first variable of
  $F^{(k + 1)}[\psi]$ and the second variable of $E_{\Sigma}$ with the first variable of $G^{(m + 1)}[\psi]$
  according to~\cite[Theorem~8.2.14]{hormanderAnalysisLinearPartial2015}.

  For simplicity, we will suppress the $\psi$ dependence of $F^{(k + 1)}[\psi]$ and $G^{(m + 1)}[\psi]$
  for the rest of this proof, and we shall also restrict our attention to $\Sigma_0 \subset \mathbb{M}^2$.
  We will also use $(\underline{s}_F; \underline{\xi}_F)$ as short hand for an element of
  $T^{*}\Sigma^{k + 1} \simeq \mathbb{R}^{2(k + 1)}$ \emph{etc}.
  As $F$ and $G$ are Bastiani smooth, $F^{(k + 1)} \otimes G^{(m + 1)}$ is compactly supported,
  we just have to consider the set
  \begin{equation*}
    \begin{split}
    \mathrm{WF}(F^{(k+1)} \otimes G^{(m+1)}) \circ \overline{\mathrm{WF}(E_{\Sigma})} &:= \\
    \big\{
        (\underline{s}_F, \underline{s}_G ; \underline{\xi}_F, \underline{\xi}_G) \in T^{*}\Sigma^{k+m}
        &\,|\,
        \exists (r_1, r_2; \eta_1, \eta_2) \in \overline{\mathrm{WF}(E_{\Sigma})}, \\
        &(r_1, \underline{s}_F, r_2, \underline{s}_G; \eta_1, \underline{\xi}_F, \eta_2, \underline{\xi}_G)
        \in \mathrm{WF}(F^{(k+1)} \otimes G^{(m+1)})
    \big\},
    \end{split}
  \end{equation*}
  where $\overline{\mathrm{WF}(E_{\Sigma})} := \mathrm{WF}(E_{\Sigma}) \cup \underline{0}_{\Sigma^2}$,
  and $\underline{0}_{X}$ denotes the zero section of $T^{*} X$.
  Firstly, if this set avoids $\underline{0}_{\Sigma^{k+m}}$,
  then the distribution in \eqref{eq:classical-testimate} is well-defined,
  in which case $\mathrm{WF}(F^{(k+1)} \otimes G^{(m+1)}) \circ \overline{\mathrm{WF}(E_{\Sigma})}$
  contains its wavefront set.

  We need firstly the estimate~%
  \cite[Theorem~8.2.9]{hormanderAnalysisLinearPartial2015}
  \begin{equation*}
    \mathrm{WF}(F^{(k + 1)} \otimes G^{(m + 1)})
    \subseteq
    \overline{\mathrm{WF}(F^{(k+1)})} \times
    \overline{\mathrm{WF}(G^{(m+1)})} \setminus \underline{0}_{\Sigma^{k + m + 2}}
  \end{equation*}
  and secondly the wavefront set of our commutator function
  \begin{equation*}
    \mathrm{WF}(E_{\Sigma}) =
    \{ (r, r; \eta, -\eta) \in \dot{T}\Sigma^2 \},
  \end{equation*}
  which is readily obtained by inspection of~\eqref{eq:chiral-commutator-minkowski}.
  Suppose
  \begin{equation*}
    (\underline{s}_F, \underline{s}_G; \underline{\xi}_{F}, \underline{\xi}_{G})
        \in
    \Xi^{k+n}_{\pm}
  \end{equation*}
  Unpacking the notation, this means there exists some
  $(r_1, r_2; \eta, -\eta) \in \overline{\mathrm{WF}(E_{\Sigma})}$ such that
  \begin{equation*}
    (r_1, \underline{s}_F;  \eta, \underline{\xi}_F) \in \overline{\mathrm{WF}(F^{(k+1)})}, \quad
    (r_2, \underline{s}_G; -\eta, \underline{\xi}_G) \in \overline{\mathrm{WF}(G^{(m+1)})}
  \end{equation*}
  with at least one of these covectors being non-zero.
  Suppose in particular that
  $
  (\underline{s}_{F}, \underline{s}_{G}; \underline{\xi}_{F}, \underline{\xi}_{G}) \in \Xi^{k+n}_+
  $
  Because $(r_1, \underline{s}_F; \eta, \underline{\xi}_F) \notin \Xi^{k+1}_+ \setminus \underline{0}_{\Sigma^{k+1}}$,
  we see that $\eta \le 0$.
  But then $(r_2, \underline{s}_G; -\eta, \underline{\xi}_G) \notin \mathrm{WF}(G^{(m+1)})$,
  hence $(r_2, \underline{s}_G, -\eta, \underline{\xi}_G) \in \underline{0}_{\Sigma^{m+1}}$.
  This in turn implies that $(r_1, \underline{s}_F; \eta, \underline{\xi}_F)$
  cannot belong to $\mathrm{WF}(F^{(k+1)})$, hence
  $
    (r_1, \underline{s}_F, r_2, \underline{s}_G; \eta_1, \underline{\xi}_F, \eta_2, \underline{\xi}_G)
    \in \underline{0}_{\Sigma^{k+m+2}},
  $
  which is disjoint from
  $ \mathrm{WF}(F^{(k + 1)} \otimes G^{(m + 1)}) \circ \overline{\mathrm{WF}(E_{\Sigma})}$.
  As $\underline{0}_{\Sigma^{k+m}} \subset \Xi^{k+m}_+$, this automatically tells us that
  Hörmander's criterion is satisfied, hence the distribution in~\eqref{eq:classical-testimate}
  is well-defined.
  Similar reasoning to the above shows that $\Xi_-^{k+m}$ is also disjoint from
  $ \mathrm{WF}(F^{(k + 1)} \otimes G^{(m + 1)}) \circ \overline{\mathrm{WF}(E_{\Sigma})}$,
  hence~\eqref{eq:classical-testimate} holds.
\end{proof}

\begin{remark}
    The equation \eqref{eq:chiral-bracket-def}
    can be expressed without coordinates as

    \begin{equation}
      \label{eq:geometric-chiral-bracket}
        \left\langle
            (\partial_\Sigma \otimes \partial_\Sigma)E,
            f \otimes g
        \right\rangle
            =
        - \frac{1}{2} \int_\Sigma
            f (*d_\Sigma g) \, \mathrm{d}V_\Sigma.
    \end{equation}

    Hence we may instead express the commutator function for the chiral bracket as
    $\frac{1}{2} * d_\Sigma$.
    This is consistent with \eqref{eq:chiral-commutator-function} which,
    due to the non-uniqueness of $\rho: \Sigma \overset{\sim}{\rightarrow} \Sigma_0$,
    implies that $E_{\Sigma}$ must be invariant under $\mathrm{Diff}_+(\Sigma)$.
\end{remark}

Finally, we examine how the algebras we have defined behave under conformally admissible embeddings.

\begin{proposition}
  \label{prop:cauchy-algebra-natural}
  Let $(\rho, \chi): \mathcal{M} \rightarrow \widetilde{\mathcal{M}}$ be a $\mathsf{CCauchy}$ morphism.
  Define
  \begin{align}
    \mathfrak{P}_{\ell} (\rho, \chi):
    \mathfrak{P}_{\ell}(\Sigma, \mathcal{M}) &\rightarrow
    \mathfrak{P}_{\ell}(\widetilde{\Sigma}, \widetilde{\mathcal{M}}) \nonumber \\
    F &\mapsto F \circ \rho^{*}_{(1)}.
  \end{align}
  Then $\mathfrak{P}_{\ell}$ defines a functor $\mathsf{CCauchy} \rightarrow \mathsf{Poi}$.
\end{proposition}

\begin{proof}
  As $\rho$ is an oriented diffeomorphism,
  one can fairly quickly convince themselves that the sets $\Xi^n_{\pm}$ are preserved by the induced maps
  $\rho_{*}: T^{*} \Sigma^n \rightarrow T^{*} \widetilde{\Sigma}^n$,
  thus $\mathfrak{F}_c(\Sigma) \rightarrow \mathfrak{F}_c(\widetilde{\Sigma})$ under this map.

  To show this map is a Poisson algebra homomorphism, consider,
  for $F, G \in \mathfrak{F}_c(\Sigma)$, $\widetilde{\psi} \in \mathfrak{E}(\widetilde{\Sigma})$
  \begin{align}
    \label{eq:p-ell-naturality-1}
    \left\{ \mathfrak{P}_{\ell} \chi F, \mathfrak{P}_{\ell} \chi G \right\}_{\ell}^{\widetilde{\Sigma}} [\widetilde{\psi}]
    =
    \left\langle
        E_{\widetilde{\Sigma}},
        \left( \mathfrak{P}_{\ell} \chi F \right)^{(1)}[\widetilde{\psi}] \otimes
        \left( \mathfrak{P}_{\ell} \chi G \right)^{(1)}[\widetilde{\psi}]
    \right\rangle.
  \end{align}
  A quick calculation shows that
  $
    \left\langle
      \left( \mathfrak{P}_{\ell} \chi F \right)^{(1)}[\widetilde{\psi}],
      f
    \right\rangle
    =
    \left\langle
      F^{(1)} [\rho^{*}_{(1)} \widetilde{\psi}],
      \rho^{*}_{(1)} f
    \right\rangle
  $.
  Then, recalling \Cref{prop:config-covariance},
  we have $\rho^{*}_{(1)} \partial_{\widetilde{\Sigma}} = \partial_{\Sigma} \chi^{*}$,
  from which we can deduce that
  \begin{align}
    \label{eq:cauchy-commutator-pullback}
    ( \rho^{*}_{(1)} \otimes \rho^{*}_{(1)} ) E_{\widetilde{\Sigma}} = (\partial_{\Sigma} \otimes \partial_{\Sigma}) E_{\mathcal{M}} = E_{\Sigma},
  \end{align}
  thus \eqref{eq:p-ell-naturality-1} becomes simply
  $\left\{ F, G \right\}_{\ell}^{\widetilde{\Sigma}}[\rho^{*}_{(1)} \widetilde{\psi}]$ as desired.
\end{proof}

\subsection{Comparison with Peierls Algebra}

As we have already seen,
direct comparison between chiral observables and observables of the full spacetime algebra
is complicated by the fact that $\partial_{\Sigma}^{*}: \mathfrak{D}'(\Sigma) \rightarrow \mathfrak{D}'(\mathcal{M})$
fails to send regular distributions on $\Sigma$ to microcausal distributions on $\mathcal{M}$.
This is due to the fact that the restriction $i_{\Sigma}^{*}: \Omega^1(\mathcal{M}) \rightarrow \Omega^1(\Sigma)$
in \eqref{eq:del-sigma-derivative} is `too sharp'.

In order to make comparisons, we therefore wish to find a more regular map,
which coincides with $\partial_{\Sigma}$ on-shell.
We begin again with the example of the $t = 0$ Cauchy surface, $\Sigma_0 \subset \mathbb{M}^2$.
In null coordinates, and for $\epsilon > 0$
we define the family of maps
\begin{align*}
  \partial_{\Sigma_0, \epsilon}: \mathfrak{E}(\mathbb{M}^2) &\rightarrow \mathfrak{E}(\Sigma_0)\\
  \phi &\mapsto \int_{\mathbb{R}} (\partial_u\phi)(-s, v) \delta_{\epsilon}(\tfrac{-s+v}{2}) \, \mathrm{d}v
\end{align*}
where the family $\left\{ \delta_{\epsilon} \right\}_{\epsilon > 0}$ constitute a \emph{nascent delta},
i.e.\ each function is smooth, integrates to $1$, and satisfies $\mathrm{supp} \, \delta_{\epsilon} = [-\epsilon, \epsilon]$.
In the limit as $\epsilon \rightarrow 0$, these maps weakly converge to $\partial_{\Sigma_0}$
in the sense that, for all $f \in \mathfrak{D}(\Sigma)$,
$\left\langle \partial_{\Sigma_0, \epsilon}\phi, f \right\rangle_{\Sigma_0} \rightarrow \left\langle \partial_{\Sigma_0}\phi, f \right\rangle_{\Sigma_0}$.
Moreover, if $\phi \in \mathrm{Ker}\,P$, then $\partial_{\Sigma_0, \epsilon} \phi = \partial_{\Sigma} \phi$ for every $\epsilon > 0$.
Define the natural transformation $\Pi_{\mathrm{on}}: \mathfrak{P} \Rightarrow \mathfrak{P}_{\mathrm{on}}$ by
\begin{align}
  (\Pi_{\mathrm{on}})_{\mathcal{M}}:
    \mathfrak{F}_{\mu c}(\mathcal{M}) &\rightarrow \mathfrak{F}_{\mu c}(\mathcal{M}) / \mathfrak{I}_S(\mathcal{M})\\
    \mathcal{F} &\mapsto [\mathcal{F}],
\end{align}
we obtain $(\Pi_{\mathrm{on}})_{\mathbb{M}^2} \circ \partial_{\Sigma_0, \epsilon}^{*} = (\Pi_{on})_{\mathbb{M}^2} \circ \partial_{\Sigma_0, \epsilon'}^{*}$,
$\forall \epsilon, \epsilon' > 0$.

\begin{proposition}
  \label{prop:classical-chiral-homo}
  For every $\epsilon>0$, the map $\partial_{\Sigma_0, \epsilon}^{*}: \mathfrak{F}(\Sigma_0) \rightarrow \mathfrak{F}(\mathbb{M}^2)$
  defined such that $\partial_{\Sigma_0, \epsilon}^{*} F[\phi] = F[\partial_{\Sigma_0, \epsilon} \phi]$
  yields an injective Poisson algebra homomorphism
  $\mathfrak{P}_{\ell}(\Sigma_0, \mathbb{M}^2) \rightarrow \mathfrak{P}(\mathbb{M}^2)$.
\end{proposition}

\begin{proof}
  Firstly, we must show that the image of $\mathfrak{F}_c(\Sigma_0)$ under $\partial_{\Sigma_0, \epsilon}^{*}$
  lies within $\mathfrak{F}_{\mu c}(\mathbb{M}^2)$.
  If we restrict our attention to the linear, regular observables on $\Sigma$,
  we initially have a map $\partial_{\Sigma_0, \epsilon}^{*}: \mathfrak{D}(\Sigma_0) \rightarrow \mathfrak{D}'(\mathbb{M}^2)$.
  This map is linear and continuous, hence it has an associated Schwartz kernel
  $K \in \mathfrak{D}'(\Sigma_0 \times \mathbb{M}^2)$.
  We can then use this kernel to compute the wavefront sets of
  functionals in the image of $\partial_{\Sigma_0, \epsilon}^{*}$,
  as $\forall \phi \in \mathfrak{E}(\mathbb{M}^2), h \in \mathfrak{D}((\mathbb{M}^2)^n)$
  \begin{align}
    \left\langle (\partial_{\Sigma_0, \epsilon}^{*}F)^{(n)}[\phi], h \right\rangle
    &=
    \left\langle K^{\otimes n}, h \otimes F^{(n)}[\partial_{\Sigma_0, \epsilon} \phi] \right\rangle
  \end{align}
  where the first variable of each copy of $K$ is paired with
  a variable of $h$ and so the rest with $F^{(n)}[\partial_{\Sigma_0, \epsilon} \phi]$.
  We can then use~\cite[Theorem~8.2.14]{hormanderAnalysisLinearPartial2015} once again
  to estimate $\mathrm{WF}((\partial_{\Sigma_0, \epsilon}^{*}F)^{(n)}[\phi])$ given estimates for
  $\mathrm{WF}(K^{\otimes n})$ and $\mathrm{WF}(F^{(n)}[\partial_{\Sigma_0, \epsilon} \phi])$.

  By inspection, and using the appropriate coordinates, the integral kernel $K$ may be written as
  \begin{equation}
      K(s, u, v) = -\partial_u \left( \delta(u + s) \delta_{\epsilon}(\tfrac{u+v}{2}) \right),
  \end{equation}
  from which we may deduce that
  \begin{equation}
      \mathrm{WF}(K) = \left\{ (-u, u, v; \xi, \xi, 0) \in \dot{T}^{*} (\Sigma_0 \times \mathbb{M}^2) \right\}.
  \end{equation}
  After some work, the corresponding estimate is then
  \begin{equation}
    \begin{split}
      \mathrm{WF}((\partial_{\Sigma_0, \epsilon}^{*}F)^{(n)}[\phi])
        \subseteq
      \big\{
        (u_1, v_1, \ldots, &u_n, v_n; \xi_1, 0, \ldots, \xi_n, 0) \in \dot{T}^{*} (\mathbb{M}^2)^n
        \,|\,\\
        &(-u_1, \ldots, -u_n; \xi_1, \ldots, \xi_n) \in \mathrm{WF}(F^{(n)}[\partial_{\Sigma_0, \epsilon} \phi])
      \big\}.
    \end{split}
  \end{equation}
  The wavefront set condition on $F^{(n)}[\partial_{\Sigma_0, \epsilon} \phi]$ then precludes the option that
  $\xi_i$ all have the same sign, which is precisely what we need to conclude that
  $\mathrm{WF}((\partial_{\Sigma_0, \epsilon}^{*}F)^{(n)}[\phi]) \cap (V^n_+ \cup V^n_-) = \emptyset$,
  i.e.\ $\partial_{\Sigma_0, \epsilon}^{*}: \mathfrak{F}_c(\Sigma_0) \rightarrow \mathfrak{F}_{\mu c}(\mathbb{M}^2)$.

  Having established this, the more straightforward task is to show
  that this map preserves Poisson brackets.
  Spelling it out, we need to demonstrate that
  $\forall F, G \in \mathfrak{F}_c(\Sigma_0), \phi \in \mathfrak{E}(\mathbb{M}^2)$
  \begin{equation}
    \left\langle E, (\partial_{\Sigma_0, \epsilon}^{*}F)^{(1)}[\phi] \otimes (\partial_{\Sigma_0, \epsilon}^{*}G)^{(1)}[\phi] \right\rangle
    =
    \left\langle E_{\Sigma_0}, F^{(1)}[\partial_{\Sigma_0, \epsilon} \phi] \otimes G^{(1)}[\partial_{\Sigma_0, \epsilon} \phi] \right\rangle.
  \end{equation}
  This simply amounts to the statement that
  $(\partial_{\Sigma_0, \epsilon} \otimes \partial_{\Sigma_0, \epsilon}) E = (\partial_{\Sigma_0} \otimes \partial_{\Sigma_0}) E$.
  This is easily verified by looking at the precise form~\eqref{eq:pauli-jordan-minkowski} of $E$.
  However, note also that, as a map, the image of $E$ is in the kernel of the wave-operator $P$,
  hence $\partial_{\Sigma_0, \epsilon} E = \partial_{\Sigma_0} E$.
  By skew-symmetry, we can also verify that acting on the second argument of $E$
  with $\partial_{\Sigma_0, \epsilon}$ behaves the same way.

  Lastly, the fact that this map is injective is a direct consequence of the fact that
  $\partial_{\Sigma_0, \epsilon}$ is surjective.
\end{proof}

Finally, we use our embedding theorems to create analogous embeddings
for the chiral algebra of an arbitrary element $(\Sigma, \mathcal{M}) \in \mathsf{CCauchy}$.
Suppose $(\rho, \chi): (\Sigma, \mathcal{M}) \rightarrow (\Sigma_0, \mathbb{M}^2)$ is a $\mathsf{CCauchy}$ morphism.
Starting from the equation $\partial_{\Sigma} \chi^{*} = \rho^{*}_{(1)} \partial_{\Sigma_0}$,
we might try to define a regularised chiral derivative for $\Sigma$ by
$\partial_{\Sigma, \epsilon} = \rho^{*}_{(1)} \partial_{\Sigma_0} \chi_{*}$,
where we make sense of the pushforward by only asking for
$\chi_{*} \phi (x) = \phi(\chi^{-1}(x))$ to hold in some neighbourhood of $\Sigma_0$.
As $\partial_{\Sigma_0} \phi_0$ only depends on $d \phi_0|_{\Sigma_0}$,
this map is well defined.
However, by smoothing out $\partial_{\Sigma_0}$ to $\partial_{\Sigma_0, \epsilon}$,
we increased the region it is sensitive to.
In order to make sense of $\partial_{\Sigma_0, \epsilon} \chi_{*}$,
we need to make sure that the support of $\partial_{\Sigma_0, \epsilon}$ is contained within the image of $\chi$.
For the particular way we have constructed $\partial_{\Sigma_0, \epsilon}$,
this means that $\mathrm{Img} \chi$ must contain the $t=\pm \epsilon$ Cauchy surfaces $\Sigma_{\pm \epsilon}$.
This is because $\partial_{\Sigma_0, \epsilon}$ defines a map $\mathfrak{E}(U) \rightarrow \mathfrak{E}(\Sigma_0)$
for any open, causally-convex neighbourhood $U$ of $\mathcal{J}^-(\Sigma_{\epsilon}) \cap \mathcal{J}^+(\Sigma_{-\epsilon})$.
Hence, for any embedding $\chi: \mathcal{M} \rightarrow \mathbb{M}^2$ such that $\chi(\mathcal{M}) \supset U$,
we can define $\partial_{\Sigma_0, \epsilon} \circ \chi_* := \partial_{\Sigma_0, \epsilon} (\chi^{-1})^{*}$,
where $\chi^{-1}: \chi(\mathcal{M}) \overset{\sim}{\rightarrow} \mathcal{M}$.

This might seem simple, as we can make $\epsilon$ arbitrarily small.
However, if we consider the case where $\mathcal{M} \subset \mathbb{M}^2$
is the space in between the Cauchy surfaces expressed in $(t, x)$ coordinates as
$\Sigma_{\pm} = \{ ( \pm e^{-x^2}, x ) \}_{x \in \mathbb{R}}$,
then clearly there is no $\epsilon>0$ such that $\Sigma_{\pm \epsilon} \subset \mathcal{M}$.

The solution in this case is to find a new conformal embedding $\mathcal{M} \hookrightarrow \mathbb{M}^2$,
which `expands' $\mathcal{M}$ to contain these Cauchy surfaces.
We accomplish this using the following lemma.%

\begin{lemma}
  \label{lem:dilation}
  Let $\mathcal{M} \subset \mathbb{M}^2$ be an open, causally-convex neighbourhood of the $t=0$ Cauchy surface $\Sigma_0$,
  and let $\Sigma_{\pm}$ be a pair of Cauchy surfaces of $\mathcal{M}$ such that $\Sigma_- \prec \Sigma_0 \prec \Sigma_+$.
  Then there exists a $\mathsf{CLoc}$ morphism $\chi: \mathcal{M} \rightarrow \mathbb{M}^2$,
  such that $\chi(\Sigma_0) = \Sigma_0$ and $\chi(\Sigma_-) \prec \Sigma_{-1} \prec \Sigma_1 \prec \chi(\Sigma_+)$.
\end{lemma}

\begin{remark}
  The relation $\prec$ between subsets of a spacetime $\mathcal{M}$
  is defined such that $U \prec V$ if and only if there is no future-directed causal curve in $\mathcal{M}$ from
  $v \in V$ to $u \in U$, in which case we say that $U$ \emph{is no later than} V.
  For Cauchy surfaces $\Sigma \prec \Sigma'$, this necessarily implies that
  every point $\Sigma$ is either in the past, or is spacelike separated from, every point in $\Sigma'$.
\end{remark}

\begin{proof}
  For the purposes of this proof, it is convenient to work in
  $(t, x)$ coordinates rather than our usual null $(u, v)$ system.
  Given this, we express each Cauchy surface as $\Sigma_{\pm} = \{(\pm t_{\pm}(x), x)\}_{x \in \mathbb{R}}$
  for a pair of smooth functions $t_{\pm}$ satisfying
  $t_{\pm}(x) > 0, |t'_{\pm}(x)| < 1 \forall x \in \mathbb{R}$.

  The goal is to find some $\rho \in \mathrm{Diff}_+(\Sigma_0)$ such that its extension to an embedding
  $\chi: \mathcal{M} \rightarrow \mathbb{M}^2$ takes $\Sigma_+$ to the future of $\Sigma_1$ and $\Sigma_-$ to the past of $\Sigma_{-1}$.
  If we consider only $\Sigma_+$ (i.e.\ we look at the special case where $t_+ = t_-$),
  then we need to show that $t > 1$ for every $(t, x) \in \chi(\Sigma_+)$.
  Using~\eqref{eq:cauchy-extension}, and noting that in our case $\pi_{\ell / r}^{\Sigma_0}(t, x) = (0, x \pm t)$,
  we can formulate the equivalent condition for $\rho$ as
  \begin{align}
    \label{eq:dilation-cond}
    \rho(x + t_+(x)) - \rho(x - t_+(x)) > 2.
  \end{align}

  One choice of $\rho$ that satisfies this inequality is
  \begin{align}
    \rho(x) := 2 \int_{x' = 0}^x \frac{\mathrm{d} x'}{t_+(x')}.
  \end{align}
  Note that $\rho$ is well defined as we always have that $t_+(x) > 0$,
  and from $\rho'(x) = \frac{2}{t_{+}(x)}$, we see that $\rho \in \mathrm{Diff}_+(\Sigma_0)$.
  To see that $\rho$ satisfies the inequality~\eqref{eq:dilation-cond},
  we substitute our choice into the above expression to find
  \begin{align}
    \rho(x + t_+(x)) - \rho(x - t_+(x))
    = 2 \int_{x - t_+(x)}^{x + t_+(x)} \frac{\mathrm{d} x'}{t_+(x')}
    = \frac{4 t_+(x)}{t_+(x + c)},
  \end{align}
  for some $c \in (-t_+(x), t_+(x))$, using the intermediate value theorem.
  Finally, using $|t'_+(x)| < 1$, we see that
  $t_+(x + c) \in (t_+(x) - c, t_+(x) + c) \subset (0, 2 t_+(x))$,
  hence $4 t_+(x) / t_+(x + c) > 2$ as required.

  In the case where $t_- \neq t_+$, one can use similar arguments to show that
  \begin{align}
    \rho(x) = 2 \int_{x' = 0}^x \left( \frac{1}{t_{+}(x)} + \frac{1}{t_{-}(x)} \right) \mathrm{d} x'
  \end{align}
  is an element of $\mathrm{Diff}_{+}(\mathbb{R})$ which satisfies both the necessary inequalities:
  $| \rho(x + t_{\pm}(x)) - \rho(x - t_{\pm}(x)) | > 2$.
\end{proof}

With that, we are finally able to establish how $\mathfrak{P}_{\ell}(\Sigma, \mathcal{M})$ embeds into
$\mathfrak{P}(\mathcal{M})$ in the general case.

\begin{theorem}
  \label{thm:classical-natural}
  For every $(\Sigma, \mathcal{M}) \in \mathsf{CCauchy}$,
  there exists an injective Poisson algebra homomorphism
  $\partial_{\Sigma, \epsilon}^{*}: \mathfrak{P}_{\ell}(\Sigma, \mathcal{M}) \rightarrow \mathfrak{P}(\mathcal{M})$.
  Moreover, by selecting such a map for each object in $\mathsf{CCauchy}$,
  we obtain a natural transformation
  \begin{align}
    \partial^{*}_{\epsilon}: \mathfrak{P}_{\ell} \Rightarrow \mathfrak{P}_{\mathrm{on}} \circ \Pi_2.
  \end{align}
\end{theorem}

\begin{proof}
  To first define the maps $\partial_{\Sigma, \epsilon}$,
  we must treat topologically planar spacetimes separately from cylindrical spacetimes.
  So let us first suppose that $\Sigma \simeq \mathbb{R}$.
  We can define a \emph{time function} on $\mathcal{M}$ such that $\Sigma = t^{-1}\{0\}$.
  Using this, we then specify a pair of Cauchy surfaces $\Sigma_{\pm} \subset \mathcal{M}$ by $t^{-1}\{\pm 1\}$.
  By taking a diffeomorphism $\rho_0: \Sigma \overset{\sim}{\rightarrow} \Sigma_0 \subset \mathbb{M}^2$,
  we then obtain a $\mathsf{CLoc}$ morphism $\mathcal{M} \rightarrow \mathbb{M}^2$ by~\cref{thm:cauchy-extension-skeletal}.
  Let $\Sigma_{t_0}$ denote the $t = t_0$ Cauchy surface in $\mathbb{M}^2$.
  If there exists some $\epsilon > 0$ such that $\mathcal{J}^+(\Sigma_{-\epsilon}) \cap \mathcal{J}^-(\Sigma_{\epsilon}) \subset \chi_0(\mathcal{M})$,
  then already we can define the map $\partial_{\Sigma_0, \epsilon} {\chi_0}_{*}$ by the argument preceding~\cref{lem:dilation},
  hence we can set $\partial_{\Sigma, \epsilon} := \rho^{*}_{(1)} \partial_{\Sigma_0, \epsilon} \chi_{*}: \mathfrak{E}(\mathcal{M}) \rightarrow \mathfrak{E}(\Sigma)$.

  If this is not the case, we simply apply~\cref{lem:dilation} to $\chi_0(\mathcal{M})$ with Cauchy surfaces $\chi_0(\Sigma_{\pm})$,
  to obtain a new embedding $\chi: \mathcal{M} \rightarrow \mathbb{M}^2$, whereupon we can use any value $\epsilon \in (0, 1)$.

  For $\Sigma \simeq S^1$, we proceed along similar lines,
  defining a pair of Cauchy surfaces $\Sigma_- \prec \Sigma \prec \Sigma_+$ and a diffeomorphism $\Sigma \overset{\sim}{\rightarrow} \Sigma_0 \subset \mathscr{E}$
  extending to a $\mathsf{CLoc}$ morphism $\chi: \mathcal{M} \rightarrow \mathscr{E}$.
  This time, however, due to the compactness of $S^1$ we can deduce that,
  if $\chi(\Sigma_{\pm}) = \{(\pm t_{\pm}(x), x)\}_{x \in \mathbb{R}}$,
  then both $t_\pm(x)$ are bounded from below, as $\Sigma \cap \Sigma_{\pm} = \emptyset$.
  In particular, there exists some $\epsilon > 0$ such that $\mathrm{Img}\, t_+ \cup \mathrm{Img}\, t_- \subset (\epsilon, \infty)$,
  hence $\chi(\Sigma_-) \prec \Sigma_{-\epsilon} \prec \Sigma_0 \prec \Sigma_{\epsilon} \prec \chi(\Sigma_+)$,
  and we may define $\partial_{\Sigma, \epsilon} := \rho^{*}_{(1)} \partial_{\Sigma_0, \epsilon} \chi_{*}: \mathfrak{E}(\mathcal{M}) \rightarrow \mathfrak{E}(\Sigma)$.

  Now that we have maps between configuration spaces,
  we must now show that they are Poisson algebra homomorphisms,
  and that they satisfy the desired naturality condition.
  The fact that $\partial_{\Sigma, \epsilon}^{*} F \in \mathfrak{F}_{\mu c}(\mathcal{M})$
  follows from the fact that $(\chi^{-1})_{*}: \mathfrak{F}_{\mu c}(\mathcal{M}) \rightarrow \mathfrak{F}_{\mu c}(\chi(\mathcal{M}))$
  To show the map is Poisson, by following similar arguments to~\cref{prop:cauchy-algebra-natural}
  it suffices to show that $(\partial_{\Sigma, \epsilon} \otimes \partial_{\Sigma, \epsilon}) E_{\mathcal{M}} = E_{\Sigma}$.
  This can be shown readily as
  \begin{align}
    (\partial_{\Sigma, \epsilon} \otimes \partial_{\Sigma, \epsilon}) E_{\mathcal{M}}
    &= (\rho^{*}_1 \otimes \rho^{*}_1)(\partial_{\Sigma_0, \epsilon} \otimes \partial_{\Sigma_0, \epsilon}) E_{\chi(\mathcal{M})} \nonumber \\
    &= (\rho^{*}_1 \otimes \rho^{*}_1)(\partial_{\Sigma_0, \epsilon} \otimes \partial_{\Sigma_0, \epsilon}) E_{\mathcal{M}_0} \nonumber \\
    &= (\rho^{*}_1 \otimes \rho^{*}_1)E_{\Sigma_0} \nonumber \\
    &= E_{\Sigma}.
  \end{align}
  where in the second line we have used the fact that
  the support of $\partial_{\Sigma_0, \epsilon}$ is within the image of $\chi(\mathcal{M})$,
  where $E_{\chi(\mathcal{M})}$ coincides with $E_{\mathcal{M}_0}$.

  Finally we consider the naturality.
  Let $(\rho, \chi): (\Sigma, \mathcal{M}) \rightarrow (\widetilde{\Sigma}, \widetilde{\mathcal{M}})$
  be a $\mathsf{CCauchy}$ morphism.
  Suppose that we have constructed $\partial_{\Sigma, \epsilon}$ and $\partial_{\widetilde{\Sigma}, \widetilde{\epsilon}}$.
  For simplicity we shall also use $\partial_{\Sigma, \epsilon}^{*}$ to denote the map
  $\mathfrak{P}_{\ell}(\Sigma, \mathcal{M}) \rightarrow \mathfrak{P}_{\mathrm{on}}(\mathcal{M})$.
  To show that
  \begin{equation}
    \begin{tikzcd}
      \mathfrak{P}_{\ell}(\Sigma, \mathcal{M})
        \ar[r, "\partial_{\Sigma, \epsilon}^{*}"]
        \ar[d, "\mathfrak{P}_{\ell}(\rho{,} \chi)"]
        &
      \mathfrak{P}_{\mathrm{on}}(\mathcal{M})
        \ar[d, "\mathfrak{P}_{\mathrm{on}} \chi"]
        \\
      \mathfrak{P}_{\ell}(\widetilde{\Sigma}, \widetilde{\mathcal{M}})
        \ar[r, "\partial_{\widetilde{\Sigma}, \widetilde{\epsilon}}^{*}"]
        &
      \mathfrak{P}_{\mathrm{on}}(\widetilde{\mathcal{M}})
    \end{tikzcd}
  \end{equation}
  commutes, we need only show that
  $F[\partial_{\Sigma, \epsilon} \chi^{*} \phi] = F[\rho^{*}_{(1)} \partial_{\widetilde{\Sigma}, \widetilde{\epsilon}} \phi]$,
  for every $F \in \mathfrak{P}_{\ell}(\Sigma, \mathcal{M})$ and $\phi \in \mathrm{Ker} \, P_{\widetilde{\mathcal{M}}}$.
  Immediately we have that $\partial_{\widetilde{\Sigma}, \widetilde{\epsilon}} \phi = \partial_{\widetilde{\Sigma}} \phi$,
  and also $\chi^{*}: \mathrm{Ker}\, P_{\widetilde{\mathcal{M}}} \rightarrow \mathrm{Ker}\, P_{\mathcal{M}}$,
  hence $\partial_{\Sigma, \epsilon} \chi^{*} \phi = \partial_{\Sigma} \chi^{*} \phi$.
  Thus, both functionals are equal to $F[\partial_{\Sigma} \chi^{*} \phi]$ and the diagram commutes.
\end{proof}

\begin{remark}
  Even though we also include the full spacetime $\mathcal{M}$ in our definition of the algebra
  $\mathfrak{P}_{\ell}(\Sigma, \mathcal{M})$, it is important to note that this is mostly a bookkeeping device.
  For instance, if $\widetilde{\Sigma} = \Sigma$,
  and $\chi$ is simply the inclusion for $\mathcal{M} \subset \widetilde{\mathcal{M}}$,
  then $\rho$ is the identity and $\mathfrak{P}_{\ell}(\rho, \chi)$ is the identity on $\mathfrak{F}_c(\Sigma)$.

  Thus we have found a theory of lower dimensionality which
  still embeds naturally into the full \emph{on-shell} spacetime algebra.
  Of course, we would not expect a sensible embedding into the off-shell algebra,
  as $\mathfrak{P}_{\ell}(\Sigma, \mathcal{M})$ is an algebra on the initial data for solutions in $\mathcal{M}$,
  hence it is intrinsically on-shell.

  Moreover, unlike the canonical algebra of Cauchy data which we discuss in the following section,
  this embedding is done without introducing any auxiliary field
  (namely the conjugate momentum in the canonical algebra).
This is because one is able to write a solution to the wave equation
  as a sum of two solutions to \emph{first order} \textsc{pde}s.
\end{remark}

\subsection{Comparison to Equal Time Commutation Relations}

An alternative way of deciding the form of the chiral bracket is
by comparison to the \emph{canonical} or \emph{equal-time} Poisson bracket.
In terms of integral kernels,
and with respect to a Cauchy surface $\Sigma \subset \mathcal{M}$,
this is typically written as
\begin{equation}
    \left\{ \Phi(\mathbf{x}), \Pi(\mathbf{y}) \right\}^\Sigma_{\mathrm{can}}
        =
    \delta_\Sigma(\mathbf{x}, \mathbf{y}),
\end{equation}
where $\mathbf{x}, \mathbf{y} \in \Sigma$
and $\delta_\Sigma \in \mathfrak{D}'(\Sigma^2)$
is the Dirac delta with support on the diagonal of $\Sigma^2$.

To make this more precise,
we define the space of Cauchy data on $\Sigma$ to be
$\mathscr{C}(\Sigma) = \mathfrak{D}(\Sigma, \mathbb{R}^2)$,
i.e.\ pairs of smooth functions on $\Sigma$ with compact support.
For each $f \in \mathfrak{D}(\Sigma)$,
we can then denote the regular linear observables on $\mathscr{C}(\Sigma)$ by
\begin{align}
    \label{eq:linear-can-obs}
    \Phi(f)[\psi, \pi] = \int_\Sigma f \psi \, \mathrm{d}V_\Sigma,
        \qquad
    \Pi(f) [\psi, \pi] = \int_\Sigma f \pi     \, \mathrm{d}V_\Sigma.
\end{align}
We define their canonical Poisson bracket as
\begin{equation}
    \left\{ \Phi(f), \Pi(g) \right\}^\Sigma_\mathrm{can}
        =
    \int_\Sigma fg \, \mathrm{d}V_\Sigma.
\end{equation}

For the massless scalar field,
if $(\varphi, \pi)$ are the Cauchy data for a solution $\phi$,
then $\varphi = \phi|_\Sigma$ and $\pi = \dot{\phi}|_\Sigma$,
where $\dot{\phi}$ is the derivative of $\phi$ along the
future-directed normal vector to $\Sigma$.
This suggests the map
$\rho_\pm : \mathscr{C}(\Sigma) \to \mathfrak{E}(\Sigma)$,
defined by
$\rho_\pm(\varphi, \pi) = \frac{1}{2}(\pi \mp *d\phi)$,
sends Cauchy data to the associated chiral configuration .

Considering the chiral boson $\Psi$ from \eqref{eq:chiral-boson},
the pullback of $\Psi_\Sigma(f)$ along $\rho_+$ is
\begin{equation*}
    (\rho_+^* \Psi_\Sigma(f))[\varphi, \pi]
        =
    \frac{1}{2} \int_\Sigma
        (f \pi - f *\!d\phi) \, \mathrm{d}V_\Sigma.
\end{equation*}
Noting that
$
    \int_\Sigma f (*d\phi) \mathrm{d}V_\Sigma =
    \int_\Sigma f d\phi =
    -\int_\Sigma \phi df =
    -\int_\Sigma \phi (*df) \mathrm{d}V_\Sigma,
$
we can write this pullback in terms of
the observables in \eqref{eq:linear-can-obs} as
\begin{equation}
    \rho^*_+ \Psi_\Sigma(f)
        =
    \frac{1}{2} \left( \Pi(f) + \Phi(*df) \right).
\end{equation}

We can then take the pullback of the canonical Poisson bracket along this map,
which we express as
\begin{align*}
    \left\{
        \Psi_\Sigma(f),
        \Psi_\Sigma(g)
    \right\}_\ell^\Sigma
        &=
    \left\{
            \rho_+^* \Psi_\Sigma(f),
            \rho_+^*\Psi_\Sigma(g)
    \right\}_\mathrm{can}^\Sigma \\
        &=
    \frac{1}{4}
    \left\{
        \Pi(f) + \Phi(*df),
        \Pi(g) + \Phi(*dg)
    \right\}_\mathrm{can}^\Sigma \\
        &=
    \frac{1}{2}
    \left\{
        \Phi(*df),
        \Pi(g)
    \right\}_\mathrm{can}^\Sigma \\
        &=
    -\frac{1}{2} \int_\Sigma f dg.
\end{align*}
This agrees exactly with~\cref{eq:geometric-chiral-bracket}.

\subsection{Chiral Primary Fields}
\label{sec:chiral-primary-fields}

Fields are, naturally, a central aspect of any approach to quantum field theory,
however the precise definition of what a field \emph{is} varies considerably.
Previously, in \cite{crawfordLorentzian2dCFT2021},
we gave a definition of fields rooted in the principle of local covariance,
where theories are functors $\mathsf{Loc} \rightarrow \mathsf{Obs}$ for some suitable category of observables,
and fields are natural transformations from the functor of test functions to the functor describing the theory.
By defining new spacetime categories which accounted for conformal isometries,
and suitably modifying the functor assigning to spaces their test functions,
we could say that a field was primary if the naturality condition held for this expanded set of morphisms.

In~\cite[Chapter~9]{schottenloherMathematicalIntroductionConformal2008},
Schottenloher provides a characterisation of primary fields in 2\textsc{d} Euclidean \textsc{cft} which we summarise below.
Firstly, similarly to Wightman \textsc{qft}s,
a field is defined as a tempered distribution over $\mathbb{C}$ with values in
the algebra of bounded operators on some Hilbert space,
i.e.\ a linear, continuous map $\Phi: \mathfrak{S}(\mathbb{C}) \rightarrow \mathfrak{B}(\mathcal{H})$.
The condition for such a field to be primary (with weight $\mu$)
can be \emph{formally} expressed as the condition that,
for every holomorphic map $z \mapsto w(z)$
\begin{align}
  \label{eq:schottenloher-primary}
  U(w)\Phi(z)U(w)^{-1} = \left( \frac{d w}{d z} \right)^{\mu} \Phi(w(z))\,,
\end{align}
where $U$ is a unitary representation of the holomorphic transformations on $\mathcal{H}$.
To obtain the precise definition, one considers the infinitesimal transformation $w(z) = z + \epsilon w_0(z)$,
for a holomorphic map $w_0$.
Differentiating each side with respect to $\epsilon$ and evaluating at $\epsilon = 0$ then generates the correct equation.
In particular, on the left, one obtains an action of holomorphic functions on $\mathfrak{B}(\mathcal{H})$ by derivations.
This is assumed (as one of the axioms) to be generated by brackets/commutators with the stress-energy tensor,
which we explore in~\cref{sec:chiral-set-generate}.
Recalling that $\Phi(x)$ is really the integral kernel of a distribution,
we can integrate both sides of~\eqref{eq:schottenloher-primary} with some $f \in \mathfrak{S}(\mathbb{C})$ to obtain
\begin{align}
  U(w)\Phi(f)U(w)^{-1} = \Phi(w_{*}^{(\mu - 1)} f),
\end{align}
where $\Phi(f) = \int_{\mathbb{C}} \Phi(z) f(z) \,\mathrm{d}z$ and
\begin{align}
  \label{eq:euclidean-pushforward}
  w_{*}^{(\mu - 1)}f := \left[\left( \frac{dw}{dz} \right)^{\mu - 1} \cdot f \right] \circ w^{-1}
\end{align}
is a map $\mathfrak{S}(\mathbb{C}) \rightarrow \mathfrak{S}(\mathbb{C})$.
(To make this more precise, for $w(z) = z + \epsilon w_0(z)$, we may take $w^{-1}(z) = z - \epsilon w_0(z)$
in order to generate the correct infinitesimal action.)

We can then use this relation to characterise primary fields as \emph{equivariant maps} between
two representations of the `conformal group' of holomorphic functions.
One representation acting on $\mathfrak{B}(\mathcal{H})$, the algebra of observables,
the other acting on $\mathfrak{S}(\mathbb{C})$, the space of test functions.


To make this precise in our framework, we begin by noting how the characterisation of primary fields as equivariant maps
is very close to the definition of locally covariant fields as natural transformations.

Recall that, given any object $c$ of a category $\mathcal{C}$,
there is a group $\mathrm{Aut}_{\mathcal{C}}(c)$ comprising the invertible morphisms $c \rightarrow c$.
If we restrict our attention to this subcategory,
a functor $F: \mathrm{Aut}_{\mathcal{C}}(c) \rightarrow \mathsf{Vec}$
is simply a representation of $\mathrm{Aut}_{\mathcal{C}}(c)$ on the space $F(c)$,
and a natural transformation $F \Rightarrow G$ between two functors is a single map $F(c) \rightarrow G(c)$ which is
equivariant with respect to the two representations of $\mathrm{Aut}_{\mathcal{C}}(c)$.
Hence if we consider just one framed spacetime $\mathscr{M} = (M, (e^{\ell}, e^r)) \in \mathsf{CFLoc}$
as in~\cite{crawfordLorentzian2dCFT2021},
and take the group $\mathrm{Aut}(\mathscr{M})$ of
conformally admissible diffeomorphisms $\mathscr{M} \rightarrow \mathscr{M}$,
then $\mathfrak{D}^{(\mu, \widetilde{\mu})}$ defines a representation of
$\mathrm{Aut}(\mathscr{M})$ on $\mathfrak{D}(M)$,
any functor $\mathfrak{A}: \mathsf{CFLoc} \rightarrow \mathsf{Obs}$
defines a representation of $\mathrm{Aut}(\mathscr{M})$ on $\mathfrak{A}(\mathscr{M})$,
and a primary field of weight $(\mu, \widetilde{\mu})$ defines an equivariant map between the two.

Notably, we had to go from $\mathsf{CLoc}$ to a new category,
$\mathsf{CFLoc}$ which assigns additional data to each spacetime in the form of a global frame.
This is so that tensor fields such as the stress-energy tensor may be separated into scalar components.
As we do not wish for the theory to depend on this additional data,
we have to assume that the functor $\mathsf{CFLoc} \rightarrow \mathsf{Obs}$
factorised into the composition of the surjective functor $\mathfrak{p}: \mathsf{CFLoc} \rightarrow \mathsf{CLoc}$,
which forgets about the additional frame data,
and a functor $\mathsf{CLoc} \rightarrow \mathsf{Obs}$ defining the theory.

In a similar vein, a full definition of a chiral primary field should involve a similar assignment of frame data by the construction of a new category
$\mathsf{CFCauchy}$ possessing surjective functors to both $\mathsf{CCauchy}$ and $\mathsf{CFLoc}$.
We would then need a functor $\mathfrak{D}^{(\mu)}_{\ell}: \mathsf{CFCauchy} \rightarrow \mathsf{Vec}$
which assigns to each object the space of test functions of the appropriate Cauchy surface,
and whose morphisms implement the weighted pushforward analogous to~\eqref{eq:euclidean-pushforward}.
A classical chiral primary field of weight $\mu$ would then be a natural transformation from $\mathfrak{D}_\ell^{(\mu)}$
to the lift of $\mathfrak{P}_{\ell}$ to $\mathsf{CFCauchy}$.
(A quantum chiral primary field would be defined in much the same manner,
using the functor $\mathfrak{A}_{\ell}$ from~\cref{def:quantum-chiral-algebra}.)

Whilst it is currently unclear what extra data a general object of $\mathsf{CFCauchy}$ should possess,
it is natural to assume that this data can be chosen canonically for objects arising from Minkowski space, so we focus on that first. 
More precisely, we consider the subcategory of $\mathsf{CFLoc}$
comprising open, causally convex subsets of $\mathbb{M}^2$,
moreover, the restriction of the surjective functor $\mathfrak{p}: \mathsf{CFLoc} \rightarrow \mathsf{CLoc}$
to this subcategory is also injective, hence the subcategory is isomorphic to its image.

Analogously, there must exist a wide subcategory of $\mathsf{CFCauchy}$ comprising pairs
$(\Sigma, U)$ where $\Sigma \subset U \subseteq \mathbb{M}^2$ is a Cauchy surface of an open, causally convex subset of $\mathbb{M}^2$.
which is isomorphic to the subcategory of $\mathsf{CCauchy}$ comprising the same objects and all the morphisms between them.

Let us denote this subcategory $\mathsf{Cauchy}(\mathbb{M}^2)$.
We can implement the weighted pushforwards as a family of functors
$\mathfrak{D}^{(\mu)}_{\ell}: \mathsf{Cauchy}(\mathbb{M}^2) \rightarrow \mathsf{TVec}$ such that
$\mathfrak{D}^{(\mu)}_{\ell}(\Sigma, U) = \mathfrak{D}(\Sigma)$,
and for a morphism $(\rho, \chi): (\Sigma, U) \rightarrow (\widetilde{\Sigma}, \widetilde{U})$
$\mathfrak{D}^{(\mu)}_{\ell}(\rho, \chi) (f) := \rho_{*}(\omega_{\ell}|_{\Sigma}^{\mu - 1} f)$,
where $\chi^{*} du = \omega_{\ell} du$.
Lastly, we define the functor $\mathfrak{U}: \mathsf{Cauchy}(\mathbb{M}^2) \rightarrow \mathsf{CFLoc}$
by $\mathfrak{U}(\Sigma, U) = U$, equipped with the canonical frame it inherits from $\mathbb{M}^2$
(which should be interpreted as the restriction of a forgetful functor $\mathsf{CFCauchy} \rightarrow \mathsf{CFLoc}$).
We can then relate our weighted representations using the following result

\begin{proposition}
  \label{prop:partial-naturality}
  For an object $(\Sigma, U) \in \mathsf{Cauchy}(\mathbb{M}^2)$,
  define the map $\eta_{(\Sigma, U)}: \mathfrak{D}(U) \rightarrow \mathfrak{D}(\Sigma)$ by
  \begin{align}
    (\eta_{(\Sigma, U)} h)(s) := \int_{\mathbb{R}} h(-s, v) \mathrm{d}v,
  \end{align}
  where the coordinate $s$ on $\Sigma$ is obtained as the restriction of the map $\mathbb{M}^2 \rightarrow \mathbb{R}; (u, v) \mapsto -u$.
  Then $\eta$ defines a natural transformation $\mathfrak{D}^{(\mu, 0)} \circ \mathfrak{U} \Rightarrow \mathfrak{D}^{(\mu)}_{\ell}$,
  i.e.\ for every morphism $(\rho, \chi): (\Sigma, U) \rightarrow (\widetilde{\Sigma}, \widetilde{U})$, the following diagram commutes.
  \begin{equation}\label{eq: weighted diagram commutes}
    \begin{tikzcd}
      \mathfrak{D}(U)
        \ar[r, "\eta_{(\Sigma, U)}"] \ar[d, "\mathfrak{D}^{(\mu, 0)}\chi"]
        &
      \mathfrak{D}(\Sigma)
        \ar[d, "\mathfrak{D}^{(\mu)}({\rho, \chi})"]
        \\
      \mathfrak{D}(\widetilde{U})
        \ar[r, "\eta_{(\widetilde{\Sigma}, \widetilde{U})}"]
        &
      \mathfrak{D}(\widetilde{\Sigma})
    \end{tikzcd}
  \end{equation}
\end{proposition}

\begin{proof}
  Suppose that $\Sigma$ is sent to the interval $\mathcal{I} \subseteq \mathbb{R}$ under the map $(u, v) \mapsto -u$.
  We then express $\Sigma$ as the graph $\{(-s, \gamma(s)) \in U\}_{s \in \mathcal{I}}$ for some smooth embedding
  $\gamma: \mathcal{I} \hookrightarrow \mathbb{R}$.
  We similarly express $\widetilde{\Sigma}$ in terms of some
  $\widetilde{\gamma}: \widetilde{\mathcal{I}} \hookrightarrow \mathbb{R}$.
  Using~\eqref{eq:cauchy-extension}, we write $\chi$ explicitly in null-coordinates as
  \begin{align}
    \chi(u, v) = (-\rho(-u), \widetilde{\gamma} \rho \gamma^{-1}(v)),
  \end{align}
  where $\gamma^{-1}: \gamma(\mathcal{I}) \rightarrow \mathcal{I}$, and we have identified $\rho$ with a smooth embedding
  $\mathcal{I} \hookrightarrow \widetilde{\mathcal{I}}$ using our coordinate system.
  From this, we calculate the conformal factors of $\chi$ as
  \begin{align}
    \label{eq:conformal-factors-chart}
    \omega_\ell(u, v) = \omega_{\ell}(u) = \rho'(-u), \qquad
    \omega_{r}(u, v) = \omega_r(v) = \frac{\widetilde{\gamma}'(\rho\gamma^{-1}(v))}{\gamma'(\gamma^{-1}(v))} \rho'(\gamma^{-1}(v)).
  \end{align}
  In particular this means that $\omega_{\ell}(-s, v) = \omega_{\ell}|_{\Sigma}(s) = \rho'(s)$.

  With our conformal factors suitably equated,
  we may now compute each path around the above diagram.
  Let $h \in \mathfrak{D}(\mathbb{M}^2)$, then
  \begin{align}
    \begin{split}
    (\rho_{*}^{(1 - \mu)} \circ \eta (h)) (s) &= \omega_{\ell}|_{\Sigma}^{\mu - 1}(s) \int_{\mathbb{R}} h(- \rho^{-1}(s), v) \,\mathrm{d}v \\
    (\eta \circ \chi_{*}^{(1-\mu, 1)} (h)) (s) &= \\
    \int_{\mathbb{R}}
      \omega_{\ell}(-\rho^{-1}(s))^{\mu - 1} &\omega_r(\gamma \rho^{-1} \widetilde{\gamma}^{-1}(v'))^{-1}
      h(-\rho^{-1}(s), \gamma \rho^{-1} \widetilde{\gamma}^{-1}(v')) \,\mathrm{d}v'.
    \end{split}
  \end{align}
  We then make the change of variables for the second integral $v = \gamma \rho \widetilde{\gamma}^{-1} (v')$,
  noting that then $dv = \omega_r(v)^{-1}dv'$ hence the two integrals are in fact equal.
\end{proof}

\begin{remark}
  What made the diagram \eqref{eq: weighted diagram commutes} commute were the facts that, firstly $\omega_{\ell}$ depended only on $u$,
  and hence could be taken out of the integral and secondly that $\omega_r$ depended only on $v$
  and naturally accounted for the change of variables necessary to relate the integrals.
  Given that these two facts are not true for a general morphism of $\mathsf{CFLoc}$,
  this argument cannot be readily generalised.
\end{remark}

Given that $\mathsf{Cauchy}(\mathbb{M}^2)$ is a subcategory of $\mathsf{CCauchy}$,
we already have a classical theory defined on it,
$\mathfrak{P}_{\ell}|_{\mathbb{M}^2}: \mathsf{Cauchy}(\mathbb{M}^2) \rightarrow \mathsf{Poi}$.
Thus, for any natural transformation $\Psi: \mathfrak{D}^{(\mu)}_{\ell} \Rightarrow \mathfrak{P}_{\ell}|_{\mathbb{M}^2}$,
the maps $\eta$ from the~\cref{prop:partial-naturality}, along with the natural transformation
$\partial^{*}_{\epsilon}: \mathfrak{P}_{\ell} \Rightarrow \mathfrak{P}_{\mathrm{on}} \circ \Pi_2$ from~\cref{thm:classical-natural}
can be composed horizontally with $\Psi$ to define a natural transformation
\begin{align}
  \partial^{*}_{\epsilon} \circ \Psi \circ \eta: \mathfrak{D}^{(\mu, 0)} \circ \mathfrak{U} \Rightarrow \mathfrak{P}_{\mathrm{on}} \circ \mathfrak{p} \circ \mathfrak{U}.
\end{align}
Unpacking the definition, this means that, for every conformally admissible embedding $\chi: U \longrightarrow \widetilde{U}$
which restricts to a map $\rho: \Sigma \rightarrow \widetilde{\Sigma}$ between Cauchy surfaces,
the following diagram commutes.
\begin{equation}
  \label{eq:primary-field-correspondence}
  \begin{tikzcd}
    \mathfrak{D}(U)
      \ar[r, "\eta_{(\Sigma, U)}"] \ar[d, "\mathfrak{D}^{(\mu, 0)}\chi"]
      &
    \mathfrak{D}(\Sigma)
      \ar[r, "\Psi_{(\Sigma, U)}"] \ar[d, "\mathfrak{D}_{\ell}^{(\mu)}({\rho, \chi})"]
      &
    \mathfrak{P}_{\ell}(\Sigma, U)
      \ar[r, "\partial_{\Sigma, \epsilon}^{*}"] \ar[d, "\mathfrak{P}_{\ell}({\rho, \chi})"]
      &
    \mathfrak{P}_{\mathrm{on}}(U)
      \ar[d, "\mathfrak{P}_{\mathrm{on}}\chi"]
      \\
    \mathfrak{D}(\widetilde{U})
      \ar[r, "\eta_{(\widetilde{\Sigma}, \widetilde{U})}"]
      &
    \mathfrak{D}(\widetilde{\Sigma})
      \ar[r, "\Psi_{(\widetilde{\Sigma}, \widetilde{U})}"]
      &
    \mathfrak{P}_{\ell}(\widetilde{\Sigma}, \widetilde{U})
      \ar[r, "\partial_{\widetilde{\Sigma}, \widetilde{\epsilon}}^{*}"]
      &
    \mathfrak{P}_{\mathrm{on}}(\widetilde{U})
  \end{tikzcd}
\end{equation}

If one can then further show that, for any pair of Cauchy surfaces $\Sigma, \widetilde{\Sigma} \subset U$,
\begin{align}
  \label{eq:chiral-primary-independence}
  \partial_{\Sigma, \epsilon}^{*} \circ \Psi_{(\Sigma, U)} \circ \eta_{(\Sigma, U)} =
  \partial_{\widetilde{\Sigma}, \widetilde{\epsilon}}^{*} \circ
  \Psi_{(\widetilde{\Sigma}, U)} \circ
  \eta_{(\widetilde{\Sigma}, U)},
\end{align}
then it is clear that we in fact have a natural transformation
$\mathfrak{D}^{(\mu, 0)} \Rightarrow \mathfrak{P}_{\mathrm{on}} \circ \mathfrak{p}$,
i.e.\ a primary field of weight $(\mu, 0)$ in the sense of~\cite[\S 4.2]{crawfordLorentzian2dCFT2021}.

As an example where this is the case,
we can consider the \emph{chiral boson}
\begin{align}
  \Psi_{(\Sigma, U)}(f)[\psi] := \int_{\Sigma} f \psi \mathrm{d}V_{\Sigma}.
\end{align}
As we are working on-shell, \eqref{eq:chiral-primary-independence} is satisfied if,
$\forall \phi \in \mathrm{Ker}\, P_{\mathbb{M}^2}, h \in \mathfrak{D}(U)$
\begin{align}
  \Psi_{(\Sigma, U)}(\eta_{(\Sigma, U)} h)[\partial_{\Sigma, \epsilon} \phi]
  =
  \Psi_{(\widetilde{\Sigma}, U)}(\eta_{(\widetilde{\Sigma}, U)} h)[\partial_{\widetilde{\Sigma}, \widetilde{\epsilon}} \phi].
\end{align}
By expanding out each definition, one can eventually show that these maps indeed coincide,
and that in fact they are equal to $\partial\Phi_{\mathbb{M}^2}(h)[\phi]$,
the null derivative field from \cite[Example 4.2]{crawfordLorentzian2dCFT2021}.

With this, we formulate our definition of a chiral primary fields on Minkowski spacetime.

\begin{definition}
  Let $\mathfrak{A}_{\ell}: \mathsf{Cauchy}(\mathbb{M}^2) \rightarrow \mathsf{TVec}$
  be a functor describing some locally covariant theory.
  A \emph{chiral primary field of weight $\mu$ on $\mathbb{M}^2$ with values in} $\mathfrak{A}_{\ell}$
  is then defined as a natural transformation $\mathfrak{D}^{(\mu)}_{\ell} \Rightarrow \mathfrak{A}_{\ell}$.
\end{definition}

Matching \cite{crawfordLorentzian2dCFT2021}, we find a family of examples meeting this definition
by taking powers of the chiral boson.

\begin{example}
  \label{ex:monomial-fields}
  For $n \in \mathbb{N}$, and $(\Sigma, U) \in \mathsf{Cauchy}(\mathbb{M}^2)$ the maps
  $\Psi^n: \mathfrak{D}(\Sigma) \rightarrow \mathfrak{P}_{\ell}(\Sigma, U)$ defined by
  \begin{align}
    \Psi^n_{(\Sigma, U)}(f)[\psi] := \int_{\mathcal{I}} f(s) \psi^n(s) \gamma'(s)^{\tfrac{n}{2}} \mathrm{d}s
  \end{align}
  constitute a chiral primary field of weight $n$.
  Moreover $\partial^{*} \circ \Psi^n \circ \eta = \Pi_{\mathrm{on}} \circ \partial \Phi^n$.

  To see this we must show that, for every commuting square
  \begin{equation}
    \begin{tikzcd}
        \Sigma
          \ar[r, "\rho"]
          \ar[d]
          &
        \widetilde{\Sigma}
          \ar[d]
          &
        \\
        U
          \ar[r, "\chi"]
          &
        \widetilde{U}
    \end{tikzcd}
  \end{equation}
  and every pair $f \in \mathfrak{D}(\Sigma)$, $\psi \in \mathfrak{E}(\widetilde{\Sigma})$
  \begin{align}
    \Psi^n_{(\Sigma, U)}(f)[\rho^{*}_{(1)} \psi] =
    \Psi^n_{(\widetilde{\Sigma}, \widetilde{U})}(\rho_{*} (\omega_{\ell}|_{\Sigma}^{\mu - 1} f))[\psi].
  \end{align}

  Using the same coordinate system as above,
  we can express the conformal factors in terms of \eqref{eq:conformal-factors-chart} and hence
  write each side of this equation explicitly as
  \begin{align}
    \Psi^n_{(\Sigma, U)}(f)[\rho^{*}_{(1)} \psi]
    &=
    \int_{\mathcal{I}} f(s) (\rho^{*} \psi)^n (s) \left(\frac{\widetilde{\gamma}'(\rho(s))}{\gamma'(s)}\right)^{\tfrac{n}{2}} \rho'(s)^n \gamma'(s)^{\tfrac{n}{2}} \mathrm{d}s, \\
    \Psi^n_{(\widetilde{\Sigma}, \widetilde{U})}(\rho_{*} (\omega_{\ell}|_{\Sigma}^{n - 1} f) )[\psi]
    &=
    \int_{\widetilde{\mathcal{I}}} \rho_{*} ((\rho')^{n - 1} f) (\tilde{s}) \psi^n (\tilde{s}) \widetilde{\gamma}'(\tilde{s})^{\tfrac{n}{2}} \mathrm{d}\tilde{s},
  \end{align}
  from which we can clearly see the two expressions coincide.

  Finally, if we take $\phi \in \mathrm{Ker} P_U$, then
  \begin{align}
    \partial \Phi^n_{U} (f) [\phi]
    &=
    \int_U f(u, v) (\partial_u \phi)^n(u) \, \mathrm{d}u \mathrm{d}v \nonumber \\
    &=
    \int_{\mathcal{I}} \left( \int_{\mathbb{R}} f(u, v) \, \mathrm{d}v \right) (\partial_u \phi)^n(u) \, \mathrm{d}u \nonumber \\
    &=
    \Psi^n (\eta f) [\partial_{\Sigma} \phi]
  \end{align}
  where we have used the facts that
  $\partial_{\Sigma}\phi (s) = \frac{1}{\sqrt{\gamma'(s)}} (\partial_u \phi)(-s, \gamma(s))$,
  and $\mathrm{supp} \, (\eta f) \subseteq \mathcal{I}$.
  This demonstrates that $\partial \Phi^n$ and $\partial^{*} \circ \Psi^n \circ \eta$
  define the same on-shell observable as required.
\end{example}

However, we shall usually be able to work with a weaker property,
which is satisfied, for example, by the quantum stress energy tensor
as well as fields identified as \emph{quasiprimary} in the \textsc{cft} literature.
For these fields, it is often enough to require that they transform naturally
with respect to boosts and dilations.
We make this notion precise with the following definition.

\begin{definition}
  Define the subcategory $\mathsf{Cauchy}(\mathbb{M}^2)_0$ of $\mathsf{Cauchy}(\mathbb{M}^2)$
  which contains the same objects, but only those morphisms for which $\omega_{\ell}$ is constant.
  Let $\mathfrak{A}_{\ell}: \mathsf{Cauchy}(\mathbb{M}^2) \rightarrow \mathsf{TVec}$ as before.
  A \emph{homogeneously scaling locally covariant field of weight $\mu$ on
    $\mathbb{M}^2$ with values in} $\mathfrak{A}_{\ell}$
  is a natural transformation
  $
    \mathfrak{D}_{\ell}^{(\mu)}|_{\mathsf{Cauchy}(\mathbb{M}^2)_0}
    \Rightarrow \mathfrak{A}_{\ell}|_{\mathsf{Cauchy}(\mathbb{M}^2)_0}
  $
\end{definition}

\begin{remark}
  A difference between locally covariant fields in the sense we have been using and
  natural Lagrangians (as defined in, for example \cite[Definition 4.6]{rejznerPerturbativeAlgebraicQuantum2016}) is that the former
  are assumed to be both linear and continuous in $\mathfrak{D}(\mathcal{M})$%
  \footnote{
    In the latter case these assumptions are typically relaxed to the weaker notion of
    \emph{additivity}, as linearity is not preserved by the action of the renormalisation group.
  }.
  As we shall see in the following section,
  these additional properties will allow us to apply our algebraic operations,
  such as Poisson brackets, $\star$-products and commutators,
  to fields in order to produce ordinary, $\mathbb{C}$-valued distribution.
  The transformation properties of the fields will then descend to the level of these distributions,
  allowing us to impose tight constraints.
\end{remark}

To conclude this section,
we demonstrate a property of fields satisfying the above definition that shall be useful in later proofs.

\begin{lemma}
  \label{lem:hom-scale-support}
  Let
  $\Psi: \mathfrak{D}_{\ell}^{(\mu)}|_{\mathsf{Cauchy}(\mathbb{M})_0}
    \Rightarrow \mathfrak{P}_{\ell}|_{\mathsf{Cauchy}(\mathbb{M}^2)_0}$
  be a locally covariant field on $\mathbb{M}^2$ with values in
  $\mathfrak{P}_{\ell}|_{\mathsf{Cauchy}(\mathbb{M}^2)_0}$,
  then $\forall f \in \mathfrak{D}(\Sigma)$,
  $\mathrm{supp}\, \Psi_{(\Sigma, U)}(f) \subseteq \mathrm{supp}\, f$.
\end{lemma}

\begin{proof}
  Due to the linearity of $\Psi_{(\Sigma, U)}$, we may assume that $\mathrm{supp}\, f$ is connected,
  (otherwise $f$ is a finite sum of $f_i \in \mathfrak{D}(\Sigma)$ which have connected supports).
  Let $\Sigma_f \subseteq \Sigma$ be any open (in $\Sigma$), connected neighbourhood of $\mathrm{supp}\, f$,
  then there exists an open, causally-convex neighbourhood $\Sigma_f \subset U_f \subseteq U$
  (which can be seen by taking the intersection of
  any open, causally convex neighbourhood of $\Sigma_f$ with $U$).
  Let us denote by $(i, i_U)$ the inclusion morphism $(\Sigma_f, U_f) \rightarrow (\Sigma, U)$,
  then we can clearly see that $\mathfrak{D}^{(\mu)}_{\ell}(i, i_U): \mathfrak{D}(\Sigma_f) \rightarrow \mathfrak{D}(\Sigma)$
  is simply the pushforwards along the inclusion $i$,
  hence $f = \mathfrak{D}^{(\mu)} (i, i_U) (f|_{\Sigma}) = i_{*} f|_{\Sigma}$.

  Using the naturality of $\Psi$, we may then write that
  \begin{align}
    \Psi_{(\Sigma, U)}(f)
    &=
    \Psi_{(\Sigma, U)}\left( \mathfrak{D}^{(\mu)} (i, i_U) (f|_{\Sigma}) \right) \nonumber \\
    &=
    \mathfrak{P}_{\ell}(i, i_U) \left( \Psi_{(\Sigma_f, U_f)} (f|_{\Sigma}) \right).
  \end{align}
  From the definition of the morphism $\mathfrak{P}_{\ell}(i, i_U)$,
  it is then clear that $\mathrm{supp}\, \Psi_{(\Sigma, U)}(f) \subseteq \Sigma_f$.
  Given that this holds for any open, connected neighbourhood of $\mathrm{supp}\, f$,
  we must conclude that $\mathrm{supp}\, \Psi_{(\Sigma, U)}(f) \subseteq \mathrm{supp}\, f$.
\end{proof}

In the case of the cylinder,
recall that we can identify $\mathfrak{E}(\mathscr{E}) \simeq \mathfrak{E}(\mathbb{M}^2)^{\mathbb{Z}}$,
where the action of $\mathbb{Z}$ on $\mathbb{M}^2$ is given by $n \cdot (u, v) = (u + 2 \pi, v - 2 \pi)$.
Similarly, $\pi_{\ell}(\mathscr{E}) \simeq \mathbb{R} / 2 \pi \mathbb{Z}$, hence $\mathfrak{D}(\pi_{\ell}(\mathscr{E})) \simeq \mathfrak{E}(\mathbb{R})^{2 \pi \mathbb{Z}}$.
Let $h \in \mathfrak{D}(\mathscr{E})$, and denote by $h_{\mathbb{M}}$ the corresponding element of $\mathfrak{E}(\mathbb{M}^2)^{\mathbb{Z}}$.
Clearly $v \mapsto h_{\mathbb{M}}(u, v)$ is compactly supported for every $u \in \mathbb{R}$,
hence $h \mapsto \int_{\mathbb{R}} h_{\mathbb{M}}(u, v) \,\mathrm{d}v$ is a well-defined map
$\mathfrak{D}(\mathscr{E}) \rightarrow \mathfrak{E}(\mathbb{R})$.
In particular, the image of this map contains only $2 \pi \mathbb{Z}$ invariants,
hence we have a map $\int_{\mathbb{R}} \mathrm{d}v: \mathfrak{D}(\mathscr{E}) \rightarrow \mathfrak{D}(\pi_{\ell}(\mathscr{E}))$.

If $\mathrm{supp}\, h \subset U$ for some open, causally convex subset $U \subseteq \mathscr{E}$,
then we can also see that the support of $\int_{\mathbb{R}}\mathrm{d}v\, (h)$ is contained within $\pi_{\ell}(U)$.
Given that $\mathfrak{D}(\Sigma) \simeq \mathfrak{D}(\pi_{\ell}(U))$ for any Cauchy surface $\Sigma$ of $U$,
this means we can define a map $\eta_{(\Sigma, U)}: \mathfrak{D}(U) \rightarrow \mathfrak{D}(\Sigma)$.
One can then carry out a similar procedure to above to confirm that these maps constitute a natural transformation.

Considering diagram \eqref{eq:primary-field-correspondence},
it is clear that the limiting factor for establishing a correspondence between
chiral primary fields and primary fields in the sense of \cite[\S 4.2]{crawfordLorentzian2dCFT2021}
is in establishing the natural transformation $\eta$.
Given that we now have this natural transformation for
Minkowski space and the Einstein cylinder,
it follows from the results of \cite{beniniSkeletalModel2d2021}
that the correspondence should hold in general.

\subsection{General Form of Chiral Brackets}
\label{sec:constraints}

We have now seen several ways the Poisson bi-vector of the chiral algebra may be obtained.
In this section, we see that some generic assumptions about a conformal field theory
can yield tight constraints on the Poisson structure.

Other than the conformal covariance property,
the key feature we employ is \emph{Einstein causality}.
A classical theory $\mathfrak{P}: \mathsf{Loc} \rightarrow \mathsf{Poi}$
(or a quantum theory $\mathfrak{A}: \mathsf{Loc} \rightarrow \mathsf{Alg}$)
satisfies Einstein causality if,
for any pair $\mathcal{N}, \mathcal{N}' \in \mathcal{M}$ of casually convex
open sets which are spacelike separated,
the Poisson bracket (\emph{resp.} commutator) of any pair
$F \in \mathfrak{P}(\mathcal{N}), G \in \mathfrak{P}(\mathcal{N}')$ vanishes.

It is in this section that we take advantage of
our alternative definition of locally covariant fields in~\cref{sec:chiral-primary-fields},
as it enables us to use results from the theory of distributions in our analysis.
Define a \emph{field with values in $\mathfrak{P}_{\ell}(\Sigma, \mathcal{M})$}
(with no assumptions on covariance) as a linear continuous map
$\Psi^i: \mathfrak{D}(\Sigma) \rightarrow \mathfrak{P}_{\ell}(\Sigma, \mathcal{M})$
satisfying $\mathrm{supp}\, \Psi^i_{\Sigma}(f) \subseteq \mathrm{supp}\, f$.

\begin{proposition}
  \label{prop:poisson-bracket-is-distribution}
    Let $\Psi^i$ and $\Psi^j$ be a pair of fields
    with values in $\mathfrak{P}_{\ell}(\Sigma, \mathcal{M})$ such that,
    for any $\psi \in \mathfrak{E}(\Sigma)$, the Schwartz kernel $K^{i/j}_{\psi} \in \mathfrak{D}(\mathcal{M}^2)$ associated to the map
    $f \mapsto \Psi^{i/j}(f)^{(1)}[\psi]$ satisfies
    $\mathrm{WF}(K^{i/j}_{\psi}) \cap \left\{ (x, y; \xi, 0) \in \dot{T}^*\mathcal{M}^2 \right\} = \emptyset$.
    Then the map
    \begin{equation}
        f \otimes g
            \mapsto
        \left\{ \Psi^i(f), \Psi^j(g) \right\}_\ell^\Sigma[\psi]
    \end{equation}
    defines a distribution $E^{ij}_{\psi} \in \mathfrak{D}'(\Sigma^2)$.
\end{proposition}

\begin{proof}
  We equip the underlying space $\mathfrak{F}_c(\Sigma)$ of $\mathfrak{P}_{\ell}(\Sigma, \mathcal{M})$
  with the topology $\tau_{BDF}$, which is the initial topology with respect to the maps
  \begin{align*}
    \mathfrak{F}_c(\Sigma^n) &\rightarrow \mathfrak{E}'_{\Xi_n}(\Sigma^n)\\
                          F &\mapsto F^{(n)}[\psi]
  \end{align*}
  where the topology on
  $\mathfrak{E}'_{\Xi_n}(\Sigma^n) = \left\{ u \in \mathfrak{E}(\Sigma^n) \,|\, \mathrm{WF}(u) \in \Xi_n \right\}$
  is the \emph{Hörmander topology} \cite[p2]{brouderContinuityFundamentalOperations2016},
  and the cones $\Xi_n = \Xi^n_+ \cup \Xi^n_-$ are defined in \cref{prop:chiral-poisson}.
  Because we assume by definition that fields are linear and continuous in $f$,
  it follows that
  \begin{equation}
    \label{eq:primary-derivative}
      f \mapsto \Psi^i(f)^{(1)}[\psi]
  \end{equation}
  is a linear, continuous map $\mathfrak{D}(\Sigma) \rightarrow \mathfrak{D}(\Sigma) \subset \mathfrak{E}'(\Sigma)$.
  In particular, this means that the Schwartz kernel $K^{i/j}_{\psi}$ in the statement of the proposition is well-defined.

  We then claim that the desired distribution has the integral kernel
  \begin{align}
    \left\{ \Psi^i(z_1), \Psi^j(z_2) \right\}[\psi] :=
    \int_{\Sigma^2} E(y_1, y_2) K^i_{\psi}(z_1, y_1) K^j_{\psi}(z_2, y_2) \,\mathrm{d}y_1 \mathrm{d}y_2
  \end{align}
  where $K^j_{\psi}$ the corresponding distribution from $\Psi^j_{\Sigma}$.
  To show that this integral kernel is well defined,
  we use \cite[Theorem 8.2.14]{hormanderAnalysisLinearPartial2015}, for which the necessary conditions are
  \begin{enumerate}
    \item The map $\mathrm{supp}\, (K^i_{\psi} \otimes K^j_{\psi}) \ni (z_1, y_1, z_2, y_2) \mapsto (z_1, z_2)$ is proper,
          i.e.\ the pre-image of any compact set is compact.
    \item $\left\{ (y_1, y_2; -\eta_1, -\eta_2) \in \mathrm{WF}(E) \,|\, \exists (z_1, y_1, z_2, y_2; 0, \eta_1, 0, \eta_2) \in \mathrm{WF}(K^i_{\psi} \otimes K^j_{\psi}) \right\} = \emptyset$.
  \end{enumerate}

  The first of these follows from the fact that $\mathrm{supp}\, \Psi^{i}(f) \subseteq \mathrm{supp}\, f$,
  from which we may deduce that
  $\mathrm{supp}\, (K^i_{\psi} \otimes K^j_{\psi}) \subseteq \left\{ (z_1, z_1, z_2, z_2) \in \Sigma^4 \right\}_{(z_1, z_2) \in \Sigma^2}$,
  hence the projection map is clearly proper.
  The second is then a straightforward consequence of the restriction we imposed on $\mathrm{WF}(K^{i/j}_{\psi})$ by hypothesis
\end{proof}

\begin{remark}
  The technical condition on $\mathrm{WF}(K^{i/j}_{\psi})$ may appear restrictive.
  However, if one considers the motivating example of such fields,
  \begin{align}
    \Psi^{P}(f)[\psi] = \int_{\mathbb{R}} f(x) P(\phi(x), \partial_x \phi(x), \ldots, \partial_x^n \phi(x)) \,\mathrm{d}x,
  \end{align}
  where $P$ is some polynomial with coefficients in $\mathfrak{E}(\mathbb{R})$,
  then $K^P_{\psi}(x, y)$ is a polynomial in $\delta(x - y)$ and its derivatives
  (with coefficients in $\mathfrak{E}(\mathbb{R})$).
  This means $\mathrm{WF}(K^P_{\psi})$ is orthogonal to the tangent bundle of $\Delta_2 \subset \mathcal{M}^2$,
  hence, in particular it satisfies the condition set out in \cref{prop:poisson-bracket-is-distribution}.
\end{remark}

We can think of the map $\mathfrak{P}_{\ell}(\Sigma, \mathcal{M}) \rightarrow \mathbb{R}$ given by
evaluation at a fixed $\psi \in \mathfrak{E}(\Sigma)$ as a classical state.
One of these states, namely $\psi \equiv 0$ is special in that it is invariant under
the action of $\mathrm{Diff}_+(\Sigma) \simeq \mathrm{Aut}_{\mathsf{CCauchy}}(\Sigma, \mathcal{M})$ on $\mathfrak{P}_{\ell}(\Sigma, \mathcal{M})$.
This is the classical version of the statement that
the vacuum state is invariant under conformal transformations.
As such, one can consider the following results as statements about the
`vacuum expectation values' of the corresponding observables.

For the remainder of this section, we assume that
$\Sigma = \Sigma_0$, the $t=0$ Cauchy surface of $\mathbb{M}^2$.
This allows us to use both translation as well as dilation morphisms.
We do this primarily for convenience,
as we can then easily formulate the condition of homogeneous scaling.
To generalise, we may either use the embedding results such at \cref{thm:cauchy-extension-skeletal},
or we could adopt a more geometric approach using the \emph{microlocal scaling degree}~%
\cite[\S 6]{brunettiMicrolocalAnalysisInteracting2000}
of the relevant distributions.
Note that, if we consider only conformal automorphisms of $\mathbb{M}^2$ preserving $\Sigma_0$,
then the only relevant data of a classical chiral primary field is a map $\mathfrak{D}(\Sigma_0) \rightarrow \mathfrak{P}_{\ell}(\Sigma_0, \mathbb{M}^2)$
which is equivariant with respect to the aforementioned action of $\mathrm{Diff}_+(\Sigma_0)$ on $\mathfrak{P}_{\ell}(\Sigma_0, \mathbb{M}^2)$
and the appropriately weighted action on $\mathfrak{D}(\Sigma_0)$.
Moreover, the subgroup of $\mathrm{Diff}_+(\Sigma_0)$ for which each of these weighted actions coincide
is precisely the group of translations of $\Sigma_0$, which motivates the additional generality in the following statement.

\begin{proposition}
  Let $\Psi^{i}, \Psi^{j}: \mathfrak{D}(\Sigma_0) \rightarrow \mathfrak{P}_{\ell}(\Sigma_0, \mathbb{M}^2)$ be a pair of linear, continuous maps
  which are equivariant with respect to translations on $\Sigma_0$.
  Then the distribution with integral kernel
  $E^{ij}_{0}(x, x') = \left\{ \Psi^{i}(x),  \Psi^{j}(x')\right\}_{\ell}^{\Sigma_0}[0]$
  is translation invariant,
  i.e.\ there exists a distribution in $\mathfrak{D}(\Sigma_0)$
  (which we shall also denote $E^{ij}_0$ by an abuse of notation)
  such that $E^{ij}_{0}(x, x') = E^{ij}_{0}(x - x')$.
\end{proposition}

\begin{proof}
  Let $t_c : x \mapsto x + c$ be a translation operator on $\Sigma_0$,
  it suffices to show that
  \begin{equation}
    \label{eq:translation-1}
      \left\langle E^{ij}_0, {t_{c}}_{*}f \otimes  {t_{c}}_{*}g \right\rangle = \left\langle E^{ij}_0, f \otimes  g \right\rangle
  \end{equation}
  for all $f, g \in \mathfrak{D}(\Sigma_0)$.
  The hypothesised equivariance of our fields amounts to the statement that
  \begin{equation}
      \Psi^i({t_c}_{*}f) = \mathfrak{P}_{\ell}t_c \Psi^i(f).
  \end{equation}
  Expanding out the left-hand side of~\eqref{eq:translation-1},
  and making use of the fact that $\mathfrak{P}_{\ell}{t_c}_{*}$ is a Poisson algebra homomorphism, we find
  \begin{equation}
      \left\langle E^{ij}_0,  {t_{c}}_{*}f \otimes  {t_{c}}_{*}g \right\rangle
      =
      \left( \mathfrak{P}_{\ell} t_c \left\{ \Psi^i(f),  \Psi^j(g) \right\}_{\ell}^{\Sigma_0} \right)[0].
  \end{equation}
  Recall that these homomorphisms were defined by
  $\left( \mathfrak{P}_{\ell} \rho F \right)[\psi] := F[\rho^{*}_{(1)} \psi]$.
  However, in this case our choice of configuration is invariant under the action of all such morphisms,
  hence we arrive at the desired equation.
\end{proof}

\begin{proposition}
    $E^{ij}_{0}$ is supported on the diagonal
    $\left\{ (x, x) \right\}_{x \in \Sigma_0} \subset \Sigma_0^2$,
    hence it is of the form
    \begin{equation}
      \label{eq:constraint-diagonal}
        E^{ij}_{0}(x, x')
            =
        \sum_{k = 0}^n
            a_k \left( \frac{\partial}{\partial x} \right)^k \delta(x - x'),
    \end{equation}
    for some $n \in \mathbb{N}$ and $a_k \in \mathbb{R}$.
\end{proposition}

The first statement is actually a consequence of Einstein causality in the full theory,
as well as the fact that $\mathrm{supp} \, \Psi^i(f) \subseteq \mathrm{supp} \, f$,
as was shown in \cref{lem:hom-scale-support}.
As such, we shall save the proof of this until~\cref{thm:chiral-causality}
The fact that a distribution on the diagonal is necessarily of the form~\eqref{eq:constraint-diagonal}
is~\cite[Theorem~2.3.4]{hormanderAnalysisLinearPartial2015}.

Now we have taken full advantage of the translation morphisms,
we introduce the dilation morphisms, for $\Lambda>0$,
$m_\Lambda: \mathbb{M}^2 \to \mathbb{M}^2; x \mapsto \Lambda \cdot x$.
Clearly these preserve $\Sigma_0$,
and we shall denote their restriction/co-restriction to $\Sigma_0$ also by $m_\Lambda$.
The next result is the first that utilises \emph{conformal} covariance,
which is why it is only now relevant whether or not the fields $\Psi^{i/j}$ are \emph{primary}.

First, we need to briefly introduce a new definition:

\begin{definition}
  A distribution $u \in \mathfrak{D}'(\mathbb{R}^n)$
  \emph{scales homogeneously with degree} $\mu \in \mathbb{R}$ if,
  $\forall \Lambda > 0$,
  $m_{\Lambda}^{*}u = \Lambda^{-\mu}u$ or, in terms of integral kernels,
  $u(\Lambda x) = \Lambda^{-\mu} u(x)$.
\end{definition}

\begin{proposition}
  \label{prop:full-constraint}
    If $\Psi^i$ and $\Psi^j$ are homogeneously scaling with weights $\mu_i, \mu_j \in \mathbb{N}$ respectively,
    then $E^{ij}_{0}$ scales homogeneously with degree $\mu_i + \mu_j$,
    hence
    \begin{equation}
      \label{eq:final-constraint}
        E^{ij}_0(x, x')
            =
        a \left( \frac{\partial}{\partial x} \right)^{\mu_i + \mu_j - 1}
        \delta(x - x').
    \end{equation}
\end{proposition}

\begin{proof}
  First we consider the claim of homogeneous scaling.
  Similarly to the translation invariance, it will suffice to show that
  \begin{equation}
      \left\langle m_{\Lambda}^{*} E^{ij}_0, f \otimes g \right\rangle
      \equiv
      \Lambda^{- 2} \left\langle E^{ij}_0, ({m_{\Lambda}}_*f) \otimes ({m_{\Lambda}}_*g) \right\rangle
      =
      \left\langle E^{ij}_0, f \otimes g \right\rangle
  \end{equation}
  For the dilation morphisms introduced above,
  the conformal factor $\omega_{\ell}$ is the constant $\Lambda^{-1}$,
  hence the naturality condition implies
  \begin{equation}
    \Lambda^{\mu_i - 1} \Psi^i \left( {m_{\Lambda}}_{*} \left( f \right) \right)
    =
    \mathfrak{P}_{\ell}m_{\Lambda} \Psi^i (f),
  \end{equation}
  where, since $\Lambda^{\mu_i - 1}$ is constant, we can simply pull it outside $\Psi^i$.
  Bringing these factors over to the right-hand side, we find
  \begin{equation}
      \left\langle E^{ij}_0, ({m_{\Lambda}}_*f) \otimes ({m_{\Lambda}}_*g) \right\rangle
      =
      \Lambda^{2 - \mu_i + \mu_j}
      \left( \mathfrak{P}_{\ell}m_{\Lambda} \left\{ \Psi^i(f), \Psi^j(g) \right\}_{\ell}^{\Sigma_0} \right)[0].
  \end{equation}
  Again, noting that our `state'' $F \mapsto F[0]$ is invariant under
  the $\mathfrak{P}_{\ell} m_{\Lambda}$ morphisms, we arrive at the desired equation.

  A quick calculation shows that that the distribution $(\partial/\partial x)^k \delta(x)$
  scales homogeneously with degree $1 + k$.
  \footnote{
    The general result is that for $\delta \in \mathfrak{D}'(\mathbb{R}^n)$,
    and $\alpha \in \mathbb{N}^n$ a multi-index,
    the distribution $\partial^{\alpha} \delta$ scales homogeneously with degree $n + |\alpha|$
  }
  As we have already established $E^{ij}_0$ to be of the form~\eqref{eq:constraint-diagonal},
  we see that all but one of these terms must vanish, leaving us with~\eqref{eq:final-constraint}.
\end{proof}

Finally, applying this result when $\Psi^i = \Psi^j$ is the chiral boson, we get.

\begin{corollary}
    The commutator of the chiral boson on $\Sigma_0$ is proportional to
    $\delta'(x - x')$.
\end{corollary}

Naturally, we already know this to be the case,
but it is nevertheless significant that conformal covariance alone determines everything except the constant of proportionality.
We shall see in~\cref{sec:boundary-term} that,
for the chiral boson, even this constant can determined
using partial knowledge of the \textsc{ope} of $\Psi$ with itself.

\begin{remark}
  Note that $\delta^{(\mu_i + \mu_j - 1)}$ is skew-symmetric precisely when $\mu_i + \mu_j$ is even.
  This is clearly satisfied for the bracket of a field with itself given that its weight is a natural number.
  Moreover, we can easily see how a similar result would look for fermionic fields.
  If the Poisson bracket was suitably graded,
  then the bracket of a fermionic field with itself
  would instead be a symmetric distribution supported on the diagonal
  and would hence vanish unless the weight $\mu$ was a half-integer.
\end{remark}

\subsection{Chiral Stress-Energy Tensor and Conformal Symmetry}
\label{sec:chiral-set-generate}

Another important example of a chiral primary field is the
(chiral component of the) stress-energy tensor.
This is simply half the $n=2$ case from \cref{ex:monomial-fields}
\begin{equation}
    T_\Sigma(f)[\psi] := \frac{1}{2}\int_\mathcal{I} f(s) \psi^2(s) \sqrt{\gamma'(s)} \mathrm{d}V_\Sigma(s).
\end{equation}
Note that the corresponding observable in the full algebra,
given as an integral kernel, is then
$\partial_\Sigma^*T_\Sigma(x)[\phi] = \tfrac{1}{2}(\partial_\Sigma \phi)^2(x)$,
so this is indeed the left-moving component of
the stress-energy tensor for the massless scalar field.

It is well-understood in the physical literature that
spacetime symmetries are generated infinitesimally by the stress-energy tensor:
either with the Poisson bracket in the classical theory,
or the commutator in the quantum theory.
In the framework of locally covariant \textsc{qft},
this fact is encapsulated by the principle of \emph{relative Cauchy evolution}.
The concept is a little more subtle than in the Wightman framework,
as a generic spacetime does not posess translation symmetries
(which by Noether's theorem would then be associated to momentum operators).

In relative Cauchy evolution,
rather than considering infinitesimal translations,
one instead perturbs the spacetime metric slightly,
$g \mapsto g + \epsilon h$ for some \emph{compactly supported}
symmetric tensor $h$.
One then compares how the time-evolution of an observable $\mathcal{O}$
localised in the past of $\mathrm{supp} \, h$ proceeds in
the perturbed and unperturbed spacetimes,
the discrepancy is what we call the relative Cauchy evolution of $\mathcal{O}$
with respect to $\epsilon h$.
For more details, see \cite[\S4.1]{brunettiGenerallyCovariantLocality2003}.
Notably, in many examples, it has been shown that relative Cauchy evolution
is generated infinitesimally by the stress-energy tensor.

In the present framework, we can demonstrate this explicitly with the following result

\begin{proposition}
    Let $\Sigma \subset \mathscr{E}$ be a Cauchy surface of the Einstein cylinder
    $\mathscr{E}$
    and let $h \in \mathfrak{D}(\Sigma)$ such that the flow
    $\rho^{(t)} \in \mathrm{Diff}(\Sigma)$ generated by the vector field
    $h \, \mathrm{d}/\mathrm{d}x$
    is orientation preserving at every $t$ where it is defined.
    Let $\Psi$ be the chiral boson \eqref{eq:chiral-boson}
    and $T$ be the chiral stress energy tensor.
    Then
    \begin{equation}
      \label{eq:set-conformal}
        \left\{ T_\Sigma(h), \Psi_\Sigma(f) \right\}_\ell^\Sigma
            =
        - \Psi_\Sigma(h f')
            =
        \frac{d}{dt} \left( \mathfrak{P}_{\ell} \rho^{(t)} \Psi_\Sigma(f) \right)\big|_{t = 0}.
    \end{equation}
    In other words, the $\mathrm{Diff}_+(\Sigma)$ covariance of $\Psi$
    is generated by taking the Peierls bracket with $T$.
\end{proposition}

\begin{proof}
  For the flow $\{\rho^{(t)} \in \mathrm{Diff}_+(\Sigma)\}_{t \in (-\epsilon, \epsilon)}$ of a vector field $X \in \mathfrak{X}(\Sigma)$,
  we can write $\rho^{(t)}_{*} f = (\rho^{(-t)})^{*} f = f - t \mathcal{L}_X f + \mathcal{O}(t^2)$
  where $\mathcal{L}_X$ denotes the Lie derivative along $X$.
  By linearity of $\Psi_{\Sigma}$ in the test function, we then have that
  \begin{equation}
    \mathfrak{P}_{\ell} \rho^{(t)} \Psi_{\Sigma}(f)
    =
    \Psi_{\Sigma}(\rho^{(t)}_{*}f)
    =
    \Psi_{\Sigma}(f) - t \Psi_{\Sigma}(\mathcal{L}_X f) + \mathcal{O}(t^2).
  \end{equation}
  setting $X = h \tfrac{d}{dx}$, we see that
  the second and third terms of~\eqref{eq:set-conformal} are equal.

To establish the first equation we just
compute the chiral Poisson bracket explicitly:
  \begin{align*}
    \left\{ T(h), \Psi(f) \right\}_\ell^\Sigma [\varphi]
        & =
    -\int_\mathcal{I}
        T(h)^{(1)}[\varphi](s)
        \frac{d \Psi(f)^{(1)}[\phi](s)}{ds}
        \, \mathrm{d}s \\
        &=
    -\int_\mathcal{I}
        h(s) \varphi(s) \sqrt{\gamma'(s)}
        \frac{df}{ds}(s)
        \, \mathrm{d}s \\
        &=
    -\Psi \left( h f' \right) [\varphi].
  \end{align*}
\end{proof}

\section{Quantisation of the Chiral Algebra}

Now that we have studied the classical algebra in detail,
we shall now see that many of our constructions require
only minimal adjustments to obtain the analogous quantum constructions.
As one might expect, we start by finding a deformation of
the classical algebras $\mathfrak{P}_{\ell}(\Sigma, \mathcal{M})$.
We then show how these algebras embed naturally into
$\mathfrak{A}(\mathcal{M})$, the algebra of quantum observables for the massless scalar field
(as constructed in \cite{crawfordLorentzian2dCFT2021}).

Finally, we discuss how this chiral algebra can compute
the \emph{operator product expansions} of both the chiral boson and the stress energy tensor,
and comment on how the form of these \textsc{ope}s is constrained by scaling invariance.

\subsection{The Quantum Chiral Algebra}
We shall start as we did when constructing the classical chiral algebra,
namely by considering a pair of linear functionals
$\Psi(f), \Psi(g)$ for $f, g \in \mathfrak{D}_{\ell / r}(\Sigma)$.
Firstly, ignoring wavefront sets,
$\partial_\Sigma^* \Psi(f)$ is a linear observable on $\mathcal{M}$,
thus we may attempt to compute its $\star_H$ product, resulting in
\begin{equation}
    \partial_\Sigma^* \Psi(f) \star_H \partial_\Sigma^* \Psi(g)
        =
    \partial_\Sigma^* \Psi(f) \cdot \partial_\Sigma^* \Psi(g) +
    \hbar \left\langle (\partial_\Sigma \otimes \partial_\Sigma)[\tfrac{i}{2}E + H], f \otimes g \right\rangle_{\Sigma^2}.
\end{equation}
Thus, similarly to the case of the chiral Poisson bracket,
the product of linear observables is computed by the bi-distribution
$(\partial_{\Sigma} \otimes \partial_{\Sigma}) [\tfrac{i}{2}E + H] =: W_{\Sigma} \in \mathfrak{D}'(\Sigma^2)$,
where again we can verify this distribution is well-defined by combining
\cite[proposition 4.4]{crawfordLorentzian2dCFT2021} and \cref{prop:config-covariance}
to obtain $W_{\Sigma} = (\rho^{*}_{(1)} \otimes \rho^{*}_{(1)}) W_{\Sigma_0}$,
for a suitable choice $\Sigma \overset{\sim}{\rightarrow} \Sigma_0 \subset \mathcal{M}_0 \in \{\mathbb{M}^2, \mathscr{E}\}$.
Moreover, precisely the same wavefront set condition that caused $\mathfrak{F}_c(\Sigma)$
to be closed as a Poisson algebra allows us to define
a deformation quantisation of that algebra via the following proposition:

\begin{proposition}
  Let $\Sigma \subset \mathcal{M}$ be a Cauchy surface of some globally hyperbolic spacetime,
  and let $\mathfrak{F}_c(\Sigma)$ denote the space of functionals defined in~\cref{prop:chiral-poisson}.
  Then, for any $H \in \mathrm{Had}(\mathcal{M})$,
  the space $\mathfrak{F}_c(\Sigma)[[\hbar]]$ equipped with the product
    \begin{equation}
      \label{eq:chiral-star-product}
    F \star_{H, \ell} G [\psi]
    =
    \sum_{n = 0}^{\infty} \frac{\hbar^{n}}{n!}
    \left\langle W_{\Sigma}^{\otimes n}, F^{(n)}[\psi] \otimes G^{(n)}[\psi] \right\rangle
    \end{equation}
is a closed $*$-algebra, with involution given by pointwise complex conjugation.
\end{proposition}

\begin{proof}
  The proof is comparable to that of~\Cref{prop:chiral-poisson}.
  We won't prove the more elementary properties of a $*$-algebra,
  though we shall point out that associativity follows from the particular form of~%
  \eqref{eq:chiral-star-product}~\cite[Proposition~4.5]{hawkinsStarProductInteracting2019}.

  For well-definedness and closure we must show that,
  for every $n,m > k \in \mathbb{N}$, the map
    $\mathfrak{D}(\Sigma^n) \otimes \mathfrak{D}(\Sigma^m) \to \mathfrak{D}'(\Sigma^{n + m - 2k})$
    defined by
    \begin{equation}
      \label{eq:quantum-unextended}
        \begin{split}
            f \otimes g
                \mapsto
            \int_{\Sigma^{2k}}
                &
                W_\Sigma(u_1, u_{k + 1})
                    \cdots
                W_\Sigma(u_k, u_{2k}) \\
                &
                f(u_1, \ldots, u_k, u'_1, \ldots u'_{n - k})
                g(u_{k+1}, \ldots, u_{2k}, u''_1, \ldots u''_{m - k})
            \, \mathrm{d}u_1 \cdots \, \mathrm{d}u_{2k}
        \end{split}
    \end{equation}
    extends to a map
    \begin{equation}
      \label{eq:quantum-extension}
        \mathfrak{E}'_{\Gamma_n}(\Sigma^n)
            \otimes
        \mathfrak{E}'_{\Gamma_m}(\Sigma^m)
            \to
        \mathfrak{D}'_{\Gamma_{n+m-2k}}(\Sigma^{n + m - 2k})
    \end{equation}
    where $\Gamma_n = T^{*} \Sigma^n \setminus (\Xi^n_+ \cup \Xi^n_-)$
    is the cone of allowable wavefronts from~\Cref{prop:chiral-poisson}.

    Looking in particular at $W_{\Sigma_0}$, or in general by applying Hörmander's pullback theorem
    to $(\Pi_{\ell}d \otimes \Pi_{\ell}d)W$ along the embedding $\Sigma \times \Sigma \hookrightarrow \mathcal{M} \times \mathcal{M}$,
    we see that
    \begin{equation}
      \label{eq:chiral-hadamard-condition}
      \mathrm{WF}(W_{\Sigma}) =
      \left\{ (r, r; \xi, -\xi) \in \dot{T}^{*}\Sigma^2 \,|\, \xi > 0 \right\},
    \end{equation}
    where the sign of $\xi$ is defined with respect to an arbitrary oriented chart on $\Sigma$.
    We must now consider the set
    \begin{equation*}
    \begin{split}
      (\overline{\Gamma}_{n+k} \times \overline{\Gamma}_{m+k} \setminus \underline{0}_{\Sigma^{n+m+2k}}) \circ \overline{\mathrm{WF}(W_{\Sigma}^{\otimes k})} &:=
      \big\{
        (\underline{s}_F, \underline{s}_G; \underline{\xi}_F, \underline{\xi}_G) \in T^{*}\Sigma^{n + m}
        \,|\,\\
        \exists (\underline{r}_1,\underline{r}_2; \underline{\eta}, -\underline{\eta}) \in T^{*}\Sigma^{2k},
        (r_{11}, &r_{12}, \ldots, r_{k1}, r_{k2} ; \eta_1, -\eta_1, \ldots , \eta_n, -\eta_n) \in \overline{\mathrm{WF}(W_{\Sigma}^{\otimes k})},\\
        (\underline{r}_1, \underline{s}_F&, \underline{r}_2, \underline{s}_G;
        \underline{\eta}, \underline{\xi}_F, -\underline{\eta} \underline{\xi}_G) \in
        (\overline{\Gamma}_{n+m} \times \overline{\Gamma}_{m+k} \setminus \underline{0}_{\Sigma^{n+m+2k}})
      \big\}.
    \end{split}
    \end{equation*}
    If this set is disjoint from $\underline{0}_{\Sigma^{n+m}}$,
    then the domain of~\eqref{eq:quantum-unextended} can be extended to the set in~\eqref{eq:quantum-extension},
    and if the set is disjoint from $\Xi^{n+m}_{\pm}$, then its codomain is also where we need it.
    Once again, it is in practise sufficient to consider the case
    $(\underline{s}_F, \underline{s}_G; \underline{\xi}_{F}, \underline{\xi}_{G}) \in \Xi^{n+m}_+$.
    For this covector to also belong to the set defined above, we must have
    $(\underline{r}_1, \underline{s}_F; \underline{\eta}, \underline{\xi}_F) \in \overline{\Gamma}_{n+k}$.
    By a similar method to that used for~\cref{prop:chiral-poisson},
    we can then show that this implies $\underline{\eta}$ and $\underline{\xi}_F$ both vanish,
    thus in turn $\underline{\xi}_G$ also vanishes, leading to a contradiction.
    (Here, the smaller wavefront set of $W_{\Sigma}$ compared to $E_{\Sigma}$ becomes relevant.)
\end{proof}

A routine calculation then shows that, for $H, H' \in \mathrm{Had}$,
one can define maps $\beta_{H' - H}: \mathfrak{F}_c(\Sigma)[[\hbar]] \rightarrow \mathfrak{F}_c(\Sigma)[[\hbar]]$
analogously to \cite[\S 5.1]{rejznerPerturbativeAlgebraicQuantum2016},
\begin{align}
  \label{eq:beta-H-Hp-relation}
  \beta_{H' - H} F = \sum_{n=0}^{\infty} \frac{\hbar^n}{2^n n!} \left\langle (H'_{\Sigma} - H_{\Sigma})^{\otimes n}, F^{(2n)} \right\rangle
\end{align}
which intertwine $\star_{H, \ell}$ with $\star_{H', \ell}$.
We then define the quantum chiral algebra as follows:

\begin{definition}[Quantum Chiral Algebra]
  \label{def:quantum-chiral-algebra}
  Let $\Sigma$ be a Cauchy surface of some globally hyperbolic spacetime $\mathcal{M}$,
  the \emph{quantum chiral algebra} on $\Sigma$ is the $*$-algebra defined by
  \begin{equation}
    \mathfrak{A}_{\ell}(\Sigma, \mathcal{M}) = \left\{
      (F_H)_{H \in \mathrm{Had}(\mathcal{M})} \subset \mathfrak{F}_c(\Sigma)[[\hbar]]
      \,|\,
      \beta_{H' - H} F_H = F_{H'}
    \right\}
  \end{equation}
  with product
  \begin{equation}
      (F_H)_{H \in \mathrm{Had}(\mathcal{M})} \star_{\ell}
      (G_H)_{H \in \mathrm{Had}(\mathcal{M})} =
      (F_H \star_{H, \ell} G_H)_{H \in \mathrm{Had}(\mathcal{M})}.
  \end{equation}
  To a given $\mathsf{CCauchy}$ morphism $(\rho, \chi): (\Sigma, \mathcal{M}) \rightarrow (\widetilde{\Sigma}, \widetilde{\mathcal{M}})$,
  we assign the map defined, for $\widetilde{H} \in \mathrm{Had}(\widetilde{\mathcal{M}})$ by
  \begin{align}
    \label{eq:quantum-chiral-hom}
    (\mathfrak{A}_{\ell}(\rho, \chi) F)_{ \widetilde{H} } = F_{\chi^{*} \widetilde{H}} \circ \rho^{*}_{(1)}.
  \end{align}
\end{definition}

\begin{remark}
  Note that the map \eqref{eq:quantum-chiral-hom} is well-defined because:
  firstly, we have already seen that $F \mapsto F \circ \rho^{*}_{(1)}$ is a well-defined map
  $\mathfrak{F}_c(\Sigma) \mapsto \mathfrak{F}_c(\widetilde{\Sigma})$;
  secondly, the consistency condition is satisfied, because
  \begin{align*}
    \left\langle
      (\widetilde{H}'_{\Sigma} - \widetilde{H}_{\Sigma})^{\otimes n},
      (F_{\chi^{*} \widetilde{H}} \circ \rho^{*}_{(1)})^{(2n)}
    \right\rangle
    &= \left\langle
      (\rho^{*}_{(1)})^{\otimes 2n} (\widetilde{H}'_{\Sigma} - \widetilde{H}_{\Sigma}),
      F_{\chi^{*} \widetilde{H}}^{(2n)} \circ \rho^{*}_{(1)}
    \right\rangle, \\
    &= 2^n\left( \frac{d^{n}}{d \hbar^{n}} \beta_{\chi^* \widetilde{H}' - \chi^* \widetilde{H}} F |_{\hbar = 0} \right) \circ \rho^{*}_{(1)};
  \end{align*}
  and lastly, using a similar equation to the above we can verify that $\mathfrak{A}_{\ell}(\rho, \chi)$
  is a homomorphism with respect to $\star_{\ell}$.
\end{remark}

Having defined the quantum chiral algebra,
we must ask both how these algebras vary as we change Cauchy surfaces,
and how they relate to the full algebras $\mathfrak{A}(\mathcal{M})$.
Both of these are addressed by the following theorem.

\begin{theorem}
  For an object $(\Sigma, \mathcal{M}) \in \mathsf{CCauchy}$,
  any choice of map $\partial_{\Sigma, \epsilon}: \mathfrak{E}(\mathcal{M}) \rightarrow \mathfrak{E}(\Sigma)$ as defined in \cref{thm:classical-natural}
  defines a map $\theta_{\Sigma, \epsilon}: \mathfrak{A}_{\ell}(\Sigma, \mathcal{M}) \rightarrow \mathfrak{A}(\mathcal{M})$
  by
  \begin{align}
    \theta_{\Sigma, \epsilon} \left( F_H \right)_{H \in \mathrm{Had}(\mathcal{M})}
    := \left( \partial^{*}_{\Sigma, \epsilon} F_{H} \right)_{H \in \mathrm{Had}(\mathcal{M})}
  \end{align}

  Moreover, choosing such a map for every pair $(\Sigma, \mathcal{M})$ yields a natural transformation
  $\theta: \mathfrak{A}_{\ell} \Rightarrow \mathfrak{A}_{\mathrm{on}} \circ \Pi_2$.
\end{theorem}

\begin{proof}
  Firstly, this map is well-defined because
  \begin{align}
    \alpha_{H' - H} \partial_{\Sigma, \epsilon}^{*} = \partial_{\Sigma, \epsilon}^{*} \beta_{H' - H},
  \end{align}
  as can be verified by a direct computation using property \textbf{H2} from
  \cite[\S 5.1]{rejznerPerturbativeAlgebraicQuantum2016} of Hadamard distributions to conclude that
  $(\partial_{\Sigma, \epsilon} \otimes \partial_{\Sigma, \epsilon}) H = (\partial_{\Sigma} \otimes \partial_{\Sigma}) H$.
  Given this, most of what remains to be shown follows directly from
  the corresponding classical result \cref{thm:classical-natural}.
  For the first statement all that we really need to check is that
  each $\theta_{\Sigma, \epsilon}$ is a homomorphism of $\star$ products.
  Looking at the coefficient of $\hbar^n$, this requires one to show that
  for every $H \in \mathrm{Had}(\mathcal{M})$ and $n \in \mathbb{N}$
  \begin{align}
  \begin{split}
    \big\langle (\tfrac{i}{2}E + &H)^{\otimes n}, (\partial_{\Sigma, \epsilon} F)^{(n)}[\phi] \otimes (\partial_{\Sigma, \epsilon} G)^{(n)}[\phi] \big\rangle\\
    &= \left\langle [(\partial_{\Sigma} \otimes \partial_{\Sigma})^{\otimes n} (\tfrac{i}{2}E + H)]^{\otimes n}, F^{(n)}[\partial_{\Sigma, \epsilon} \phi] \otimes G^{(n)}[\partial_{\Sigma, \epsilon} \phi] \right\rangle.
  \end{split}
  \end{align}
  Similarly to before, we can show that, for $u \in \mathfrak{D}'(\mathcal{M}^{2n})$ for any $n \in \mathbb{N}$,
  \begin{align}
    \begin{split}
      \left\langle u, (\partial_{\Sigma, \epsilon} F)^{(n)}[\phi] \otimes (\partial_{\Sigma, \epsilon} G)^{(n)}[\phi] \right\rangle \\
      = \left\langle (\partial_{\Sigma, \epsilon} \otimes \partial_{\Sigma, \epsilon} )^{\otimes n} u, F^{(n)}[\partial_{\Sigma, \epsilon} \phi] \otimes G^{(n)}[\partial_{\Sigma, \epsilon} \phi] \right\rangle.
    \end{split}
  \end{align}
  which once again follows from the fact that both $E$ and $H$ are bi-solutions to the equations of motion.

  Naturality is also essentially unchanged from the classical result.
  Let $(\rho, \chi): (\Sigma, \mathcal{M}) \rightarrow (\widetilde{\Sigma}, \widetilde{\mathcal{M}})$ be a $\mathsf{CCauchy}$ morphism.
  For the appropriate \emph{on-shell} diagram to commute, we need to show that,
  $\forall (F_H) \in \mathfrak{A}_{\ell}(\Sigma, \mathcal{M})$
  \begin{align}
    \mathfrak{A} \chi \theta_{\Sigma, \epsilon} (F_H)
    =
    \theta_{\widetilde{\Sigma}, \widetilde{\epsilon}} \mathfrak{A}_{\ell} (\rho, \chi) (F_H)
      + \mathfrak{I}_S(\widetilde{\mathcal{M}}),
  \end{align}
  where we have suppressed the index set $\mathrm{Had}(\mathcal{M})$,
  and we have used the fact that the ideal $\mathfrak{I}_S(\mathcal{M})$ can
  be identified unambiguously as a subspace of $\mathfrak{A}(\widetilde{\mathcal{M}})$.
  Using the definition of each of these maps, this means we need that,
  for every $\widetilde{H} \in \mathrm{Had}(\widetilde{\mathcal{M}})$,
  $\phi \in \mathrm{Ker}\, P_{\widetilde{\mathcal{M}}}$
  \begin{align*}
    F_{H}[\partial_{\Sigma, \epsilon} \chi^{*} \phi] = F_{H}[\rho^{*}_{(1)} \partial_{\widetilde{\Sigma}, \widetilde{\epsilon}} \phi],
  \end{align*}
  where $H = \chi^{*} \widetilde{H}$,
  which is precisely the same equation as in~\cref{thm:classical-natural}.
\end{proof}

\subsection{Operator Product Expansions}

We now turn our attention to one of the central features of 2\textsc{dcft}:
the \emph{operator product expansion}.
We shall begin by summarising how they arise in Euclidean signature,
then continue to demonstrate how several of these equations arise in the Lorentzian setting.

There exists a very powerful axiomatisation of the chiral sector of a conformally invariant Wightman \textsc{qft}
build upon the mathematical structure of \emph{vertex operator algebras}.
For a comprehensive account of this framework, we refer the reader to
\cite{kacVertexAlgebrasBeginners1998} and \cite{frenkelVertexAlgebrasAlgebraic2004}.
However, we shall give a brief account here of the concepts of fields and their \textsc{ope}s,
which are of central significance in this approach.

A \emph{field} in a vertex operator algebra is
defined as formal power-series $a(z) = \sum_{n \in \mathbb{Z}} a_n z^{-n}$,
where $a_n \in \mathrm{End}(V)$ for some complex vector space $V$
(the space of states of the theory)
and subject to the condition that, for every $v \in V$, there exists $N \in \mathbb{Z}$ such that
$a_n v = 0 \forall n \ge N$.

This definition can be connected to our previous notions of a field.
If we express a Laurent \emph{polynomial} $f \in \mathbb{C}[z, z^{-1}]$ as
$\sum_{n \in \mathbb{Z}} f_n z^{-n}$ where only finitely many $f_n \in \mathbb{C}$ are non-zero,
then one can think of $a(z)$ as a distribution on this space
with values in $\mathrm{End}(V)$, where the pairing is given by
\begin{align}
  \left\langle a, f \right\rangle := \sum_{n \in \mathbb{Z}} a_n f_{-n}.
\end{align}

Given a pair, $a(z), b(w)$ of such fields,
one is often interested in whether they are \emph{mutually local}.
We can define the composition of $a(z)$ with $b(w)$ as an element of
$\mathrm{End}(V)[[z, z^{-1}, w, w^{-1}]]$ by taking
the appropriate compositions of each coefficient $a_n b_m$.
The fields are then said to be mutually local if, for some $N \in \mathbb{N}$
\begin{equation}
  \label{eq:voa-locality-def}
    {(z - w)}^N[a(z), b(w)] = 0.
\end{equation}

One can then show~\cite[Corollary~2.2]{kacVertexAlgebrasBeginners1998} that
this commutator is of the form
\begin{equation}
  \label{eq:voa-local-commutator}
  [a(z), b(w)] = \sum_{j=0}^{N-1} \partial_w^{(j)}\delta(z-w)c^j(w),
\end{equation}
for some $c^j(w) \in \mathrm{End}(V)[[w, w^{-1}]]$.
The distributions $\partial_w^{(j)}$ are then expressed as the difference of boundary values of
two holomorphic functions on open subsets of $\mathbb{C}^2$.
Both functions are written as $(z - w)^{-(j+1)}$,
however the domain of the first instance is taken to be the region $|z| > |w|$,
whereas the domain of the second is $|w| < |z|$.
The \emph{OPE} of $a$ and $b$ in either of these domains is then obtained
by replacing each $\partial^{(j)}_w \delta(z - w)$ with
the appropriate holomorphic function in~\eqref{eq:voa-local-commutator}.
This is typically written
\begin{equation}
  \label{eq:voa-ope}
  a(z)b(w) \sim
  \sum_{j=0}^{N-1} \frac{c^j(w)}{(z-w)^{j+1}}.
\end{equation}
This series differs from the \emph{actual} product $a(z)b(w)$
by a term referred to as the \emph{normally ordered product} $\nord{a(z)b(w)}$,
which can be computed by a decomposition of $a(z)$ and $b(z)$ into
positive and negative frequency modes.

We have already seen a property similar to~\eqref{eq:voa-locality-def} in the classical theory.
Recall that the chiral bracket of the chiral boson with itself
may be considered a bi-distribution on $\Sigma^2$ for any Cauchy surface $\Sigma$.
In the theory of distributions used in {\smaller{p}}\textsc{aqft}, the analogue of this equation
is the combined statement that a distribution $u \in \mathfrak{D}'(\mathbb{R}^2)$
is supported on the diagonal $\left\{ x, x \right\}_{x \in \mathbb{R}} \subset \mathbb{R}^2$,
and that it has a scaling degree $N$ with respect to the diagonal.
Even generically, the axiom of Einstein causality,
when restricted to a chiral algebra on a Cauchy surface,
implies that the Peierls bracket/commutator of two fields has this property
in a Lorentzian \textsc{aqft}.

We now turn our attention to the problem of finding an analogue of~\eqref{eq:voa-ope}
in our current framework.
We shall discuss the general case shortly,
however we shall begin by computing several examples.

In order to be concrete, we restrict once more to spacetimes which are
open, causally convex subsets of Minkowski space.
Suppose that $\Sigma$ is a Cauchy surface of such a subset $U \subseteq \mathbb{M}^2$.
From now on, we shall always assume our Hadamard distribution $W$ is the restriction of
\begin{align}
  W_{\mathbb{M}^2}(u, v, u', v') = \lim_{\epsilon \searrow 0} \frac{- 1}{4\pi} \ln\left(\frac{-(u - u')(v - v') + i \epsilon t}{\Lambda^2}\right).
\end{align}
For a Cauchy surface $\Sigma = \{(-s, \gamma(s))\}_{s \in \mathbb{R}}$ of $\mathbb{M}^2$,
we can write the chiral derivative of $W_{\mathbb{M}^2}$ as
\begin{align}
  W_{\Sigma}(s, s') = \frac{-1}{4 \pi} \frac{1}{\sqrt{\gamma'(s) \gamma'(s')}} \left( \frac{1}{(s - s')^2} + i \delta'(s - s') \right).
\end{align}
From this we can also get the formula for $\Sigma \subset U \subseteq \mathbb{M}^2$ by a suitable restriction,
as every Cauchy surface of $U$ can be extended to a Cauchy surface of $\mathbb{M}^2$.
This Hadamard distribution gives us a concrete realisation of $\mathfrak{A}_{\ell}(\Sigma, U)$
as $(\mathfrak{F}_c(\Sigma)[[\hbar]], \star_{H, \ell})$
Let us take a pair $f, g \in \mathfrak{D}(\Sigma)$ with disjoint support.
Using this realisation, the $\star$ product of the chiral boson evaluated on each test function is
\begin{align}
  (\Psi_{\Sigma}(f) \star_{H, \ell} \Psi_{\Sigma}(g))[\psi]
  =
  (\Psi_{\Sigma}(f) \cdot \Psi_{\Sigma}(g))[\psi] - \frac{\hbar}{4 \pi} \int_{\mathbb{R}^{2}} \frac{f(s) g(s')}{(s - s')^2} \,\mathrm{d}s\,\mathrm{d}s',
\end{align}
where the disjoint support allows us to drop the imaginary part of the $\mathcal{O}(\hbar)$ term.
In the above expression there is nothing to stop us from taking a limit where
$f$ and $g$ converge to Dirac deltas with supports at a pair of fixed points $s, s'$.
Doing so, we can write
\begin{align}
  (\Psi_{\Sigma}(s) \star_{H, \ell} \Psi_{\Sigma}(s'))[\psi]
  =
  \psi(s)\psi(s') - \frac{\hbar}{4 \pi} \frac{1}{(s - s')^2}.
\end{align}
This is a smooth function on $\Sigma^2 \setminus \Delta$.
If we restrict our attention to the term which is singular as we approach the diagonal,
we could then write
\begin{align}
  (\Psi_{\Sigma}(s) \star_{H, \ell} \Psi_{\Sigma}(s'))[\psi] \sim -\frac{\hbar}{4 \pi} \frac{1}{(s - s')^2}
\end{align}

As another example, we may define the chiral stress energy tensor by
\begin{equation}
    T_\Sigma(f)[\psi]
        :=
    \frac{1}{2} \int_\Sigma f(s) \psi{(s)}^2 \sqrt{\gamma'(s)} \,\mathrm{d}s.
\end{equation}
Applying the same procedure to this field, we find
\begin{align}
    \begin{split}
        \left(T_{\Sigma}(f) \star_{H, \ell} T_{\Sigma}(g)\right)[\psi]
            =
        T_{\Sigma}(f) \cdot T_{\Sigma}(g) [\psi]
            +
        \frac{\hbar}{4 \pi} \int_{\mathbb{R}^2} \frac{
            f(s) \psi(s) g(s') \psi(s')
        }{
            (s - s')^2
        }   \, \mathrm{d}s \, \mathrm{d}s' \\
            +
        \frac{\hbar^2}{32 \pi^2} \int_{\mathbb{R}^4}
            \frac{
                f(s_1) \delta(s_1 - s_2) g(s_3) \delta(s_3 - s_4)
            }{(s_1 - s_3)^2 (s_2 - s_4)^2}
        \, \mathrm{d}s_1 \cdots \, \mathrm{d}s_4.
    \end{split}
\end{align}
Where we once again can allow $f, g$ to approach deltas to obtain
\begin{equation}
    \label{eq:TT-OPE}
    \left( T_{\Sigma}(s) \star_{H, \ell} T_{\Sigma}(s') \right)[\psi]
        =
    \frac{1}{4} \psi{(s)}^2 \psi{(s')}^2
        +
    \frac{\hbar}{4 \pi}
    \frac{\psi(s) \psi(s')}{{(s - s')}^2}
        +
    \frac{\hbar^2}{32 \pi^2}
    \frac{1}{{(s - s')}^4}.
\end{equation}

Taylor expanding $\phi(s)$ around $s'$, we then obtain a well-known \textsc{ope}
\begin{equation}
    \omega_{H_{\mathbb{M}^2}, \psi} \left( T_{\Sigma}(s) \star_{H, \ell} T_{\Sigma}(s') \right)
    \sim
    \frac{\hbar^2}{32 \pi^2} \frac{1}{{(s - s')}^4} +
    \frac{\hbar}{4 \pi} \frac{2 T_{\Sigma}(s')[\psi]}{{(s - s')}^2} +
    \frac{\hbar}{4 \pi} \frac{T'_{\Sigma}(s')[\psi]}{s - s'},
  \end{equation}
where $T'_{\Sigma}(g) := - T_{\Sigma}(*_{\Sigma} d_{\Sigma} g)$.

Now that we have seen some explicit examples,
we can make some comments about the general connection.
In the \textsc{voa} formalism, the analysis of the products of fields
is enabled by the existence of a product
\begin{align}
  \mathrm{End}(V)[[z, z^{-1}]] \otimes
  \mathrm{End}(V)[[w, w^{-1}]] \rightarrow
  \mathrm{End}(V)[[z, z^{-1}, w, w^{-1}]]
\end{align}
from operator-valued distributions to operator valued \emph{bi}-distributions.
Given that our construction creates algebras first and fields second,
it is not immediately obvious that such a result holds.

Recall that in the classical theory we had a topology on our algebra
$\mathfrak{P}_{\ell}(\Sigma, \mathcal{M})$ which was initial with respect to the differentiation maps
$\mathfrak{F}_c(\Sigma) \rightarrow \mathfrak{E}'_{\Xi_n}(\Sigma^n)$.
This topology naturally extends to each of the spaces
$\mathfrak{F}_c(\Sigma)[[\hbar]]$,
and we can further show that the maps
$\beta_{H' - H}: \mathfrak{F}_c(\Sigma)[[\hbar]] \rightarrow \mathfrak{F}_c(\Sigma)[[\hbar]]$
are homeomorphisms with respect to this topology \cite[\S 3.1]{brunettiPerturbativeAlgebraicQuantum2009},
hence if we equip $\mathfrak{A}_{\ell}(\Sigma, \mathcal{M})$ with
the initial topology with respect to the maps
\begin{align*}
  (F_H)_{H \in \mathrm{Had}(\mathcal{M})} \mapsto F_{H_0}
\end{align*}
for each $H_0 \in \mathrm{Had}(\mathcal{M})$,
then these maps are also homeomorphisms.
In other words, to check that a map into or out of $\mathfrak{A}_{\ell}(\Sigma, \mathcal{M})$
is continuous, one need only show that this is the case for any of its concrete realisations.
This fact is particularly convenient when considering states of the form we define below.

\begin{definition}
  A \emph{Gaussian state} is a map $\omega_{H_0, \psi}: \mathfrak{A}_{\ell}(\Sigma, \mathcal{M}) \rightarrow \mathbb{C}[[\hbar]]$,
  where $H_0 \in \mathrm{Had}(\mathcal{M})$, $\psi \in \mathfrak{E}(\Sigma)$,
  defined by
  \begin{align}
    \omega_{H_0, \psi}(F_{H})_{H \in \mathrm{Had}(\mathcal{M})} \mapsto F_{H_0}[\psi].
  \end{align}
\end{definition}

\begin{remark}
  As we have already noted, evaluation maps are continuous in the Bastiani topology,
  hence the map $\omega_{H_0, \psi}: \mathfrak{A}_{\ell}(\Sigma, \mathcal{M}) \rightarrow \mathbb{C}[[\hbar]]$
  is also continuous.
\end{remark}

Now that we have identified a suitable topology,
we may define a \emph{field} on $\Sigma$ (without any mention of local covariance)
as any continuous, linear map $\Psi^i: \mathfrak{D}(\Sigma) \rightarrow \mathfrak{A}_{\ell}(\Sigma, \mathcal{M})$.
For reasons we shall see in the following proposition,
we shall also add the constraint that
\begin{align}
  \label{eq:field-support-condition}
  \mathrm{supp}\, \Psi^i(f) \subseteq \mathrm{supp}\, f,
\end{align}
where the support of an observable $(F_H)_{H \in \mathrm{Had}(\mathcal{M})} \in \mathfrak{A}_{\ell}(\Sigma, \mathcal{M})$
is defined as the union of the supports of each coefficient of $\hbar^n$ of $F_H$ for any $H \in \mathrm{Had}(\mathcal{M})$.
(This definition makes sense as $\beta_{H' - H}$ does not affect the support.)

\begin{proposition}
  Let $\omega_{H, \psi}: \mathfrak{A}_{\ell}(\Sigma, \mathcal{M}) \rightarrow \mathbb{C}$ be a Gaussian state,
  and let $\Psi^i, \Psi^j: \mathfrak{D}(\Sigma) \rightarrow \mathfrak{A}_{\ell}(\Sigma, \mathcal{M})$ be a pair of fields.
  Then the map
  \begin{align}
    f \otimes g \mapsto \omega_{H, \psi}\left( \Psi^i(f) \star \Psi^j(g) \right)
  \end{align}
  defines a distribution on $\Sigma^2$.
\end{proposition}

\begin{proof}
  Given the topology we have defined on $\mathfrak{A}_{\ell}(\Sigma, \mathcal{M})$
  all we must show is that, for each $n \in \mathbb{N}$, the map
  \begin{align}
    \label{eq:star-distribution-on-tensor}
    f \otimes g \mapsto
    \left\langle (\tfrac{i}{2}E_{\Sigma} + H_{\Sigma})^{\otimes n}, \Psi^i(f)^{(n)}[\psi] \otimes \Psi^j(g)^{(n)}[\psi] \right\rangle
  \end{align}
  defines a distribution in $\mathfrak{D}'(\Sigma^2)$.

  The argument follows a similar line to the classical case in \cref{sec:constraints}.
  Firstly, we see that $\forall H \in \mathrm{Had}(\mathcal{M})$ the map
  $f \mapsto (\Psi^i(f))_H^{(n)}[\psi]$ is a linear, continuous map
  $\mathfrak{D}(\Sigma) \rightarrow \mathfrak{E}'_{\Xi_n}(\Sigma^n)$.
  Hence, there is associated to it by the Schwartz kernel theorem
  a distribution $K^i_{n, \psi} \in \mathfrak{D}'(\Sigma^{n + 1})$.
  The distribution corresponding to the coefficient of $\hbar^n$ in the $\star$ product is then
  \begin{align*}
    \begin{split}
      \frac{d^{n}}{d \hbar^n} \omega_{H, \psi}(\Psi^i(z_1) \star \Psi^j(z_2))|_{\hbar = 0} =
      \int_{\Sigma^{2n}} \big[ W_{\Sigma}(y_1, y_2) \cdots W_{\Sigma}(y_{2n-1}, y_{2n})\\
      K^i_{n, \psi}(z, y_1, \ldots, y_{2n-1}) K^j_{n, \psi}(z, y_2, \ldots, y_{2n}) \big]\,\mathrm{d}y_1 \cdots \mathrm{d}y_{2n}.
    \end{split}
  \end{align*}
  where $W_{\Sigma} = \tfrac{i}{2}E_{\Sigma} + H_{\Sigma}$.

  We must then check the same conditions for this composition as we did in \cref{sec:constraints}.
  From $\mathrm{supp}\, \Psi^i(f) \subseteq \mathrm{supp}\, f$,
  we know $\Psi^i(f)^{(n)}[\psi] \subseteq (\mathrm{supp}\, f)^{\times n}$ and hence
  \begin{align*}
    \left\{ (z, y_1, \ldots y_n) \in \mathrm{supp}\, K^i_{n, \psi} \,|\, z \in U  \right\}
    \subseteq U^{n+1},
  \end{align*}
  which shows that the required projection is proper.

  For the wavefront sets, we have that
  \begin{align*}
    (z, y_1, \ldots, y_n; 0, \eta_1, \ldots, \eta_n) \in \mathrm{WF}(K^i_{n, \psi}) \Rightarrow
    (y_1, \ldots, y_n; \eta_1, \ldots, \eta_n) \in \Xi_n,
  \end{align*}
  which in turn implies
  \begin{align*}
    (z_1, z_2, y_1, \ldots, y_{2n}; 0, 0, \eta_1, \ldots \eta_{2n}) \in \mathrm{WF}(K^i_{n, \psi} \otimes K^j_{n, \psi}) \Rightarrow\\
    (y_1, \ldots, y_{2n}; \eta_1, \ldots \eta_{2n}) \in \Xi_{2n} \subseteq \dot{T}^{*}\Sigma^{2n} \setminus -\mathrm{WF}(W_{\Sigma}^{\otimes n}).
  \end{align*}
  From which we may conclude that each composition is a well-defined distribution coinciding with
  the coefficient of $\hbar^n$ in \eqref{eq:star-distribution-on-tensor}.
\end{proof}

We shall denote this distribution using its integral kernel
\begin{align*}
  \omega_{H, \psi}\left( \Psi^i(s) \star \Psi^j(s') \right).
\end{align*}

\begin{proposition}
  The anti-symmetric part of the above distribution can be written as
  \begin{align*}
    \omega_{H, \psi}\left( \left[ \Psi^i(s), \Psi^j(s') \right]_{\star} \right)
  \end{align*}
  and is supported on the diagonal $\Delta \subset \Sigma^2$.
\end{proposition}

\begin{proof}
  The fact that we can express the anti-symmetric part of the distribution in terms of the commutator
  is just a consequence of the linearity of $\omega_{H, \psi}$.
  To show that it is supported on the diagonal,
  we use the condition~\eqref{eq:field-support-condition} and note that any $h \in \mathfrak{D}(\Sigma^2 \setminus \Delta)$
  can be approximated by a series $f \otimes g$ such that $\mathrm{supp}\, f \cap \mathrm{supp}\, g = \emptyset$.
  We can then use the chiral version of Einstein causality
  (which we prove in a general setting in~\cref{thm:chiral-causality})
  to show that $[\Psi^i(f), \Psi^j(g)]_{\star} = 0$ for each such $f, g$,
  hence the distribution must send $h$ to 0.
\end{proof}

We can now provide an analysis similar to that of~\cref{sec:constraints}
by limiting our attention to homogeneously scaling chiral fields on $\mathbb{M}^2$.
As before, we shall consider the $\Sigma_0$ Cauchy surface,
and the action of the dilation operators $\{m_{\Lambda}: \Sigma_0 \rightarrow \Sigma_0\}_{\Lambda > 0}$
and translation operators $\{t_c: \Sigma_0 \rightarrow \Sigma_0\}_{c \in \mathbb{R}}$.
Once again, we shall use $t_c, m_{\Lambda}$ to refer to both the automorphisms of $\Sigma_0$
as well as the full spacetime $\mathbb{M}^2$.

If
$
  \Psi^i: \mathfrak{D}^{(\mu)}_{\ell}|_{\mathsf{Cauchy}(\mathbb{M}^2)_0}
  \Rightarrow
  \mathfrak{A}_{\ell}|_{\mathsf{Cauchy}(\mathbb{M}^2)_0}
$
be a homogeneously scaling locally covariant field on $\mathbb{M}^2$ with values in $\mathfrak{A}_{\ell}$,
then in particular it responds to scalings and translations as
\begin{align}
  \Psi^i_{\Sigma_0}(\Lambda s) = \Lambda^{\mu - 1} \mathfrak{A}_{\ell}m_{\Lambda}(\Psi^i_{\Sigma_0}(s)), \\
  \Psi^i_{\Sigma_0}(s + c) = \mathfrak{A}_{\ell}t_c(\Psi^i_{\Sigma_0}(s)).
\end{align}

\begin{proposition}
  \label{prop:quantum-scale-invariance}
  The Gaussian state $\omega_{H_{\mathbb{M}^2}, 0}: \mathfrak{A}_{\ell}(\Sigma_0, \mathbb{M}^2) \rightarrow \mathbb{C}[[\hbar]]$
  is invariant with respect to the action of the scaling and translation morphisms.
\end{proposition}

\begin{proof}
  Translation invariance is trivial.
  Explicitly, the equation for scaling invariance is satisfied if
  \begin{align}
    F_{m_{\Lambda}^{*}H_{\mathbb{M}^2}}[0]
    =
    F_{H_{\mathbb{M}^2}}[0].
  \end{align}
  This is satisfied for every $F$ because $\beta_{H' - H}$ depends on
  the Hadamard distributions only through their chiral derivatives.
  Given that $m_{\Lambda}^{*}H_{\mathbb{M}^2} - H_{\mathbb{M}^2} = -\tfrac{1}{2\pi} \ln(\Lambda)$,
  the chiral derivative of this term vanishes, hence $\beta_{m_{\Lambda}^{*}H_{\mathbb{M}^2} - H_{\mathbb{M}^2}}$
  acts as the identity.
\end{proof}

\begin{remark}
  It is worth noting that it was necessary for us to use the fact that
  chiral observables are all defined in terms of the derivative field $\partial_{\Sigma} \phi$.

\end{remark}

We may now state the quantum analogue of~\cref{prop:full-constraint},
the proof of which can be directly adapted from that of the classical result.

\begin{proposition}
  Let $\Psi^i, \Psi^j$ be a pair of homogeneously scaling locally covariant fields of weights $\mu_i, \mu_j$
  on $\mathbb{M}^2$ with values in $\mathfrak{A}_{\ell}$.
  Then
  \begin{align}
    \label{eq:quantum-commutator-constraint}
    \omega_{H_{\mathbb{M}^2}, 0}
    \left(
      \left[
        \Psi^i_{\Sigma_0}(s), \Psi^j_{\Sigma_0}(s')
      \right]_{\star}
    \right)
    \propto \delta^{(\mu_i + \mu_j - 1)}(s - s').
  \end{align}
\end{proposition}

For the chiral boson, this reduces to the classical result we have already seen.
By combining \cref{ex:monomial-fields} with the argument in the proof of \cref{prop:quantum-scale-invariance},
we can see that
\begin{align}
  \mathfrak{A}_{\ell}(m_{\Lambda}) \circ \nord{T}_{(\Sigma_0, \mathbb{M}^2)} = \nord{T}_{(\Sigma_0, \mathbb{M}^2)} \circ \mathfrak{D}_{\ell}^{(2)}(m_{\Lambda}),
\end{align}
where $(\nord{T}_{(\Sigma_0, \mathbb{M}^2)})_{H_{\mathbb{M}^2}} := T_{(\Sigma_0, \mathbb{M}^2)}$.
Hence \eqref{eq:quantum-commutator-constraint} also holds for $\Psi^i = \Psi^j = T$,
where we find that the expectation value of the commutator is proportional to
$\delta'''(s - s')$.
This is a well-known result in the \textsc{voa} formalism,
where the constant of proportionality is known as the \emph{central charge},
and is a vital parameter in the characterisation of a 2\textsc{dcft}.

As before, none of this tells us anything we did not already know.
Instead, it is the style of argument we wish to emphasise.
There is a generalisation of the notion of a distribution scaling homogeneously with degree $\mu$.
The \emph{Steinmann scaling degree} of a distribution $u \in \mathfrak{D}'(U)$,
where $U \subseteq \mathbb{R}^n$ is an open subset such that $\lambda x \in U \forall x \in U, \lambda > 0$,
is defined as
\begin{align}
  \mathrm{sd}(u) = \inf \left\{ \delta \in \mathbb{R} \,|\, \lim_{\lambda \rightarrow 0} \lambda^{\delta} u(\lambda x) = 0 \right\}.
\end{align}
Note that if $u$ scales homogeneously with degree $\mu$, then $\mathrm{sd}(u) = \mu$,
and if $u$ is a regular distribution, then by Taylor expanding $u$ around $x = 0$ one finds that
$\mathrm{sd}(u) \le 0$.
Whilst we found in our two examples that the highest order term in $\hbar$
(equivalently the term which did not vanish when evaluating at $\psi = 0$)
was homogeneously scaling with degree given by the degrees of our fields
each term in the expansion of the \textsc{opes} has a well-defined scaling degree
which counts the power of $(s - s')$ appearing in the denominator.

This suggest that in the future it may be possible to decompose the product of two fields
into a series of bi-distributions with decreasing scaling degrees.
Truncating this series at $\mathrm{sd} = 0$ would then give the usual form of the \textsc{ope}
modulo smooth terms.

\subsection{Reconstructing the boundary term}
\label{sec:boundary-term}

To close out this section,
we shall briefly comment on how the constants of proportionality
that appeared in our arguments from homogeneous scaling can be fixed by constraints on wavefront sets.

\begin{proposition}
    Let $u_n \in \mathfrak{D}'(\mathbb{R} \setminus \{0\})$ be
    the regular distribution defined by the function $\tfrac{1}{x^{n+1}}$.
    If $\overline{u}_n \in \mathfrak{D}'(\mathbb{R})^{\mathbb{C}}$ is
    an extension of $u_n$ which also scales homogeneously with degree $n + 1$, then
    \begin{equation}
        \overline{u}_n(x)
        =
        \frac{(-1)^n}{n!} \frac{d^n}{dx^n}
        \left[ \alpha \delta(x) + \mathrm{PV}\left( \frac{1}{x} \right) \right]
    \end{equation}
    for some $\alpha \in \mathbb{C}$, where $\mathrm{PV}$ denotes the Cauchy principle value defined by
    \begin{align*}
      \left\langle \mathrm{PV}\left( \tfrac{1}{x} \right), f \right\rangle
      =
      \lim_{\epsilon \searrow 0} \int_{\mathbb{R} \setminus (-\epsilon, \epsilon)} \frac{f(x)}{x} \,\mathrm{d}x.
    \end{align*}

    Moreover, the wavefront set of $\overline{u}_n$ is
    \begin{equation}
      \label{eq:extension-wavefront-sets}
      \mathrm{WF}(\overline{u}_n)
      =
      \left\{
    \begin{array}{l@{\quad:\quad}l}
    \{0\} \times \mathbb{R}_{<0} & \alpha = i\pi\\
    \{0\} \times \mathbb{R}_{>0} & \alpha = -i\pi\\
    \{0\} \times \mathbb{R} \setminus \{0\} & \text{else}
    \end{array}\right.
    \end{equation}
    and for the values $\alpha = \pm i\pi$, we may write $\overline{u}_n$ as
    \begin{equation}
        \overline{u}_n(x)
        =
        \frac{(-1)^n}{n!} \frac{d^n}{dx^n}
        \lim_{\epsilon \searrow 0} \frac{1}{x \mp i\epsilon}.
    \end{equation}
\end{proposition}

\begin{proof}
  Firstly, note that
  $\frac{(-1)^n}{n!} \frac{d^n}{dx^n} \mathrm{PV}\left( \frac{1}{x} \right)$
  is a well defined distribution on all of $\mathbb{R}$ which
  coincides with $u_n$ on the complement of $\{0\}$.
  It also scales homogeneously with degree $n+1$,
  hence the difference between this distribution and any other extension of $u_n$
  must be proportional to $\delta^{(n)}(x)$.

  As $\bar{u}_n$ is a tempered distribution, we can compute its wavefront set by simply taking its Fourier transform,
  which is
  \begin{align}
    \hat{\overline{u}}_n(\xi) = \frac{(-i \xi)^n}{n!} \left[ \alpha - i \pi \, \mathrm{sgn}(\xi) \right].
  \end{align}
  From this we can clearly see which values of $\alpha$ correspond to which wavefront set in~\eqref{eq:extension-wavefront-sets}.
\end{proof}

Similarly to our embedding $\partial_{\Sigma_0, \epsilon}^{*}: \mathfrak{F}_c(\Sigma_0) \rightarrow \mathfrak{F}_{\mu c}(\mathbb{M}^2)$,
we can use $\partial_{\Sigma_0, \epsilon}^{*}$ to map $\overline{u}_n(s - s') \in \mathfrak{D}'(\Sigma_0^2)$
to a distribution $(\partial_{\Sigma_0, \epsilon}^{*})^{\otimes 2} \overline{u}_n \in \mathfrak{D}'((\mathbb{M}^2)^2)$.
Again, we can compute the wavefront set of this new distribution by considering its Fourier transform
as a Schwartz distribution, which is
\begin{align}
  \widehat{(\partial_{\Sigma_0, \epsilon}^{*})^{\otimes 2} \overline{u}_n}(\xi, \eta)
  =
  -4 \xi_u \eta_u \hat{\delta}_{\epsilon}(2 \xi_v) \hat{\delta}_{\epsilon}(2 \eta_v)
  \int_{\mathbb{R}^{2}} \overline{u}_n(s - s')
  e^{i(\xi_u - \xi_v)s} e^{i(\eta_u - \eta_v)s} \, \mathrm{d}s \, \mathrm{d}s'.
\end{align}
The presence of the $\hat{\delta}_{\epsilon}$ terms means that this function decays rapidly for
all directions in which $\xi_v, \eta_v$ are non-zero.
Restricting to the remaining directions, we then compute the Fourier transform as
\begin{align}
  \widehat{(\partial_{\Sigma_0, \epsilon}^{*})^{\otimes 2} \overline{u}_n}(\xi, \eta)
  &=
  -4 \xi_u \eta_u \delta(\xi_u + \eta_u)\widehat{\overline{u}}(\xi_u) \nonumber \\
  &=
  \frac{(-i)^{n} \xi_u^{n + 2}}{n!}[\alpha - i \pi \mathrm{sgn}(\xi_u)].
\end{align}
From this, we may define a coordinate-free condition to uniquely determine $\alpha$,
namely
\begin{align}
  \mathrm{WF}((\partial_{\Sigma_0, \epsilon}^{*})^{\otimes 2} \overline{u}_n)
  \subseteq \mathrm{WF}(W)
\end{align}
where $W$ is a Hadamard distribution, if and only if $\alpha = - i \pi$.

Going the other way, the Schwartz kernel $K \in \mathfrak{D}'(\Sigma_0 \times \mathbb{M}^2)$
associated to $\partial_{\Sigma, \epsilon}$ allows us to define a correspondence between
subsets of $\dot{T}^{*} \mathbb{M}^2$ and subsets of $\dot{T}^{*} \Sigma_0$
\begin{align}
  \begin{split}
    \Gamma \mapsto
    \big\{ (s, \xi_s) &\in \dot{T}^{*} \Sigma_0
      \,|\, \\
    &\exists (u, v; \xi_u, \xi_v) \in \Gamma, (s, u, v; \xi_s, \xi_u, \xi_v,) \in \mathrm{WF}(K) \big\}
  \end{split}
\end{align}
In two dimensions, the wavefront set of a Hadamard distribution factorises into two distinct parts
\begin{align}
  \mathrm{WF}(W) = \Gamma_{\ell} \sqcup \Gamma_r,
\end{align}
where, in $\mathbb{M}^2$ we can write $\Gamma_{\ell} = \left\{ (u, v, u, v'; \xi_u, 0, -\xi_u, 0) \,|\, \xi_u > 0 \right\}$
and a similar expression for $\Gamma_r$.
Under the correspondence associated to $K$, $\Gamma_{\ell}$ leads to a spectral condition on $\Sigma_0$
\begin{align}
  \mathrm{WF}(u) = \left\{ (s, s; \xi_s, -\xi_s) \in \dot{T}^{*} \Sigma_0^2 \,|\, \xi_s > 0 \right\}.
\end{align}
It is this wavefront set condition which uniquely determines the constant of proportionality $\alpha$.

It is interesting to note that the wavefront set spectral condition may be interpreted as a
`positivity of energy' condition which is valid even on curved spacetimes.
What we have seen is that the remnant of this condition for states in the chiral algebra
is enough to uniquely determine the extension of homogeneously scaling distributions on $\mathfrak{D}'(\Sigma_0^2 \setminus \Delta)$.
One might argue that this is a somewhat artificial scenario.
However, in the formulation of \textsc{qft}s using \emph{factorisation algebras} \cite{costelloFactorizationAlgebrasQuantum2016},
the product of observables (known as the \emph{factorisation product}) is defined only for
observables localised to disjoint regions,
hence one would expect such distributions to arise naturally in this setting.

\section[Model Independent Chiral Algebras]{Towards a Model Independent Definition of Chiral Algebras}

To conclude, we shall discuss how one might in general define the `chiral sector' of a generic \textsc{aqft}.
For an on-shell algebra, we are able to provide a precise definition.
We shall see that our classical algebra fits this definition,
and also that locality in the sense of commutativity of chiral algebras localised in disjoint regions
is implied by Einstein causality in the full algebra.

\subsection{Definition}
\label{sec:model-indep-definition}

We need a category whose objects are connected sets of null geodesics of spacetimes.
It is similar in definition and purpose to the category $\mathsf{CCauchy}$ introduced earlier.
However, this time we will be able to give a simpler definition by realising this structure as a \emph{comma category}.

\begin{definition}
  Let $\mathcal{A} \overset{S}{\longrightarrow} \mathcal{B} \overset{T}{\longleftarrow} \mathcal{C}$ be a pair of functors with common target.
  The \emph{comma category} $(S \downarrow T)$ is the category such that
  \begin{itemize}
    \item Objects are triples $(a, h, c)$ such that $h \in \mathcal{B}(S(a), T(c))$.
    \item Morphisms $(a, h, c) \rightarrow (a', h', c')$ are pairs $(f, g) \in \mathcal{A}(a, a') \times \mathcal{C}(c, c')$ such that
          the following diagram commutes.
          \begin{equation}
            \begin{tikzcd}
              S(a)
                \ar[r, "Sf"]
                \ar[d, "h"]
                &
              S(a')
                \ar[d, "h'"]
                \\
              T(c)
                \ar[r, "Tg"]
                &
              T(c')
            \end{tikzcd}
          \end{equation}
  \end{itemize}
\end{definition}

Recall that in \cref{prop:null-projection-functorial} we established that the map $\pi_{\ell}: \mathcal{M} \rightarrow \mathcal{M}_{\ell}$
projecting a spacetime $\mathcal{M}$ onto its space of \emph{right-moving} null geodesics was `functorial'
in the sense that, for any $\chi \in \mathsf{CLoc}(\mathcal{M}, \widetilde{\mathcal{M}})$, there was a map $\chi_{\ell}$ such that
$\widetilde{\pi}_{\ell} \chi = \chi_{\ell} \pi_{\ell}$, where $\widetilde{\pi}_{\ell}$ is the corresponding projection on $\widetilde{\mathcal{M}}$.
In particular $\pi_{\ell}$ defines a functor $\mathsf{CLoc} \rightarrow \mathsf{Man}_1^+$ where
the latter category comprises connected, oriented $1$-manifolds with orientation-preserving smooth embeddings as morphisms.

One of the categories we shall use to define chiral algebras in a more generic setting is the comma category
$(\mathrm{Id}_{\mathsf{Man}_1^+} \downarrow \pi_{\ell})$ where $\mathrm{Id}_{\mathsf{Man}_1^+}$ denotes the identity functor.
Unpacking the definition, this means that an object of this category is defined by a choice $\mathcal{I} \in \mathsf{Man}_1^+$,
a spacetime $\mathcal{M} \in \mathsf{CLoc}$ and an embedding $i: \mathcal{I} \hookrightarrow \pi_{\ell}(\mathcal{M})$.
Similarly, a morphism is a pair, $\mathcal{I} \overset{\rho}{\longrightarrow} \mathcal{J}$,
$\mathcal{M} \overset{\chi}{\longrightarrow} \mathcal{N}$,
of a smooth embedding and admissible embedding respectively
such that the following diagram commutes.
\begin{equation}
  \label{eq:comma-cat-square}
  \begin{tikzcd}
    \mathcal{I}
      \ar[r, "i"] \ar[d, "\rho"]
      &
    \pi_{\ell}\left( \mathcal{M} \right)
      \ar[d, "\chi"]
    \\
    \mathcal{J}
      \ar[r, "j"]
      &
    \pi_{\ell}\left( \mathcal{N} \right)
  \end{tikzcd}
\end{equation}

\begin{definition}
  \label{def:model-indep-chiral-sub}
  Let $\mathfrak{A}: \mathsf{CLoc} \rightarrow \mathsf{Obs}$ be an \textsc{aqft} with values in some suitable category $\mathsf{Obs}$.
  A \emph{chiral subalgebra} of $\mathfrak{A}$ is a functor $\mathfrak{A}_c: \mathsf{Man}_1^+ \rightarrow *\text{-}\mathsf{Alg}$
  and a functor $N: (\mathrm{Id}_{\mathsf{Man}_1^+} \downarrow \pi_{\ell}) \rightarrow (\mathfrak{A}_c \downarrow \mathfrak{A})$.

  Explicitly, this means that for every square of the form \eqref{eq:comma-cat-square},
  there is a commuting diagram in $\mathsf{Obs}$ as follows.
  \begin{equation}
    \begin{tikzcd}
      \mathfrak{A}_c(\mathcal{I})
      \ar[d, "\mathfrak{A}_c \rho"] \ar[r, "N i"]
      &
      \mathfrak{A}(\mathcal{M})
      \ar[d, "\mathfrak{A} \chi"]
      \\
      \mathfrak{A}_c(\mathcal{J})
      \ar[r, "N j"]
      &
      \mathfrak{A}(\mathcal{N})
    \end{tikzcd}
  \end{equation}
\end{definition}

\begin{remark}
  There is a degree of redundancy in this definition.
  If we asked only for a functor
  $(\mathrm{Id}_{\mathsf{Man}_1^+} \downarrow \pi_{\ell}) \rightarrow (\mathrm{Id}_{\mathsf{Obs}} \downarrow \mathfrak{A})$,
  then this would automatically contain all the information necessary to describe $\mathfrak{A}_c$.

  For a comma category $(S \downarrow T)$ we define
  the projection functors $\Pi_{\mathcal{A}}: (S \downarrow T) \rightarrow \mathcal{A}$, $\Pi_{\mathcal{C}}: (S \downarrow T) \rightarrow \mathcal{C}$.
  Given a functor $N: (\mathrm{Id}_{\mathsf{Man}_1^+} \downarrow \pi_{\ell}) \rightarrow (\mathrm{Id}_{\mathsf{Obs}} \downarrow \mathfrak{A})$,
  we can then define $\mathfrak{A}_c := \Pi_{\mathsf{Obs}} \circ N$.

  There, is another equivalent characterisation of this definition.
  Given a diagram
  
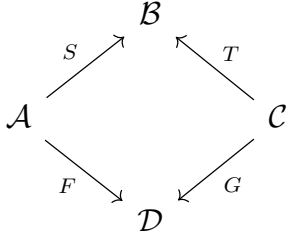
\begin{equation*}
    \begin{tikzcd}
      & \mathcal{B} & \\
      \mathcal{A} \ar[ur, "S"] \ar[dr, swap, "F"] & &
      \mathcal{C} \ar[ul, swap, "T"] \ar[dl, "G"] \\
      & \mathcal{D} &
    \end{tikzcd}
  \end{equation*}
  there is a one-to-one correspondence between natural transformations
  $F \circ \Pi_{\mathcal{A}} \Rightarrow G \circ \Pi_{\mathcal{C}}$ and functors
  $(S \downarrow T) \rightarrow (F \downarrow G)$.
  Expanding out the above definition in terms of natural transformations yields
  a set of conditions a lot closer to our original formulation in terms of the category $\mathsf{CCauchy}$.

  Secondly, this definition always admits a trivial subalgebra by taking
  $\mathfrak{A}_c(\mathcal{I}) = 0$ for every $\mathcal{I}$.
  In principle, one might define a chiral subalgebra as \emph{maximal} if, in addition to the above,
  it possesses the universal property that,
  for every alternative choice $N': (\mathrm{Id}_{\mathsf{Man}_1^+} \downarrow \pi_{\ell}) \rightarrow (\mathfrak{A}'_c \downarrow \mathfrak{A})$
  there exists a unique functor $I: (\mathfrak{A}'_c \downarrow \mathfrak{A}) \rightarrow (\mathfrak{A}_c \downarrow \mathfrak{A})$ such that
  $N' = N \circ I$.
  However, we shall not explore this idea further.
\end{remark}

\subsection{Causality in Chiral Subalgebras}

We now make use of this definition by proving that causality in the chiral sense is guaranteed for any chiral subalgebra of an \textsc{aqft} satisfying Einstein causality.
To make this precise, we define each of these properties as follows.

\begin{definition}
  \label{def:Einstein-causality}
  An \textsc{aqft} $\mathfrak{A}: \mathsf{CLoc} \rightarrow *\text{-}\mathsf{Alg}$ satisfies the \emph{Einstein causality} axiom if,
  for every diagram of $\mathsf{CLoc}$ morphisms $\mathcal{M}_1 \overset{\chi_1}{\longrightarrow} \mathcal{M} \overset{\chi_2}{\longleftarrow} \mathcal{M}_2$
  such that $\chi_1(\mathcal{M}_1)$ is spacelike separated from $\chi_2(\mathcal{M}_2)$,
  \begin{align*}
    \left[
      \mathfrak{A}\chi_1 \left( \mathfrak{A}(\mathcal{M}_1) \right),
      \mathfrak{A}\chi_2 \left( \mathfrak{A}(\mathcal{M}_2) \right)
    \right]
    = 0.
  \end{align*}
\end{definition}

\begin{definition}
  A functor $\mathfrak{A}_c: \mathsf{Man}_1^+ \rightarrow *\text{-}\mathsf{Alg}$ obeys \emph{chiral causality} if
  for every $\mathcal{I}_1 \overset{i_1}{\longrightarrow} \mathcal{I} \overset{i_2}{\longleftarrow} \mathcal{I}_2$ such that
  $i_1(\mathcal{I}_1) \cap i_2(\mathcal{I}_2) = \emptyset$,
  \begin{align*}
    \left[
      \mathfrak{A}_c i_1 (\mathfrak{A}_c(\mathcal{I}_1)),
      \mathfrak{A}_c i_2 (\mathfrak{A}_c(\mathcal{I}_2))
    \right]
    = 0.
  \end{align*}
\end{definition}

\begin{theorem}
  \label{thm:chiral-causality}
  Let $\mathfrak{A}: \mathsf{CLoc} \rightarrow *\text{-}\mathsf{Alg}$ be an \textsc{aqft} with a chiral subalgebra
  $N: (\mathrm{Id}_{\mathsf{Man}_1^+} \downarrow \pi_{\ell}) \rightarrow (\mathfrak{A}_c \downarrow \mathfrak{A})$,
  then $\mathfrak{A}_c$ obeys chiral causality if $\mathfrak{A}$ satisfies the Einstein causality axiom.
\end{theorem}

\begin{proof}
  Given $\mathcal{I}_1 \overset{i_1}{\longrightarrow} \mathcal{I} \overset{i_2}{\longleftarrow} \mathcal{I}_2$ as above,
  we shall construct a diagram $\mathcal{M}_1 \overset{\chi_1}{\longrightarrow} \mathcal{M} \overset{\chi_2}{\longleftarrow} \mathcal{M}_2$
  satisfying the conditions of~\cref{def:Einstein-causality} along with maps
  $\rho: \mathcal{I} \rightarrow \pi_{\ell}(\mathcal{M})$ and $\rho_i: \mathcal{I} \rightarrow \pi_{\ell}(\mathcal{M}_i)$ such that
  the following commutes.
  \begin{equation}
    \begin{tikzcd}
      \mathcal{I}_1
        \ar[r, "i_1"]
        \ar[d, "\rho_1"]
        &
      \mathcal{I}
        \ar[d, "\rho"]
        &
      \mathcal{I}_2
        \ar[l, swap, "i_2"]
        \ar[d, "\rho_2"]
        \\
      \pi_{\ell}(\mathcal{M}_1)
        \ar[r, swap, "\pi_{\ell} \chi_1"]
        &
      \pi_{\ell}(\mathcal{M})
        &
      \pi_{\ell}(\mathcal{M}_2)
        \ar[l, "\pi_{\ell} \chi_2"]
    \end{tikzcd}
  \end{equation}
  We first construct $\mathcal{M}$ by choosing an arbitrary metric on
  $\mathcal{I}$ compatible with the pre-existing orientation.
  $\mathcal{M}$ is then given by the manifold $\mathcal{I} \times \mathbb{R}$
  with metric $g_{\mathcal{M}} = du \odot dv$, where $du$ is the volume form on $\mathcal{I}$ induced by the metric,
  and $dv$ is the canonical one-form on $\mathbb{R}$.
  By construction, there is a canonical isomorphism $\rho: \pi_{\ell}(\mathcal{M}) \overset{\simeq}{\rightarrow} \mathcal{I}$
  given by $u \mapsto [ \left\{ (u, v) \right\}_{v \in \mathbb{R}} ]$.

  Next, we take an arbitrary Cauchy surface $\Sigma \subset \mathcal{M}$ and define
  $\Sigma_i = \Sigma \cap \pi^{-1}_{\ell}(\rho \circ i_i (\mathcal{I}_i))$ for $i \in \{1, 2\}$,
  where by $\pi^{-1}_{\ell}$ we mean in particular the preimage under the map
  $\pi_{\ell}: \mathcal{M} \rightarrow \rho(\mathcal{I})$.
  The requisite spacetime $\mathcal{M}_i$ may then be obtained as the \emph{Cauchy development} $D(\Sigma_i)$
  (with the embedding $\chi_i$ being simply the inclusion map).
  This is the set of events $x \in \mathcal{M}$ such that
  every inextensible causal curve intersecting $x$ also intersects $\Sigma_i$.
  A sketch of this situation is provided in~\cref{fig:cauchy-development}.

  Note that $\Sigma_1 \cap \Sigma_2 = \emptyset$ means that $D(\Sigma_1)$ and $D(\Sigma_2)$ must be causally disjoint:
  suppose we have $x_1 \in D(\Sigma_1)$, $x_2 \in D(\Sigma_2)$ such that there is a causal curve $\gamma$ connecting the two.
  If we maximally extend $\gamma$, then it must intersect both $\Sigma_1$ and $\Sigma_2$ by definition of the Cauchy developments.
  However, as $\Sigma$ is a Cauchy surface, each inextensible causal curve intersects it \emph{precisely} once,
  thus $x_1$ and $x_2$ cannot be connected by any causal curve.

  We can also see that $\pi_{\ell}(\mathcal{M}_i) = \pi_{\ell}(\Sigma_i) = \rho \circ i_i(\mathcal{I}_i)$.
  This means that if $\rho_i = \rho \circ i_i |^{\pi_{\ell}(\mathcal{M}_i)}$,
  then $\pi_{\ell} \chi_i \circ \rho_i = \rho \circ i_i$ as required, as $\pi_{\ell} \chi_i$ is simply the inclusion map of
  $\pi_{\ell}(\mathcal{M}_i) \subseteq \pi_{\ell}(\mathcal{M})$.

  We then apply the functor $N$ to the above diagram, which gives
  \begin{equation}
    \begin{tikzcd}
      \mathfrak{A}_c(\mathcal{I}_1)
        \ar[r, "\mathfrak{A}_c i_1"]
        \ar[d, "N \rho_1"]
        &
      \mathfrak{A}_c(\mathcal{I})
        \ar[d, "N \rho"]
        &
      \mathfrak{A}_c(\mathcal{I}_2)
        \ar[l, swap, "\mathfrak{A}_c i_2"]
        \ar[d, "N \rho_2"]
        \\
      \mathfrak{A}(\mathcal{M}_1)
        \ar[r, swap, "\mathfrak{A} \chi_1"]
        &
      \mathfrak{A}(\mathcal{M})
        &
      \mathfrak{A}(\mathcal{M}_2)
        \ar[l, "\mathfrak{A} \chi_2"]
    \end{tikzcd}
  \end{equation}

  Finally, from Einstein causality, we see that,
  for $F \in \mathfrak{A}_c(\mathcal{I}_1)$, $G \in \mathfrak{A}_c(\mathcal{I}_2)$,
  \begin{align*}
    N\rho ([\mathfrak{A}_c i_1 (F), \mathfrak{A}_c i_2 (G)]_{\mathfrak{A}_c(\mathcal{I})})
    &= [\mathfrak{A} \chi_1 \circ N \rho_1 (F), \mathfrak{A} \chi_2 \circ N \rho_2 (G)]_{\mathfrak{A}(\mathcal{M})} \\
    &= 0.
  \end{align*}
  Given that the morphisms in $*\text{-}\mathsf{Alg}$ are defined to be injective,
  we may also conclude that $[\mathfrak{A}_c i_1 (F), \mathfrak{A}_c i_2 (G)]_{\mathfrak{A}_c(\mathcal{I})} = 0$
  as required.
\end{proof}

\begin{remark}
  This definition does not actually depend on the conformal symmetry of $\mathfrak{A}$.
  In fact, a similar argument may be used in arbitrary dimensions to show that
  the canonical algebra associated to a Cauchy surface $\Sigma$,
  defined as the limit of the inverse system of neighbourhoods of $\Sigma$,
  is local in much the same way, provided such limits exist and can be assigned functorially.

  It is also worth noting that the analogous proof for $\mathfrak{A}_{\ell}: \mathsf{CCauchy} \rightarrow \mathsf{Obs}$
  is even simpler, as the necessary hypothesis amounts to the existence of a diagram of the form
  \begin{equation}
    \begin{tikzcd}
      \Sigma_1
        \ar[r, "\rho_1"] \ar[d, hookrightarrow]
        &
      \Sigma
        \ar[d, hookrightarrow]
        &
      \Sigma_2
        \ar[l, "\rho_2", swap] \ar[d, hookrightarrow]
        \\
      \mathcal{M}_1
        \ar[r, "\chi_1"]
        &
      \mathcal{M}
        &
      \mathcal{M}_2
        \ar[l, "\chi_2", swap]
    \end{tikzcd}
  \end{equation}
  such that $\rho_1(\Sigma_1) \cap \rho_2(\Sigma_2) = \emptyset$,
  which already implies that $\chi_1(\mathcal{M}_{1})$ is causally disjoint from $\chi_2(\mathcal{M}_{2})$.
\end{remark}

\begin{figure}
  \centering
  \begin{tikzpicture}[scale=2]
    \draw (0,0) .. controls (1,1) and (3,-1)  .. (4,0);
    \node[anchor=west] at (4,0) {$\Sigma$};
    \draw (0.5,1) -- (2.5,-1);
    \draw (0,0.5) -- (1.5,-1);
    \node (1) at (0.25, -0.6) {$\pi_{\ell}^{-1}(\rho \circ i_1(\mathcal{I}_1))$};
    \draw[->] (1) .. controls (0.9,-0.9) .. (1.5, -.6);
    \draw[dashed] (0.3,0.2) -- (0.8,0.7);
    \draw[dashed] (0.8,-0.3) -- (1.3,0.2);
    \node (2) at (1.4, 0.6) {$D(\Sigma_1)$};
    \draw[->] (2) .. controls (1, 0.6).. (0.8, 0.4);
    \draw (1.5,1) -- (3.5,-1);
    \draw (2.5,1) -- (4,-0.5);
    \node (3) at (3.45, 0.85) {$\pi_{\ell}^{-1}(\rho \circ i_2(\mathcal{I}_2))$};
    \draw[->] (3) .. controls (3.2,0.6) .. (2.5, .6);
    \draw[dashed] (2.7,-0.2) -- (3.2, 0.3);
    \draw[dashed] (3.2,-0.7) -- (3.7,-0.2);
    \node (4) at (2.5,-0.5) {$D(\Sigma_2)$};
    \draw[->] (4) .. controls (3, -0.5) .. (3.2, -0.35);
  \end{tikzpicture}
  \caption{A sketch of $\Sigma_1$, $\Sigma_2$ and their Cauchy developments.}
  \label{fig:cauchy-development}
\end{figure}
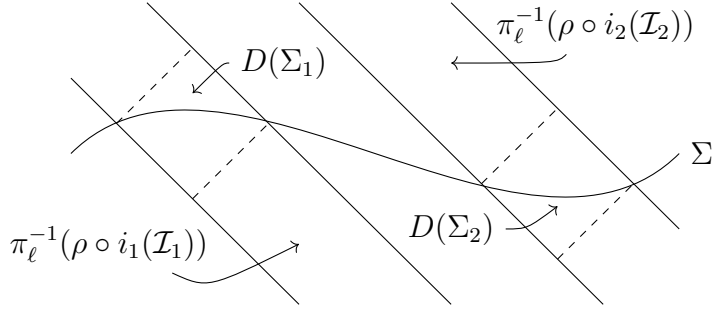

\subsection{Comparison to Earlier Definition}

From our discussion in \Cref{sec:model-indep-definition},
we can generalise the relationship between $\mathfrak{A}_{\ell}$ and $\mathfrak{A}_{\mathrm{on}}$
by introducing a notion of \emph{chiral subalgebras on Cauchy surfaces}.
For a full algebra $\mathfrak{A}: \mathsf{CLoc} \Rightarrow \mathsf{Obs}$,
this is given by a pair comprising a functor
$\mathfrak{A}_{\ell}: \mathsf{CCauchy} \rightarrow \mathsf{Obs}$ and a natural transformation
$\eta: \mathfrak{A}_{\ell} \Rightarrow \mathfrak{A} \circ \Pi_2$,
recalling $\Pi_2: \mathsf{CCauchy} \rightarrow \mathsf{CLoc}$ is the obvious projection functor,
and that $\mathfrak{A}$ should be an on-shell algebra, for reasons discussed earlier.

In this section we address the question of how this data can be used to generate
a chiral subalgebra in the sense of \cref{def:model-indep-chiral-sub}.

On $\mathcal{M}$ and $\Sigma$, our algebraic structures were determined by
distributions, or families of distributions,
which in turn defined bilinear operators on a suitable class of functionals.
In each case, the pairing between distributions and functional derivatives
was aided by the existence of a preferred volume form induced by the metric.
As we now wish to construct an algebra on a manifold with no preferred volume form,
a degree of care is required when defining these pairings.

For both the classical and quantum algebras,
the algebraic structures we need to define are determined
by bi-distributions, i.e. elements of $\mathfrak{D}'(\mathcal{I}^2)$ for
$\mathcal{I} \in \mathsf{Man}_1^+$.
In each case, candidate distributions may be obtained by
choosing an oriented diffeomorphism
$\rho: \mathcal{I} \overset{\sim}{\rightarrow} \Sigma \subset \mathcal{M}$.
We have already seen how $\partial_{\Sigma}$ can be used to define a map
$\mathfrak{D}'_{\Gamma_2}(\mathcal{M}^2) \rightarrow \mathfrak{D}'(\Sigma^2)$.

Without a preferred choice of volume form,
one must take care in the identification of functional derivatives and distributions.
For example, if we take for our configuration space $\Omega^p(X)$,
the space of $p$-forms on some $d$-dimensional, orientable manifold $X$ then,
for a Bastiani smooth map $F: \Omega^{p}(X) \rightarrow \mathbb{C}$,
the functional derivatives $F^{(n)}[\phi]$ (after completion) belong to
$\mathfrak{D}'(X^n, (\bigwedge^p T^{*}X)^{\boxtimes n})$.
This means that the derivatives of \emph{regular} functionals
should belong to $\mathfrak{D}(X^n, (\bigwedge^{d - p} T^{*}X)^{\boxtimes n})$.
For a spacetime $\mathcal{M}$ and configuration space $\Omega^p(\mathcal{M})$,
the second Euler-Lagrange derivative of an action
belongs to $\mathfrak{D}'(\mathcal{M}^2; (\bigwedge^p T^{*} \mathcal{M}^2)^{\boxtimes 2})$.
By the Schwartz kernel theorem, this then induces, for every $\phi \in \Omega^p(\mathcal{M})$
a linear operator $P_{\phi}: \Omega^p_c(\mathcal{M}) \rightarrow \mathfrak{D}'(\mathcal{M}, \bigwedge^p T^{*} \mathcal{M})$.
For physically reasonable actions, $P_{\phi}$ is a Green hyperbolic differential operator,
hence its image is $\Omega^{d-p}_c(\mathcal{M})$.
The advanced and retarded propagators are then maps
$\Omega^{d-p}_c(\mathcal{M}) \rightarrow \Omega^p(\mathcal{M})$ which, by Schwartz kernel theorem,
may then be identified with an element of $\mathfrak{D}'(\mathcal{M}^2, (\bigwedge^{d - p} T^{*} \mathcal{M})^{\boxtimes 2})$.
Hence, for the Pauli-Jordan function of the massless scalar field,
we have $E \in \mathfrak{D}'(\mathcal{M}^2, (\bigwedge^2 T^{*} \mathcal{M})^{\boxtimes 2})$.

Physically, $\partial_{\Sigma} \phi$ should be interpreted as the chiral component
of the conserved current of $\phi$, i.e. it is more naturally a one-form.
Hence, the corresponding variant of the diagram \eqref{prop:config-covariance} is
\begin{equation}
    \begin{tikzcd}
    \Omega^1(\Sigma)
        &
    \Omega^1(\widetilde{\Sigma})
        \ar[l, "\rho^{*}"]
        \\
    \mathfrak{E}(\mathcal{M})
    \ar[u, "\partial_{\Sigma}"]
    &
    \mathfrak{E}(\widetilde{\mathcal{M}})
    \ar[u, "\partial_{\widetilde{\Sigma}}"]
    \ar[l, "\chi^{*}"]
    \end{tikzcd}
\end{equation}
where $\partial_{\Sigma}$ is defined as in \eqref{eq:del-sigma-minkowski} but with the Hodge star omitted.
The corresponding extension to suitable distributions is then
$\partial_{\Sigma}^{\otimes2}: \mathfrak{D}'_{\Gamma_2}(\mathcal{M}^2, (\bigwedge^2 T^{*} \mathcal{M})^{\boxtimes 2}) \rightarrow \mathfrak{D}'(\Sigma^2)$,
where $\Gamma$ is the complement of the conormal bundle corresponding to the embedding $\Sigma^2 \hookrightarrow \mathcal{M}^2$.
Pullback along $\rho$ then defines a map
$\mathfrak{D}'(\Sigma^2) \rightarrow \mathfrak{D}'(\mathcal{I}^2)$.
Putting these two maps together means that $(\rho^{*} \partial_{\Sigma})^{\otimes 2} E_{\mathcal{M}} \in \mathfrak{D}'(\mathcal{I}^2)$

To show that $(\rho^{*}\partial_{\Sigma})^{\otimes 2} E_{\mathcal{M}}$ is
independent of the choice $\mathcal{I} \overset{\sim}{\rightarrow} \Sigma \subset \mathcal{M}$,
note that if we take a pair of such maps $\rho_i: \mathcal{I} \overset{\sim}{\rightarrow} \Sigma_i \subset \mathcal{M}_i$,
then we can find open, causally convex neighbourhoods, $\mathcal{N}_i$, of $\Sigma_i$ in $\mathcal{M}_i$
such that we can construct a commutative diagram
\begin{equation}
  \begin{tikzcd}
    &
    \mathcal{I}
      \ar[dl, "\sim", "\rho_{1}"']
      \ar[dr, "\sim"', "\rho_{2}"]
    & \\
    \Sigma_1
      \ar[rr, "\sim", "\rho"']
      \ar[d, hookrightarrow]
    &&
    \Sigma_2
      \ar[d, hookrightarrow]
    \\
    \mathcal{N}_1
      \ar[d, hookrightarrow, "i_1"]
      \ar[rr, "\sim", "\chi"']
    &&
    \mathcal{N}_2
      \ar[d, hookrightarrow, "i_2"]
    \\
    \mathcal{M}_1
    &&
    \mathcal{M}_2
  \end{tikzcd}
\end{equation}

Dualising, we then obtain a diagram for distributions
\begin{equation}
  \label{eq:sections-everywhere}
  \begin{tikzcd}[column sep=1em]
    &
    \mathfrak{D}'(\mathcal{I}^2)
      \ar[<-, dl, "\sim", "\rho_{1}^{* \otimes 2}"']
      \ar[<-,dr, "\sim"', "\rho_{2}^{* \otimes 2}"]
    & \\
    \mathfrak{D}'(\Sigma_1^2)
      \ar[<-, rr, "\sim", "\rho^{* \otimes 2}"']
      \ar[<-, d, "\partial_{\Sigma_1}^{\otimes 2}"]
    &&
    \mathfrak{D}'(\Sigma_2^2)
      \ar[<-, d, "\partial_{\Sigma_2}^{\otimes 2}"]
    \\
    \mathfrak{D}'_{\Gamma_2}(\mathcal{N}_1^2, (\bigwedge^2 T^{*} \mathcal{N}_1)^{\boxtimes 2})
      \ar[<-, d, "i_1^{* \otimes 2}"]
      \ar[<-, rr, "\sim", "\chi^{* \otimes 2}"']
    &&
    \mathfrak{D}'_{\Gamma_2}(\mathcal{N}_2^2, (\bigwedge^2 T^{*} \mathcal{N}_2)^{\boxtimes 2})
      \ar[<-, d, "i_2^{* \otimes 2}"]
    \\
    \mathfrak{D}'_{\Gamma_2}(\mathcal{M}_1^2, (\bigwedge^2 T^{*} \mathcal{M}_1)^{\boxtimes 2})
    &&
    \mathfrak{D}'_{\Gamma_2}(\mathcal{M}_2^2, (\bigwedge^2 T^{*} \mathcal{M}_2)^{\boxtimes 2})
  \end{tikzcd}
\end{equation}

From this diagram we can deduce that
\begin{align*}
  (\rho_1^{*} \partial_{\Sigma_1})^{\otimes 2} E_{\mathcal{M}_1}
  &= (\rho_1^{*} \partial_{\Sigma_1} i_{1}^{*})^{\otimes 2} E_{\mathcal{M}_1} \\
  &= (\rho_1^{*} \partial_{\Sigma_1})^{\otimes 2} E_{\mathcal{N}_1} \\
  &= (\rho_1^{*} \partial_{\Sigma_1})^{\otimes 2} E_{\mathcal{N}_1} \\
  &= (\rho_1^{*} \partial_{\Sigma_1} \chi^{*})^{\otimes 2} E_{\mathcal{N}_2} \\
  &= (\rho_1^{*} \rho^{*} \partial_{\Sigma_2})^{\otimes 2} E_{\mathcal{N}_2} \\
  &= (\rho_1^{*} \rho^{*} \partial_{\Sigma_2})^{\otimes 2} E_{\mathcal{M}_2} \\
  &= (\rho_2^{*} \partial_{\Sigma_2})^{\otimes 2} E_{\mathcal{M}_2}.
\end{align*}

By taking in particular some arbitrary choice
$\rho_0: I \overset{\sim}{\rightarrow} \Sigma_0 \subset \mathcal{M}_0 \in \left\{ \mathbb{M}^2, \mathscr{E} \right\}$,
we can deduce that
\begin{align}
  \left\langle (\rho_0^{*} \partial_{\Sigma_0})^{\otimes 2} E_{\mathcal{M}_0}, f \otimes g \right\rangle
  =
  \int_{\mathcal{I}} f \wedge dg
\end{align}

From this point, we define a Poisson algebra associated to $\mathcal{I}$
in almost the same way as $\mathfrak{P}_{\ell}(\Sigma, \mathcal{M})$.
The identification of $\Omega^1(\mathcal{I})$ with $\mathfrak{E}(\mathcal{I})$
induced by any choice of oriented chart allows the wavefront condition
from \cref{prop:chiral-poisson} to be imposed on Bastiani smooth maps
$\Omega^1(\mathcal{I}) \rightarrow \mathbb{C}$ unambiguously.
We denote the space of such maps satisfying this condition $\mathfrak{F}^1_c(\Sigma)$.
The bi-distribution $(\rho^{*} \partial_{\Sigma})^{\otimes 2} E_{\mathcal{M}}$,
then also defines the structure of a Poisson algebra on this space
in a manner directly analogous to $\partial_{\Sigma}^{\otimes 2} E_{\mathcal{M}}$ on $\mathfrak{F}_c(\Sigma)$.

Functoriality can then be established by a simple argument:
if we have an oriented embedding $\mathcal{I} \overset{f}{\rightarrow} \mathcal{J}$,
then any sequence $\rho_{\mathcal{J}}: \mathcal{J} \overset{\sim}{\rightarrow} \Sigma_{\mathcal{J}} \subset \mathcal{M}_{\mathcal{J}}$
yields a corresponding sequence
$\rho_{\mathcal{I}}: \mathcal{I} \overset{\sim}{\rightarrow} \Sigma_{\mathcal{I}} \subset \mathcal{M}_{\mathcal{I}}$
where $\rho_{\mathcal{I}} := \rho_{J} \circ f$, $\Sigma_{\mathcal{I}} := \rho_{\mathcal{I}}(\mathcal{I})$ and
$\mathcal{M}_{\mathcal{I}} := D(\Sigma_{\mathcal{I}})$.
By studying the resulting system of inclusions and embeddings, one can quickly verify that
$
    (\rho_{\mathcal{I}} \partial_{\Sigma_{\mathcal{I}}})^{\otimes 2} E_{\mathcal{M}_{\mathcal{I}}}
        =
    (f^{*} \rho_{\mathcal{J}} \partial_{\Sigma_{\mathcal{J}}})^{\otimes 2} E_{\mathcal{M}_{\mathcal{J}}}
$.

If we denote the resulting functor on $\mathsf{Man}_1^+$ by $\mathfrak{P}_c$,
then we can establish the following result.

\begin{proposition}
  There is a chiral subalgebra $(\mathfrak{P}_c, N)$ of $\mathfrak{P}$
  such that $N$ factors through $\partial^{*}: \mathfrak{P}_{\ell} \Rightarrow \mathfrak{P}_{\mathrm{on}}$
\end{proposition}

\begin{proof}

  %
  All we need to do is to define the functor
  $(\mathrm{Id}_{\mathsf{Man}_1^+} \downarrow \pi_{\ell}) \rightarrow (\mathfrak{P}_c \downarrow \mathfrak{P}_{\mathrm{on}})$.
  For $\mathcal{I} \overset{\rho}{\rightarrow} \mathcal{M}$, we take an arbitrary Cauchy surface $\Sigma \subset \mathcal{M}$,
  and a regularised chiral derivative $\partial_{\Sigma, \epsilon}: \mathfrak{E}(\mathcal{M}) \rightarrow \mathfrak{E}(\Sigma)$.
  Note that each $\Sigma$ is also an object in $\mathsf{Man}_1^+$, thus we can produce a sequence of maps
  \begin{equation}
    \begin{tikzcd}
      \mathfrak{P}_c(\mathcal{I}) \ar[r, "\mathfrak{P}_c (\pi_{\ell}|_{\Sigma} \circ \rho)"] &
      \mathfrak{P}_c(\Sigma) \ar[r] &
      \mathfrak{P}_{\ell}(\Sigma, \mathcal{M}) \ar[r, "\partial_{\Sigma, \epsilon}^{*}"] &
      \mathfrak{P}_{\mathrm{on}}(\mathcal{M}),
    \end{tikzcd}
  \end{equation}
  where the central map is defined by $F[j] \mapsto (\psi \mapsto F[\psi \mathrm{d}V_{\Sigma}])$.
  The composition is overall independent of our choices $\Sigma$ and $\partial_{\Sigma, \epsilon}$ if,
  for every other choice $\Sigma'$, $\partial_{\Sigma', \epsilon'}$, the diagram
  \begin{equation}
    \label{eq:independent-to-dependent}
    \begin{tikzcd}
        &
      \mathfrak{P}_c(\Sigma)
        \ar[r]
        \ar[dd]
        &
      \mathfrak{P}_\ell(\Sigma, \mathcal{M})
        \ar[rd]
        \ar[dd]
        & \\
      \mathfrak{P}_c(\mathcal{I})
        \ar[ur]
        \ar[dr]
        &&&
      \mathfrak{P}_{\mathrm{on}}(\mathcal{M})
        \\
        &
      \mathfrak{P}_c(\Sigma')
        \ar[r]
        &
      \mathfrak{P}_\ell(\Sigma', \mathcal{M})
        \ar[ur]
        &
    \end{tikzcd}
  \end{equation}
  commutes, where the vertical arrows are the respective homomorphisms corresponding to
  $(\pi_{\ell}|_{\Sigma'}^{-1} \circ \pi_{\ell}|_{\Sigma}): \Sigma \overset{\sim}{\rightarrow} \Sigma'$.
  This is readily verified:
  the left-hand triangle commutes by functoriality of $\mathfrak{P}_c$,
  the right-hand triangle commutes by \cref{thm:classical-natural}
  and the commutativity of the inner square is obvious once the definitions are unpacked.
\end{proof}

\begin{remark}
  To be more precise, when we say that $N$ `factors through' $\partial_{\epsilon}^{*}$,
  there is in fact yet another comma category corresponding to the diagram
  $( \mathrm{Id}_{\mathsf{Man}_1^+} \downarrow \pi_{\ell} ) \longrightarrow \mathsf{CLoc} \longleftarrow \mathsf{CCauchy}$,
  where the functors in each case forget all data other than the underlying spacetime.
  By composing $\mathfrak{P}_c$ and $\mathfrak{P}_{\ell}$ with the appropriate functors out of this new comma category,
  the horizontal arrows in \eqref{eq:independent-to-dependent} then form a natural transformation between these two which,
  as the same diagram shows, then defines a functor $N: (\mathrm{Id}_{\mathsf{Man}_1^+} \downarrow \pi_{\ell}) \rightarrow (\mathfrak{P}_c, \mathfrak{P}_{\mathrm{on}})$ independent of the auxiliary choices of Cauchy surfaces.
\end{remark}

The quantisation of this algebra yet again follows similar lines.
In this case, rather than taking a single distribution
$E_{\mathcal{I}} := (\rho^{*}\partial_{\Sigma})^{\otimes 2}E_{\mathcal{M}}$
which is defined unambiguously on $\mathcal{I}$,
we instead use the fact that the \emph{set} of distributions
$(\rho^{*}\partial_{\Sigma})^{\otimes 2}\mathrm{Had}({\mathcal{M}})$ is independent of the choice
$\rho: I \overset{\sim}{\rightarrow} \Sigma \subset \mathcal{M}$.
This follows by considering the diagram \eqref{eq:sections-everywhere}
combined with the non-trivial fact that,
if $\mathcal{N}$ is a neighbourhood of a Cauchy surface $\Sigma \subset \mathcal{M}$,
then $\chi^{*}: \mathrm{Had}(\mathcal{M}) \overset{\sim}{\rightarrow} \mathrm{Had}(\mathcal{N})$,
as one can reconstruct a Hadamard distribution $W$ from $\chi^{*} W$ by Cauchy evolution \cite{fullingSingularityStructureTwopoint1978}.
If we denote the resulting set by $\mathrm{Def}(\mathcal{I})$,
then we can define,
for each $H_{c} \in \mathrm{Def}(\mathcal{I})$ a $*$-algebra $\mathfrak{A}_c^{H_c}(\mathcal{I})$.
As one might expect, the underlying vector space is
$\mathfrak{F}_c^1(\mathcal{I})[[\hbar]]$
and the product is defined by
\begin{align}
  F \star G [j] :=
  \sum_{n = 0}^{\infty} \frac{\hbar^n}{n!}
  \left\langle
    (\tfrac{i}{2} E_{\mathcal{I}} + H_c)^{\otimes n},
    F^{(n)}[j] \otimes G^{(n)}[j]
  \right\rangle.
\end{align}
We then have isomorphisms $\mathfrak{A}_c^{H_c}(\mathcal{I}) \simeq \mathfrak{A}_c^{H_c'}(\mathcal{I})$
for exactly the same reasons as in \cref{eq:beta-H-Hp-relation}.
In fact, given $\rho: \mathcal{I} \overset{\sim}{\rightarrow} \Sigma \subset \mathcal{M}$,
one can show that the map $F[j] \mapsto F[\rho^{*} \psi \,\mathrm{dV}_{\Sigma}]$
yields a homomorphism
$\mathfrak{A}_c^{H_c}(\mathcal{I}) \simeq \mathfrak{A}_{\ell}^{H_{\Sigma}}(\Sigma, \mathcal{M})$
where $H_{\Sigma}$ and $H_c$ are the distributions induced on $\Sigma$ and $\mathcal{I}$ respectively by
some $H \in \mathrm{Had}(\mathcal{M})$.
From this, one can then once again establish that there is a functor
$N: (\mathrm{Id}_{\mathsf{Man}_1^+} \downarrow \pi_{\ell}) \rightarrow (\mathfrak{A}_c \downarrow \mathfrak{A}_{\mathrm{on}})$
so we indeed have a chiral subalgebra according to \cref{def:model-indep-chiral-sub}.

\section*{Conclusions and Outlook}

We have seen how the tools of {\smaller{p}}\textsc{aqft}
can be adapted in order to construct
a natural subalgebra of chiral observables for a 2\textsc{d}, conformally covariant field theory.
In future work we aim to use the approach presented here as the basis
of a systematic study of various models of 2\textsc{dcft},
with the goal of gaining a deeper understanding of the role of integrability in {\smaller{p}}\textsc{aqft}.
From there, one may be able to expand upon the work done in \cite{bahnsLocalNetsNeumann2017, bahnsQuantumSineGordon2018}
obtaining non-perturbative physics from perturbative constructions.

Our analysis also raises the possibility of constructing 2\textsc{d} {\smaller{p}}\textsc{aqft}s
out of chiral sectors.
However, it should be noted that this process is complicated by the fact that,
when a spacetime is not simply connected,
only the exact parts of the left/right-moving solutions are completely independent,
as demonstrated in \eqref{eq:cylinder-gen-soln}.
On the level of algebras, this interdependence is demonstrated by the fact that, on-shell,
there is an overlap of the linear chiral and anti-chiral observables,
$\partial_{\Sigma, \epsilon}^{*}(\mathfrak{D}(\Sigma)) \cap \overline{\partial}_{\Sigma, \epsilon}^{*}(\mathfrak{D}(\Sigma)) \simeq \mathbb{R}$,
corresponding to the constant functions on $\Sigma \simeq S^1$.
In the case of the massless scalar field, one can take a
fibred tensor product of the chiral and anti-chiral algebras,
which avoids double counting the common observables.
However, this construction only works due to the fact that such observables are central,
which is not guaranteed in a general model.

Another possibility presented by this work is that new connections may be found
between {\smaller{p}}\textsc{aqft} and \emph{factorisation algebras}.
The chiral subalgebra we have constructed can be trivially extended to a dg-algebra
which is concentrated in degree $0$,
reflecting the fact that the dynamics are trivial.
This would likely require a Lorentzian analogue of
the translation-invariant holomorphic factorisation algebras of \cite[Chapter 5]{costelloFactorizationAlgebrasQuantum2016}.
The dg perspective may also prove useful
in treating less trivial theories such as the principal chiral model or the Wess-Zumino-Witten model.

\bibliographystyle{alpha}
\bibliography{bibtex}

\end{document}